\documentclass[11pt]{article}

\usepackage{amsmath, amsfonts, amssymb, amsthm, graphicx, color, enumitem}
\usepackage{fullpage}
\usepackage{hyperref}

\newcommand{\HilH}{\mathcal{H}}
\newcommand{\realR}{\mathbb{R}}
\newcommand{\compC}{\mathbb{C}}
\newcommand{\intZ}{\mathbb{Z}}
\newcommand{\Prob}{\mathbb{P}}

\newcommand{\ie}{i.e.}
\newcommand{\eg}{e.g.}

\newcommand{\etal}{et al.}
\newcommand{\pdf}{p.d.f.}
\newcommand{\cdf}{c.d.f.}

\newcommand{\lHopital}{l'H\^{o}pital}
\newcommand{\LHopital}{L'H\^{o}pital}
\newcommand{\M}{\mathcal{M}}

\newcommand{\B}{\mathcal{B}}

\newcommand{\R}{\mathbb{R}}
\newcommand{\A}{\mathbf{A}}
\newcommand{\m}{\mathbf{m}}
\newcommand{\mm}{m}
\newcommand{\aaa}{\mathbf{a}}
\newcommand{\acc}{\mathbf{a}_c}
\newcommand{\ga}{\mathbf{\Gamma}}
\newcommand{\redge}{\mathbf{e}}
\newcommand{\ledge}{\tilde{\mathbf{e}}}
\newcommand{\gfn}{\mathbf{g}}
\newcommand{\Gfn}{\mathbf{G}}
\newcommand{\Hfn}{\mathbf{H}}

\newcommand{\bfGamma}{\mathbf{\Gamma}}

\newcommand{\hv}{\mathbf{v}}
\newcommand{\hw}{\mathbf{w}}

\newcommand{\error}{\omega_1}
\newcommand{\errort}{\omega_2}
\newcommand{\erf}{G}

\newcommand{\Andreief}{Andr\'{e}ief}

\newcommand{\FGUE}{F_0}
\newcommand{\FGOE}{F_1}
\newcommand{\Int}{I_n^T}
\newcommand{\Intx}{J_n^T}

\newcommand{\bfP}{\mathbf{P}}
\newcommand{\bfQ}{\mathbf{Q}}
\newcommand{\frakP}{\mathfrak{P}}
\newcommand{\frakQ}{\mathfrak{Q}}

\newcommand{\bfR}{\mathbf{p}}

\newcommand{\calP}{\mathcal{P}}
\newcommand{\calQ}{\mathcal{Q}}
\newcommand{\calR}{\mathcal{R}}

\newcommand{\calN}{\mathcal{N}}

\newcommand{\bfU}{\mathbf{U}}
\newcommand{\bfV}{\mathbf{V}}
\newcommand{\bfEo}{\mathbf{E}^{(1)}}
\newcommand{\bfEt}{\mathbf{E}^{(2)}}

\newcommand{\bfZ}{\mathbf{Z}}
\newcommand{\Ii}{\mathbf{I}}

\newcommand{\DongNew}[1]{{\color{blue}{#1}}}

\DeclareMathOperator{\Tr}{Tr}

\DeclareMathOperator{\diag}{diag}
\DeclareMathOperator{\TW}{0}

\DeclareMathOperator{\Ai}{Ai}
\DeclareMathOperator{\Airy}{Airy}

\newtheorem{lemma}{Lemma}[section]
\newtheorem{thm}{Theorem}[section]
\newtheorem{prop}{Proposition}[section]

\theoremstyle{definition}

\theoremstyle{remark}
\newtheorem{rmk}{Remark}[section]

\newcommand{\ddd}{d}
\newcommand{\Ka}{\tilde{K}}
\newcommand{\Ko}{K}
\newcommand{\cE}{\mathcal{E}}

\newcommand{\fE}{\mathfrak{E}}

\newcommand{\vf}{\mathbf{v}}
\newcommand{\Mk}{\mathbf{M}}
\newcommand{\Mkh}{\widehat{\mathbf{M}}}
\newcommand{\vfh}{\hat{\mathbf{v}}}
\newcommand{\mn}{\mathbf{A}}
\newcommand{\NN}{\mathbf{N}}

\title{On the largest eigenvalue of a Hermitian random matrix model with spiked external source II. Higher rank cases}
\author{Jinho Baik\thanks{Department of Mathematics, University of Michigan, Ann Arbor, MI, 48109, USA \newline
email: \texttt{baik@umich.edu}} \ and 
Dong Wang\thanks{Department of Mathematics, University of Michigan, Ann Arbor, MI, 48109, USA \newline
email: \texttt{dowang@umich.edu}}}

\date{\today}

\begin{document}

\maketitle


\begin{abstract} 
This is the second part of a study of the limiting distributions of the top eigenvalues of a Hermitian matrix model with spiked external source under a general external potential.
The case when the external source is of rank one was analyzed in an earlier paper. 
In the present paper we extend the analysis to the higher rank case. 
If all the eigenvalues of the external source are less than a critical value, 
the largest eigenvalue converges to the right end-point of the support of the equilibrium measure
as in the case when there is no external source. 
On the other hand, if an external source eigenvalue is larger than the critical value, 
then an eigenvalue is pulled off from the support of the equilibrium measure. 
This transition is continuous, and is universal, including the fluctuation laws, for convex potentials. 
For non-convex potentials, two types of discontinuous transitions are possible to occur generically. 
We evaluate the limiting distributions in each case for general potentials including those 
whose equilibrium measure have multiple intervals for their support. 
\end{abstract}

\maketitle


\section{Introduction} \label{section:introduction}

Let $V: \realR \to \realR$ be an analytic function such that $\frac{V(x)}{\sqrt{x^2+1}} \to +\infty$ as $|x| \to \pm \infty$.
Let $\A_n$ be an $n\times n$ Hermitian matrix and 
consider the probability density function (\pdf ) on the set $\mathcal{H}_n$ of $n\times n$ 
Hermitian 
matrices defined by
\begin{equation}\label{eq:pdf_of_external_source_model}
	p_n(M) = \frac{1}{Z_n} e^{-n \Tr(V(M)-\A_n M)}, \quad M \in \HilH_n.
\end{equation}
Here $Z_n$ is the normalization constant so that $\int_{\HilH_n} p_n(M) dM=1$ where $dM$ denotes the Lebesgue measure. 
The sequence of probability spaces $(\HilH_n, p_n)$, $n=1,2,\cdots$, is called 
a Hermitian matrix model with \underline{external source matrices} $A_n$, $n=1,2,\cdots$, 
and \underline{potential} $V$. 
Note that due to the unitary invariance of $dM$ and the presence of the trace in the exponent of~\eqref{eq:pdf_of_external_source_model}, the density $p_n(M)$ depends only on the eigenvalues of $\A_n$. Hence we may assume without loss of generality that $\A_n$ is a diagonal matrix.

The main focus of this paper is the special case when 
\begin{equation} \label{eq:defination_of_ex_source_A}
	\A_n= \diag (\aaa_1, \cdots, \aaa_\m, \underbrace{0,\cdots,0}_{n-\m})
\end{equation}
for \underline{all} $n \geq \m$ with fixed nonzero $\aaa_1, \cdots, \aaa_\m$. (We also consider the case when $\aaa_j$ depend on $n$. Indeed, in transitional cases we assume $\aaa_j$ varies in $n$ but converges to a fixed value as $n\to\infty$.)
In this case, $(\HilH_n, p_n)$  is called 
a {\em  Hermitian matrix model with spiked external source}  
({\em spiked  model} for short) of {\em rank $\m$}. 
The main interest of this paper is to study 
how the limiting location and the fluctuation of the top eigenvalue(s) of $M$ as $n\to\infty$ 
depend on the ``external source eigenvalues'' $\aaa_j$.




The case when $\m=1$ (rank one case) was studied in \cite{Baik-Wang10a} 
to which we refer the readers for the background and motivations of the spiked  models. 
See also \cite{Baik-Ben_Arous-Peche05}, \cite{Peche05}, \cite{Baik-Silverstein06}, \cite{Feral-Peche07}, \cite{McLaughlin07}, \cite{Wang08}, \cite{Capitaine-Donati_Martin-Feral09}, 
\cite{Adler-Delphine-vanMoerbeke09}, \cite{Nadakuditi-Silverstein10},  \cite{Benaych_Georges-Nadakuditi11}, \cite{Benaych_Georges-Guionnet-Maida11}, 
\cite{Mo11a}, \cite{Wang11}, \cite{Bloemendal-Virag11}, \cite{Tao11}, \cite{Bertola-Buckingham-Lee-Pierce11}, and \cite{Bertola-Buckingham-Lee-Pierce11a}.  


The spiked model of an arbitrary fixed rank was studied in great detail for Gaussian potential in \cite{Peche05} and also for the so-called complex Wishart spiked models in \cite{Baik-Ben_Arous-Peche05}. 
Let us review the Gaussian case here. 
Let $\redge$ denote the right end-point of the support of the equilibrium measure associated to the potential $V$. 
For the Gaussian potential $V(x)= \frac12 x^2$,  $\redge=2$. 
Let $\xi_{\max}(n)$ denote the largest eigenvalue of the random matrix  of size $n$. 
Then there is a constant $\beta>0$ such that 
\begin{equation}
	\Prob_n((\xi_{\max}(n)-\redge) \beta n^{2/3} \le T; \aaa_1, \cdots, \aaa_\m) \to 
	\begin{cases} F_0(T), \quad 
	&\text{if $\max\{\aaa_1, \cdots, \aaa_\m\} < \frac12 V'(\redge),$} \\
	F_\m(T), \quad 
	&\text{if $\aaa_1= \cdots= \aaa_\m = \frac12 V'(\redge).$}
	\end{cases}
\end{equation}
On the other hand,  if $\aaa_1= \cdots= \aaa_\m  > \frac12 V'(\redge)$, there is a constant $x_0(\aaa_1)>\redge$ and $\gamma(\aaa_1)>0$ such that 
\begin{equation}
	\Prob_n((\xi_{\max}(n)-x_0(\aaa_1)) \gamma(\aaa_1) n^{1/2} \le T; \aaa_1, \cdots, \aaa_\m) \to 
	G_\m(T) 
\end{equation}
for each $T\in \R$, as $n\to\infty$.
Here $F_\m(T)$, $\m=0,1,2,\cdots$, and $G_\m(T)$, $\m=1,2\cdots$, 
are certain cumulative distribution functions. 
The constant $x_0(a)$ is a continuous function in $a\in (\frac12V'(\redge), \infty)$ 
and satisfies $x(a)\downarrow \redge$ as $a\downarrow \frac12 V'(\redge)$. 
Hence for the Gaussian potential, as $\aaa_j$'s increase, the limiting location $\xi_{\max}:= \lim_{n\to\infty} \xi_{\max}(n)$ stays at the right  end-point $\redge$ of the support of the equilibrium measure until $\aaa_j$'s reach the critical value $\frac12V'(\redge)$.
After the critical value, $\xi_{\max}$ break off from $\redge$  and moves to the right \emph{continuously}. 

This continuity of $\xi_{\max}$ does not necessarily hold for general potentials. 
Indeed, for the rank $1$ case  it was shown in \cite{Baik-Wang10a}  
that $\xi_{\max}$ may be a \emph{discontinuous} function of the (unique) external eigenvalue 
for certain potentials. (It was shown that $\xi_{\max}$ is continuous when $V(x)$ is convex in $x\ge \redge$. 
A criterion when the discontinuity occurs is given in \cite{Baik-Wang10a}.)
If $V$ is such a a potential, then 
$\xi_{\max}$ for the potentials $sV$ is also discontinuous for all $s$ close enough to $1$. 
An example of such $V$ can be constructed by considering a two-well potential 
with a deep well of the left and a shallow well on the right. 
Nevertheless there still is universality: it was shown in the rank $1$ case that the limiting distribution of $\xi_{\max}(n)$ at the continuous points of $\xi_{\max}$ is (generically) same as that of 
the Gaussian potential. (It though varies depending on whether 
the external eigenvalue is sub-critical, critical, and super-critical.)
Even more, at a discontinuous point, the limiting distribution is  something new but it is still (generically) independent of $V$. 
In this paper we show that similar universality also holds for the higher rank case.

\medskip

While \cite{Baik-Wang10a} and the current paper were being written,
a work on a similar subject was announced in the recent preprints by Bertola \etal\ \cite{Bertola-Buckingham-Lee-Pierce11} and \cite{Bertola-Buckingham-Lee-Pierce11a}. The major difference of their work and ours is that we take $\aaa_j$'s to be all distinct and keep $\m$ fixed, while 
\cite{Bertola-Buckingham-Lee-Pierce11, Bertola-Buckingham-Lee-Pierce11a}  take $\aaa_j$ to be identical and let  $\m\to\infty$ with $\m=o(n)$. 
Hence these two works complement each other. 
The methods are different and it seems that each has an unique advantage in handling the situations mentioned above; see Remark~\ref{rmk:distinctness} below.
See also Section 1.1 of \cite{Baik-Wang10a} for a further comparison. 

\subsection{Algebraic relation of the higher rank case and the rank $1$ case} \label{subsec:algebraic_theorem}

The starting point of analysis in this paper is a simple algebraic relation between 
the higher rank case and the rank $1$ case. The gap probability, for example, can be written as a finite determinant built out of the gap probabilities of rank $1$ cases.

In order to state this algebraic relation, we slightly generalize the setting of the spiked  model. 
Note that in the definition of the density~\eqref{eq:pdf_of_external_source_model}, the factor $n$ in front of $\Tr(V(M)-\A_n M)$ equals the dimension of the matrices $M$ and $\A_n$. We may take the factor different from the dimension  and consider the following  \pdf\ : 
\begin{equation} \label{eq:generalized_pdf}
	p_{d,n}(M) := \frac{1}{Z_{d,n}} e^{-n\Tr(V(M)-\mathbf{A}_{d}M)}, \qquad M\in\mathcal{H}_{d},
\end{equation}
where
\begin{equation} \label{eq:A_d_in_generalized_pdf}
  \mathbf{A}_d= \diag(\aaa_1, \cdots, \aaa_\m, \underbrace{0, \cdots, 0}_{d-\m})
\end{equation}
for some nonzero numbers $\aaa_1, \cdots, \aaa_\m$.
Define, for a complex number $s$ and a measurable subset $E$ of $\R$,  
\begin{equation} \label{eq:defn_of_generating_func_generalized}
    \cE_{\ddd,n} (\aaa_1, \cdots, \aaa_\m ; E; s) 
    := \mathbb{E}\bigg[ \prod^{\ddd}_{j=1}  (1 - s\chi_E(\xi_j))\bigg]  
    = \int_{\HilH_{\ddd}} \prod^{\ddd}_{j=1} (1 - s\chi_E(\xi_j)) p_{d,n}(M)dM
\end{equation}
where $\xi_j$, $j=1, \cdots, d$ denote the eigenvalues of $M$. 
It is well known that (see \eg\ \cite{Tracy-Widom98})
\begin{equation}\label{eq:jthlarge}
  	\mathbb{P}^{(j)}_{d,n} (\aaa_1, \cdots, \aaa_\m ; E) = \sum^{j-1}_{i=0} \frac{(-1)^i}{i!} \frac{d^i}{d s^i} \bigg|_{s=1} \cE_{d,n} (\aaa_1, \cdots, \aaa_\m ; E ;s).
\end{equation}
is the probability that there are no more than $j-1$ eigenvalues in $E$. 
When $E=(x, \infty)$, this is precisely the cumulative  distribution function (\cdf ) of the $j$th largest eigenvalue. 
When $\mathbf{A}_{d}=0$ we denote~\eqref{eq:defn_of_generating_func_generalized} by $\mathcal{E}_{d,n}( E; s)$, and also define 
\begin{equation} \label{eq:defn_of_bar_E_a_dots_E_s_beginning}
  	\bar{\mathcal{E}}_{d,n} (\aaa_1, \cdots, \aaa_\m; E; s) := \frac{\mathcal{E}_{d,n}(\aaa_1, \cdots, \aaa_\m; E; s)}{\mathcal{E}_{d,n}( E; s)}.
\end{equation}

Let $p_j(x;n)$, $j=0,1,\cdots$,  be the orthonormal polynomials  with respect to the (varying) measure $e^{-nV(x)}dx$
and set $\psi_j(x;n):= p_j(x;n) e^{-\frac{n}2 V(x)}$. 
Define 
\begin{equation}\label{eq:gadef}
	\ga_j(a;n):= \int_{\R}  e^{n(a x-V(x)/2)}\psi_j(x;n)  dx.
\end{equation}

The following identity relates the higher rank case to the rank one cases. 


\begin{thm} \label{thm:alg}
Let $E$ be a subset of $\realR$ and let $s$ be a complex number such that $\cE_{n-j+1,n}(E;s) \neq 0$ for all $j = 1, \cdots, \m$. 
We have for distinct $\aaa_1, \cdots, \aaa_{\m}$, 
\begin{equation}\label{eq:alg}
	\bar{\mathcal{E}}_{n,n} (\aaa_1, \cdots, \aaa_\m; E; s)
	= \frac{\det \left[ \ga_{n-j}(\aaa_k;n) \bar{\mathcal{E}}_{n-j+1,n}(\aaa_k; E; s) \right]^\m_{j,k=1}}{\det \left[ \ga_{n-j}(\aaa_k;n) \right]^\m_{j,k=1}}.
\end{equation}
\end{thm}

\begin{rmk} \label{rmk_algebraic_lHopital}
When some $\aaa_j$ are identical, the above theorem still holds by using 
\LHopital's rule. 
This follows from the smooth dependence of the quantities above in $\aaa_j$'s
which can be proved directly. 
The explicit smooth dependence of multiple orthogonal polynomials on $\aaa_j$'s, which is essentially equivalent to the smooth dependence of quantities in \eqref{eq:alg}, is shown in, for example, \cite{McLaughlin07, Bertola-Buckingham-Lee-Pierce11, Bertola-Buckingham-Lee-Pierce11a} in similar situations. 
However, in the rest of the paper, we  consider only the case when $\aaa_j$ are all distinct. 
\end{rmk}

From the above theorem, the study of the limiting distribution of the eigenvalue of higher rank case may be reduced to 
a study of rank $1$ case, which was done in \cite{Baik-Wang10a}.
However, for the interesting cases when $\aaa_j$'s converge to the same number in the limit, 
the numerator and denominator 
both tend to zero and thus we need to perform suitable row and column operations and extract the common decaying factors to make the ratio finite. 
This requires us to extend the asymptotic result of Baik and Wang \cite{Baik-Wang10a} to include the sub-leading terms of the asymptotics of $\bar{\mathcal{E}}_{n-j+1,n}(\aaa_k; E; s)$. 
Nevertheless we requires only the existence of the asymptotic expansion but not the exact formulas, 
and hence most of the extension of the result of Baik and Wang \cite{Baik-Wang10a} is straightforward. 
The technical part is the row and column operations and to show that the ratio becomes finite after factoring out the common terms.


\subsection{Assumptions on potential $V$ and some preliminary notations}
\label{subsec:assumptions_of_V_and_preliminary_notations}

In this section, 
we first state the precise conditions on $V$. 
Then we fix some notations and discuss a few important results of the rank $1$ case. 

Assume that $V$ satisfies the following three conditions:
\begin{eqnarray}
 	&&V(x) \textnormal{ is real analytic in } \R, \label{eq:condition_of_V_1} \\
 	&&\frac{V(x)}{\sqrt{x^2+1}} \to +\infty \textnormal{ as } |x|\to\infty, \label{eq:condition_of_V_2} \\
	&&V \textnormal{ is `regular'.} \label{eq:condition_of_V_3}
\end{eqnarray}
At the end of this section, we will discuss additional technical assumptions on $V$. 

Here the regularity of $V$ is a condition defined in \cite{Deift-Kriecherbauer-McLaughlin-Venakides-Zhou99} which we do not state explicitly here. 
We  note only that this condition holds for ``generic'' $V$ \cite{Kuijlaars-McLaughlin00} and for such $V$, the
density $\Psi(x)$  of the associate equilibrium measure (the limiting empirical measure when there is no external source) vanishes like a square-root at the edges of its support. 

In the usual unitary ensembles (with no external source), 
the condition~\eqref{eq:condition_of_V_2} is typically replaced by $\frac{V(x)}{\log(x^2+1)} \to +\infty$ as $|x|\to \infty$ \cite{Deift-Kriecherbauer-McLaughlin-Venakides-Zhou99}.  
Here~\eqref{eq:condition_of_V_2}  is needed to ensure that the probability density \eqref{eq:pdf_of_external_source_model} is well defined for all (spiked) $\A_n$. 


With the above assumptions, the support $\Psi(x)$ 
consists of finitely many intervals: 
\begin{equation}\label{eq:defn_of_J}
  	J := \bigcup^N_{j=0} (b_j, a_{j+1}), \quad \textnormal{where } b_0 < a_1 < \cdots < a_{N+1},
\end{equation}
for some $N\ge 0$. Note that we allow in this paper that $N$ can be larger than $0$.
We denote the \emph{right end-point of $J$} by 
\begin{equation}
  \redge := a_{N+1}
\end{equation}
as in \cite{Baik-Wang10a}.
We also set 
\begin{equation} \label{eq:defn_of_beta}
  	\beta := \big( \lim_{x\uparrow \redge} \frac{\pi\Psi(x)}{\sqrt{\redge-x}} \big)^{2/3}.
\end{equation}
By the condition~\eqref{eq:condition_of_V_3}, $\beta$ is a nonzero positive number. 
It is also known that under the above assumptions  
(see \cite{Deift-Kriecherbauer-McLaughlin-Venakides-Zhou99} and \cite{Deift-Gioev07a}) 
for the usual unitary ensemble with no external source,
\begin{equation}
  \lim_{n \to \infty} \Prob_n \left( \textnormal{the largest eigenvalue} < \redge + \frac{T}{\beta n^{2/3}} \right) = F_{\TW}(T),
\end{equation}
where $F_{\TW}$ is the Tracy--Widom distribution (see \eqref{eq:defn_of_F_TW} for definition.) 

\bigskip

We now recall a few notations and results from the analysis of rank one case \cite{Baik-Wang10a}.
Let 
\begin{equation}
	\gfn(z) := \int_J \log(z-s) \Psi(s)ds , \qquad z \in \compC \setminus (-\infty, \redge),
\end{equation}
be the so-called $\gfn$-function associated to $V$. 
Let $\ell$ be the Robin constant which is defined by the condition 
\begin{equation}\label{eq:ellc}
	\gfn_+(x) +\gfn_-(x) - V(x) = \ell, \qquad x \in \bar{J}. 
\end{equation}
We also define two functions
\begin{equation} \label{eq:definition_of_GH}
	\Gfn(z;a):= \gfn(z)- V(z)+az, \quad
	\Hfn(z;a):= -\gfn(z) + az +\ell
\end{equation}
for $a>0$. These  functions play an important role in the analysis of rank one case. 
Observe that $\Gfn(\redge;a)= \Hfn(\redge; a)$ from~\eqref{eq:ellc}. 
The function $\Hfn(x; a)$ is convex in $x\in [\redge, \infty)$. 
Let $c(a)\in [\redge, \infty)$ be the point at which $\Hfn(x;a)$ takes its minimum. It is easy to check that $c(a)= \redge$ for $a\ge \frac12 V'(\redge)$ and 
$c(a)>\redge$ for $a< \frac12 V'(\redge)$. 

Now let $\acc$ be the \underline{critical value} associated to $V$ defined by 
\begin{equation}\label{eq:deac}
	\acc:= \inf\{ a\in (0, \infty) | \text{ there exists $\bar{x}\in (c(a), \infty), \infty)$ such that $\Gfn(\bar{x}; a)>\Hfn(c(a);a)$} \}.
\end{equation}
In general, $\acc\in (0, \frac{1}{2}V'(\redge)]$. 
If $V(x)$ is convex for $x\ge \redge$, then $\acc=\frac12V'(\redge)$. 
The limiting location of the largest eigenvalue in the rank one case depends 
on whether $a<\acc$ or $a>\acc$. 
Here we denote it by $\xi(a)$ to indicate the dependence of on $a$. 

\medskip

\textbf{Super-critical case:}
Set 
\begin{equation}\label{eq:mathcalJV}
	\mathcal{J}_V:= \{ a\in [\acc, \infty) | \text{  $\max_{x\in [c(a), \infty)} \Gfn(x;a)$ attains its maximum at more than one point} \}.
\end{equation}
This is a discrete set. If $V(x)$ is convex in $x\ge \redge$, then $\mathcal{J}_V=\emptyset$. 
For $a>\acc$ such that $a\notin \mathcal{J}_V$, let \underline{$x_0(a)$} denote the point in $[c(a), \infty)$ at which $\Gfn(x;a)$ takes its maximum.  
For such $a$, it was shown that $x_0(a)$ is a continuous, strictly increasing function. 
Moreover, $\xi(a)$, the limiting location of the largest eigenvalue in the rank one case, equals $x_0(a)$ in this case. 
On the other hand, if $a>\acc$ and $a\in \mathcal{J}_V$, 
then $\xi(a)$ is a discrete random variable whose values are the maximizers of $\max_{x\in (c(a), \infty)} G(x;a)$ (there are at least two of them). 
We call $a>\acc$ such that $a\in \mathcal{J}_V$ the \underline{secondary critical values}. 

\medskip

\textbf{Sub-critical case:}
On the other hand, if $a<\acc$, then $\xi(a)=\redge$.

\medskip

\textbf{Critical case:} At the critical case when $a=\acc$,  $\xi(a)$ depends on whether $\acc=\frac12 V'(\redge)$ or $\acc< \frac12 V'(\redge)$. 
In both cases, let us assume that $\acc\notin \mathcal{J}_V$. 
Then when $\acc=\frac12 V'(\redge)$,  $\xi(a)=\redge$ as in the sub-critical case. 
But when $\acc< \frac12 V'(\redge)$, $\xi(a)$ is a discrete random variable whose value is either $\redge$ or the unique maximizer $x_0(\acc)$ of $\max_{x\in (c(a), \infty)} G(x;a)$ (which equals $\Hfn(c(\acc);\acc)$ from the definition~\eqref{eq:deac}).

\medskip

For the rest of the paper, we  assume that $V$ is a potential such that 
\begin{eqnarray}
  	&&\textnormal{$\acc\notin\mathcal{J}_V$} \label{eq:condition_of_V_4} 
\end{eqnarray}
and 
\begin{equation}
\begin{split}
	&\textnormal{for $a\in \mathcal{J}_V\setminus \{\acc\}$, $\max_{x\in (c(a), \infty)} G(x;a)$ is attained at two points $x_1(a)$ and $x_2(a)$.}
\end{split} \label{eq:condition_of_V_7}
\end{equation}
Moreover, we assume that 
\begin{eqnarray}
	&&\textnormal{$\Gfn''(x_0(a); a)\neq 0$ for $a\in (\acc, \infty)\setminus \mathcal{J}_V$,} \label{eq:condition_of_V_7-1} \\
	&&\textnormal{$\Gfn''(x_1(a); a)\neq 0$, $\Gfn''(x_2(a); a)\neq 0$  for $a\in \mathcal{J}_V\setminus \{\acc\}$,} \label{eq:condition_of_V_6}  \\
	&&\textnormal{$\Gfn''(x_0(\acc); \acc)\neq 0$ for $\acc<\acc$.} \label{eq:condition_of_V_5}
\end{eqnarray}
Note that under this assumption, all of $\Gfn''(x_i(a); a)$ ($i = 0,1,2$) are  negative.
In \cite{Baik-Wang10a}, these excluded cases are referred as ``exceptional cases''. 
However, to be precise, even though it is reasonable to imagine that nonexceptional cases are generic in the sense of Kuijlaars and McLaughlin \cite{Kuijlaars-McLaughlin00}, this was not established in \cite{Baik-Wang10a}. This issue will be considered  somewhere else. 

The above conditions are trivially satisfied if $V(x)$, $x\ge\redge$, is convex since in this case $\mathcal{J}_V=\emptyset$. 
We note that 
if $V$ is such that $\mathcal{J}_V$ not empty, then it is 
easy to see that $\mathcal{J}_{sV}$ is also nonempty for real number $s$ close enough to $1$. 
Also it is easy to find an example of nonconvex potential $V$ such that $\mathcal{J}_V\neq\emptyset$ by considering a double-well potential. (See Remark 1.6 of \cite{Wang11}.)

\medskip

The analysis of this paper applies to the excluded cases without much change but we do not include them here for the sake of presentation.

We use the following notations for two intervals that appear frequently:
\begin{align}
  \Int:= & \big[ \redge + \frac{T}{\beta n^{2/3}}, \infty \big), \label{eq:interval} \\
  \Intx(x_*):= & \big[ x_* + \frac{T}{\sqrt{-\Gfn''(x_*)n}}, \infty\big), \label{eq:interval2}
\end{align}
for $T \in \realR$ and for $x_* > \redge$, assuming that $\Gfn''(x_*) < 0$ in the later case.

\subsection{Statement of main results}\label{sec:asym}

We now state the main results.
The asymptotic results here are stated in some cases in terms of the distribution function~\eqref{eq:jthlarge} $\Prob_{d,n}^{(j)}$ 
and in other cases in terms of the expectation~\eqref{eq:defn_of_generating_func_generalized} $\mathcal{E}_{d,n}$.
This choice is simply to make the formula compact. 
The analysis applies to both quantities and indeed it is easy to deduce one result from the other
from the relation~\eqref{eq:jthlarge} and the uniformity of the asymptotics in $s$ near $1$.
We can also express all the results in terms of correlation functions but we find that the attention to individual eigenvalue is more illustrating in the current framework.

We use the phrase that a limit holds ``uniformly in $s$ which is close to $1$'' in several places, 
for example in Theorem~\ref{thm:sub}.
This means that 
there exists a complex neighborhood of $1$ independent of $n$ in which the limit holds. 
A slightly more careful analysis would show that $s$ uniformly converges in a larger domain (\eg, in any compact subset of $\compC \setminus (1 + \epsilon, \infty)$, $\epsilon > 0$,) but we do not discuss this issue in this paper. 

We state the results under the ``genericity assumptions''~\eqref{eq:condition_of_V_4}-~\eqref{eq:condition_of_V_5}, in addition to the conditions~\eqref{eq:condition_of_V_1}-~\eqref{eq:condition_of_V_3} discussed in the last subsection. 
We group the asymptotic results into sub-critical, super-critical and critical cases.

When $d=n$, we use the notation $\Prob_{n}^{(j)}$ for $\Prob_{d,n}^{(j)}$  and $\mathcal{E}_{n}$ for $\mathcal{E}_{d,n}$, respectively. 
We also state the results only in this case. The case when $d\neq n$ is similar.

\subsubsection{Sub-critical case}

The first result is on the sub-critical case when all external eigenvalues are smaller than the critical value $\acc$. 
In this case the external source does not change the location and the limiting distribution of the top eigenvalues. 

Let 
\begin{equation} \label{eq:Airy_kernel}
  K_{\Airy}(x,y) := \frac{\Ai(x)\Ai'(y) - \Ai'(x)\Ai(y)}{x-y}
\end{equation}
be the Airy kernel.
For any $T \in \realR$ and $s \in \compC$, set 
\begin{equation} \label{eq:defn_of_TW_distr_s_version}
  F_{\TW}(T;s) := \det(1 - s\chi_{[T,\infty)}K_{\Airy}\chi_{[T,\infty)}),
\end{equation}
where $\chi_E$ denotes the projection  on the set $E$. 
Then
\begin{equation} \label{eq:defn_TW_j-th_distribution}
  	F^{(j)}_{\TW}(T) := \sum_{i=0}^{j-1} \frac{(-1)^i}{i!} \frac{d^i}{d s^i} \bigg|_{s=1} F_{\TW}(T;s)
\end{equation}
is the Tracy--Widom $j$-th eigenvalue distribution. In particular,
\begin{equation} \label{eq:defn_of_F_TW}
  F_{\TW}(T) := F^{(1)}_{\TW}(T) = \det(1-\chi_{[T,\infty)}K_{\Airy}\chi_{[T,\infty)})
\end{equation}
is the Tracy--Widom distribution.

\begin{thm}[sub-critical case]\label{thm:sub}
  Let $\aaa_1, \cdots, \aaa_\m$ be fixed numbers such that $\max\{ \aaa_1, \cdots, \aaa_\m\} < \acc$. Assume that $\aaa_1, \cdots, \aaa_{m}$ are positive and distinct. Then for each $T\in \R$, 
\begin{equation} \label{eq:result_of_subcritical}
	\lim_{n \to \infty} \cE_n(\aaa_1, \cdots, \aaa_\m; \Int; s) = \FGUE(T; s)
\end{equation}
uniformly in $s$ which is close to $1$.
\end{thm}


\begin{rmk}
The assumption that $\aaa_1, \cdots, \aaa_{\m}$ are positive in Theorem \ref{thm:sub} can be removed. 
The proof of the above theorem uses the calculations in \cite{Baik-Wang10a} of the gap probability $\Prob_{n-j+1, n}(a; E)$ which are only detailed for $a > 0$. 
As suggested in \cite[Section 2]{Baik-Wang10a}, there is a similar asymptotic result for $\Prob_{n-j+1, n}(a; E)$ for $a < 0$ from which we can obtain the same result as the above theorem when some of $\aaa_j$'s are not positive.  This remark applies also to Theorem \ref{thm:sup1} below. 
\end{rmk}

\begin{rmk} \label{rmk:distinctness}
The assumption that $\aaa_1, \cdots, \aaa_\m$ are distinct in Theorem \ref{thm:sub} is technical and the result should hold without this assumption. 
The starting formula of the proof of this theorem is the identity~\eqref{eq:alg}. When some of $\aaa_j$'s are identical, the right-hand side of~\eqref{eq:alg} becomes more complicated by using  \lHopital's rule. This in turn requires a more detailed asymptotic results for $\Prob_{n-j+1, n}(a; E)$. The analysis of Baik and Wang \cite{Baik-Wang10a} can be extended for this  but we do not pursue this 
in this paper for the sake of space and presentation. The same remark applies to all other theorems in this section. It is interesting to contrast this situation to the papers \cite{Bertola-Buckingham-Lee-Pierce11} and \cite{Bertola-Buckingham-Lee-Pierce11a} which analyzed the similar model using the Riemann--Hilbert problem for multiple orthogonal polynomials. In that approach, the case in which all $\aaa_j$'s are identical is the simplest to analyze. 
\end{rmk}

\subsubsection{Super-critical case}


In this section we consider the super-critical case in which some of the external source eigenvalues are strictly larger than the critical value $\acc$. 
In this case  large external source eigenvalues do have an effect on the top eigenvalues. 
We consider three sub-cases. 
In the first two cases, we assume that $\aaa_j\notin \mathcal{J}_V$ for all $j$.
The first among these is the case when $\aaa_j$ are separated by $O(1)$ distances. 
In the second case, the external source eigenvalues are asymptotically the same.
The third case is the secondary critical case when $\aaa_j$ are all asymptotically equal to some $a\in  \mathcal{J}_V\setminus\{\acc\}$.
From the discussion of Section~\ref{subsec:assumptions_of_V_and_preliminary_notations}, the last case does not occur if $V(x)$ is convex for $x\in [\redge, \infty)$. 
 

\bigskip

Let 
\begin{equation}\label{eq:erfdef}
    \erf(T) := \frac{1}{\sqrt{2\pi}}\int^{T}_{-\infty} e^{-x^2/2}dx
  \end{equation}
be the \cdf\ of the standard normal distribution.
  
\begin{thm}[super-critical case 1: separated external source eigenvalues]\label{thm:sup1}
Let $\aaa_1, \cdots, \aaa_{\m}$ be fixed positive and distinct numbers. Suppose that there is  $p\in\{1, \cdots, \m\}$
such that 
\begin{equation}
	\text{$\aaa_j>\acc$ for $j=1, \cdots, p$, and $\aaa_j<\acc$ for $j=p+1, \cdots, \m$. }
\end{equation}
Assume, without loss of generality, that 
$\aaa_1 > \aaa_2 > \cdots > \aaa_p$. 
Suppose that  $\aaa_j \notin \mathcal{J}_V$ and $\Gfn''(x_0(\aaa_j)) \neq 0$ for each $j = 1, \cdots, p$. 
Then for each $T\in \R$ and $j = 1, \cdots, p$,
  \begin{equation}\label{eq:thmsup1}
    \lim_{n \to \infty} \Prob^{(j)}_n(\aaa_1, \cdots, \aaa_\m; \Intx(x_0(\aaa_j)) ) = \erf(T),
  \end{equation}
  where $\Intx(x_*)$ is defined in~\eqref{eq:interval2}. 
We also have, for $j=1,2,\cdots$, 
\begin{equation}\label{eq:thmsup2}
    \lim_{n \to \infty} \Prob^{(p+j)}_n(\aaa_1, \cdots, \aaa_\m; \Int) ) = F^{(j)}_{\TW}(T),
  \end{equation}
  where $F^{(j)}_{\TW}(T)$ is the Tracy--Widom $j$-th eigenvalue distribution defined in \eqref{eq:defn_TW_j-th_distribution}.
\end{thm}

Theorem \ref{thm:sup1} demonstrates that each of the external source eigenvalues which is greater than $\acc$ 
pulls exactly one eigenvalue out of the support of the equilibrium measure. 
The limiting  location of each pulled-off eigenvalue depends only on the corresponding external source eigenvalue. 
The fluctuation of each pulled-off is Gaussian. 
The rest of the eigenvalues are unaffected by the external source eigenvalues asymptotically.

\bigskip

We now consider the situation when the external source eigenvalues are asymptotically the same. 
A non-Gaussian fluctuation appears when they converge together in a particular fashion. 
Define, for distinct $\alpha_1, \cdots, \alpha_k$, 
\begin{equation} \label{eq:defination_of_spiked_GUE_dist_rank_k}
	\erf_k(T; \alpha_1, \cdots, \alpha_k; s) := \frac{\det \left[ \int^{\infty}_{-\infty}x^{i-1}e^{-x^2/2+\alpha_j x} (1 - s\chi_{(T,\infty)}(x)) dx \right]_{1 \leq i,j \leq k}}{\det \left[ \int^{\infty}_{-\infty}x^{i-1}e^{-x^2/2+\alpha_j x}dx \right]_{1 \leq i,j \leq k}}.
\end{equation}
Observe that 
\begin{equation}
  \erf_k(T; \alpha_1, \cdots, \alpha_k; s) = \cE_{k,1}(\alpha_1, \cdots, \alpha_k; [T,\infty); s),
\end{equation}
in terms of the notation~\eqref{eq:defn_of_generating_func_generalized} when $V(x) = x^2/2$ in~\eqref{eq:generalized_pdf}. 
Hence $\erf_k(T; \alpha_1, \cdots, \alpha_k; s)$ is an expectation that arises from  
the $k\times k$ Gaussian Unitary ensemble (GUE) with external source $\diag(\alpha_1, \cdots, \alpha_\m)$. 
As a special case, 
\begin{equation} \label{eq:defination_of_distr_of_j_eigen_in_k_GUE}
  	\erf^{(j)}_k(T) := \sum^{j-1}_{i=0} \frac{(-1)^i}{i!} \frac{d^i}{d s^i} \bigg|_{s=1} \erf_k(T; 0, \cdots, 0; s)
\end{equation}
is the \cdf\ of the $j$-th largest eigenvalue of the $k$-dimensional GUE.
When $j=k=1$, this equals $\erf(T)$. 

\begin{thm}[super-critical case 2: clustered external source eigenvalues]\label{thm:sup2}
Let $a$ be a fixed number such that $a > \acc$, $a \notin \mathcal{J}_V$ and $\Gfn''(x_0(a)) \neq 0$. Set 
\begin{equation}\label{eq:aaaksca}
	\aaa_k = a+ \sqrt{-\Gfn''(x_0(a))}\frac{\alpha_k}{\sqrt{n}} , \qquad k=1, \cdots, \m.
\end{equation}
for fixed distinct $\alpha_1, \cdots, \alpha_\m$. Then for each $T\in \R$, 
\begin{equation} \label{eq:result_of_thm:sup2}
  \lim_{n \to \infty} 
   \cE_n(\aaa_1, \cdots, \aaa_\m; \Intx(x_0(a)); s) = \erf_\m(T; \alpha_1, \cdots, \alpha_\m; s)
\end{equation}
uniformly in $s$ which is close to $1$.
\end{thm}

Hence in this case the $\m$ eigenvalues which are outside of the bulk converge to the same location $x_0(a)$.
After a scaling, they fluctuate as the eigenvalues of $\m\times \m$ GUE matrix with external source $\diag(\alpha_1, \cdots, \alpha_\m)$.

\begin{rmk}\label{rmk:mixsuper}
We can also consider the general case that for $a_1, a_2, \cdots>\acc$ where $a_j\notin \mathcal{J}_V$ and $\Gfn''(x_0(a_j))\neq 0$, $p_1$ external source eigenvalues are close to $a_1$, $p_2$ external source eigenvalues are close to $a_2$, etc. Then for each $j$, $p_j$ eigenvalues converge to $x_0(a_j)$ and they fluctuate like the eigenvalues of $p_j\times p_j$ GUE with certain external source. This can be obtained by combining the proofs of Theorem~\ref{thm:sup1} and~\ref{thm:sup2} but it is tedious. We omit the proof. 
\end{rmk}

\begin{rmk}
In Theorem~\ref{thm:sup2}, we can also prove the analogue of~\eqref{eq:thmsup2} and show that for each $j \geq 1$, the $(\m+j)$-th eigenvalue converges to $\redge$ and its limiting distribution is the Tracy--Widom $j$-th eigenvalue distribution defined in \eqref{eq:defn_TW_j-th_distribution}. Similar remark also applies to Theorem \ref{thm:superjump} and to Theorem \ref{thm:critical2} below. 
\end{rmk}

\bigskip

We now consider the situation when all external source eigenvalues are near or at a secondary critical value of $V$. 
In this case, we will state the result under the assumption that 
the support of the equilibrium measure of $V$ consists of one interval (i.e. $N=0$ in \eqref{eq:defn_of_J}). 
This assumption is made only for the ease of statement: see Remark~\ref{rmk:multiinterval} below
how the result is changed if $N>0$. 

Let $a\in \mathcal{J}_V$ and we consider the situation when $\m$ external source eigenvalues converge to $a$.
Under the assumption~\eqref{eq:condition_of_V_7}, the top $\m$ eigenvalues converge to one of the two possible locations, which we denote by $x_1(a)<x_2(a)$. 
How many of the eigenvalues converge to each of them? 
It turned out that any number is possible and it depends on how fast the external source eigenvalues converges to $a$.
There are $\m$ distinct scalings. 
To each scaling indexed by an $\mm\in \{1, \cdots, \m\}$ 
a number $\bfR_m\in (0,1)$ is associated such that 
either one of the following two happens: 
with probability $\bfR_m$, 
the top $\mm-1$ eigenvalues converge to $x_2(a)$ and the next top $\m-\mm+1$ eigenvalues to $x_1(a)$, 
or 
with probability $1-\bfR_m$, the top $\mm$ eigenvalues converge to $x_2(a)$ and the next top $\m-\mm$ eigenvalues to $x_1(a)$.
In order to describe $\bfR_{\mm}$, we need some definitions. 

Since we assume $N=0$, the support $J$ of $\Psi(x)$ is of form $J= (\ledge, \redge)$. Set
$\gamma(z):= \big( \frac{z-\ledge}{z-\redge} \big)^{1/4}$ which is defined on $\compC\setminus[\ledge, \redge]$ and satisfies $\gamma(z)\sim 1$ as $z\to \infty$. 
Define, for $j\in \intZ$, 
\begin{equation}\label{eq:Mjindepn}
	\M_j(z)=  \sqrt{\frac{2}{\pi(\redge-\ledge)}} \frac{\gamma(z)+\gamma(z)^{-1}}2
	\bigg( \frac{\gamma(z)-\gamma(z)^{-1}}{\gamma(z)+\gamma(z)^{-1}} \bigg)^j, 
	\qquad z\in \compC\setminus(-\infty, \redge]
\end{equation}
for $z$ in $\compC\setminus(-\infty, \redge]$. 
For $j = 0, \cdots, \m$ and for distinct $a,b \in (\redge, \infty)$, we define the matrix
\begin{equation}\label{eq:superj1030}
	\frakP^{(a,\m - j),(b,j)} :=  
        \begin{bmatrix} 
	\M_1(a) & \M_1'(a) & \cdots & \M_1^{(\m-j-1)}(a)& \M_1(b) 
	& \M_1'(b) & \cdots & \M_1^{(j-1)}(b) \\
	\vdots& \vdots& \ddots& \vdots& \vdots& \vdots& \ddots& \vdots \\
	\M_\m(a) & \M'_\m(a) & \cdots & \M_\m^{(\m-j-1)}(a) & \M_\m(b) 
	& \M'_\m(b) & \cdots & \M_\m^{(j-1)}(b)  \\
	\end{bmatrix}.
\end{equation}
We show in Proposition~\ref{prop:non-vanishing}\ref{enu:prop:non-vanishing:b}  that the determinant of $\frakP^{(a,\m - j),(b,j)}$ is nonzero and $(-1)^{\m(\m-1)/2}$ times the determinant is positive 
if $a<b$.
We also define, for $c \neq 0$ and distinct real numbers $\alpha_1, \cdots, \alpha_{\m}$, 
\begin{equation}\label{eq:superj1040}
	\frakQ_{(0,\m-j),(c,j)}(\alpha_1, \cdots, \alpha_{\m}) :=  
        \begin{bmatrix} 1 & \alpha_1  &\cdots & \alpha_1^{\m-j-1}
	& e^{c\alpha_1} & \alpha_1 e^{c\alpha_1}&\cdots 
	& \alpha_1^{j-1} e^{c \alpha_1} \\
	\vdots & \vdots & \ddots & \vdots & \vdots& \vdots & \ddots & \vdots &  \\
	1 & \alpha_\m  &\cdots & \alpha_\m^{\m-j-1} & e^{c \alpha_\m}
	& \alpha_\m e^{c\alpha_\m}&\cdots 
	& \alpha_\m^{j-1}  e^{c\alpha_\m}
	\end{bmatrix},
\end{equation}
where $j=0, 1, \cdots, \m$. 
Note that  $(-1)^{\m(\m-1)/2} \det [\frakQ_{(0,\m-j),(c,j)}(\alpha_1, \cdots, \alpha_{\m})]>0$ if $c>0$ and $\alpha_1>\cdots>\alpha_\m$. 

\begin{thm}[secondary critical case]\label{thm:superjump}
Assume that the support of the equilibrium measure associated to $V$ consists of one interval. 
Let $a$ be a secondary critical value 
(i.e. $a\in\mathcal{J}_V \setminus \{ \acc \}$ ) such that $\Gfn(x; a)$, 
$x\in (c(a), \infty)$, attains its maximum value at two points 
$x_1(a) < x_2(a)$. 
Assume that  $\Gfn''(x_1(a); a)\neq 0$ and $\Gfn''(x_2(a); a)\neq 0$. Fix $m \in \{1, 2, \cdots, \m\}$, and set
\begin{equation}\label{eq:qmm}
  	q_\mm := \frac{\m-2\mm+1}{x_2(a)-x_1(a)}
\end{equation}
and
\begin{equation}\label{eq:superj005}
\begin{split}
	K_\mm := \left( 
        \frac{(\m-m)!}{(m-1)!} \frac{(-\Gfn''(x_1(a)))^{\m-m+1/2} }{(-\Gfn''(x_2(a)))^{m-1/2}}
	 \right)^{\frac{1}{\m-2m+1}}.
\end{split}
\end{equation}
Suppose that the external source eigenvalues are 
\begin{equation}\label{eq:superj000}
	\aaa_k=  a- q_m\frac{\log (K_\mm n)}{n}+ \frac{\alpha_k}{n}, \qquad k = 1, \cdots, \m,
\end{equation}
for fixed distinct $\alpha_1 > \cdots > \alpha_\m$.
Then for any $T \in \realR$, as $n \to \infty$ we have the following.
\begin{enumerate}[label=(\alph*)]
\item
  For $k = 1, \cdots, m-1$,
  \begin{equation}\label{eq:supjmp1}
    \Prob^{(k)}_n(\aaa_1, \cdots, \aaa_\m; \Intx(x_2(a))) = \bfR_m \erf^{(k)}_{m-1}(T) +(1-\bfR_m) \erf^{(k)}_m(T) + o(1).
  \end{equation}

\item
  For $k = m$,
  \begin{align}
    \Prob^{(k)}_n(\aaa_1, \cdots, \aaa_\m; \Intx(x_2(a))) = &\bfR_m +(1-\bfR_m) \erf^{(m)}_m(T) + o(1), \\
    \Prob^{(k)}_n(\aaa_1, \cdots, \aaa_\m; \Intx(x_1(a))) = & \bfR_m \erf^{(1)}_{\m-m+1}(T) + o(1).
  \end{align}
\item
  For $k = m+1, \cdots, \m$,
  \begin{equation}
    \Prob^{(k)}_n(\aaa_1, \cdots, \aaa_\m; \Intx(x_1(a))) = \bfR_m \erf^{(k-m+1)}_{\m-m+1}(T) + (1-\bfR_m) \erf^{(k-m)}_{\m-m}(T) + o(1).
  \end{equation}
\end{enumerate}
Here $\bfR_m$ is a number in $(0,1)$  defined by 
\begin{equation}\label{eq:defn_of_R^L_i}
\begin{split}
  	\bfR_m 
	 := & \frac{\det[\bfP_{m-1}] \det[\bfQ_{m-1}]}{\det[\bfP_{m-1}] \det[\bfQ_{m-1}] + \det[\bfP_m] \det[\bfQ_m]}
\end{split}
\end{equation}
where $\bfP_{\ell} := \frakP^{(x_1(a),\m - \ell),(x_2(a), \ell)}$ and $\bfQ_{\ell} := \frakQ_{(0, \m-\ell),(x_2(a)-x_1(a), \ell)}(\alpha_1, \cdots, \alpha_{\m})$. 
The  function $\erf^{(\ell)}_k(T)$ is the distribution function of the $\ell$th largest eigenvalue of $k\times k$ GUE defined in \eqref{eq:defination_of_distr_of_j_eigen_in_k_GUE}.
\end{thm}

Observe that one more eigenvalue is pulled off from $x_1(a)$ to $x_2(a)$ as 
$\mm$ increases by $1$.
The eigenvalues clustered near each of $x_1(a)$ or $x_2(a)$ fluctuate like 
the eigenvalues of a GUE matrix of dimension equal to the cluster size.

Note that $\bfR_m$ is well defined even for nondistinct $\alpha_j$ if we apply \lHopital's rule to the right\DongNew{-}hand side of \eqref{eq:defn_of_R^L_i}.  
The theorem holds without the assumption of distinctness but we do not pursue it here.
See Remark \ref{rmk:distinctness}.

\begin{rmk}\label{rmk:multiinterval}
When the support of equilibrium measure associated to the potential $V$ consists of more than one interval (i.e. $N>0$), the above theorem still holds after one change: 
the probability $\bfR_m$ depends on $n$. 
In the formula~\eqref{eq:Mjindepn}, $\M_{j}$ needs to be changed to $\M_{j,n}$
in \cite[Formula (311)]{Baik-Wang10a}, which is expressible in terms of a Riemann theta function and depends on $n$ quasi-periodically. 
Nevertheless, it can be shown that $\bfR_m$ lies in a compact subset of $(0,1)$ for all large enough $n$
from Proposition~\ref{prop:non-vanishing}.
This is enough to extend the proof of the above theorem from $N=0$ case to $N>0$ case.

\end{rmk}

\subsubsection{Critical case}

The final two theorems concern the critical case. In this case the limiting location of the top eigenvalue(s) is about to break off from $\redge$. Recall that the critical value $\acc$ which is determined by the potential $V$ satisfies $\acc\le \frac12 V'(\redge)$. 
Depending on whether $\acc= \frac12 V'(\redge)$ or $\acc<\frac12 V'(\redge)$, 
the break-off is continuous or discontinuous. 
When $V(x)$ is convex in $x\in [\redge, \infty)$, we always have $\acc= \frac12 V'(\redge)$ and the break-off is continuous. 


\begin{figure}[htb]
  \centering
  \includegraphics{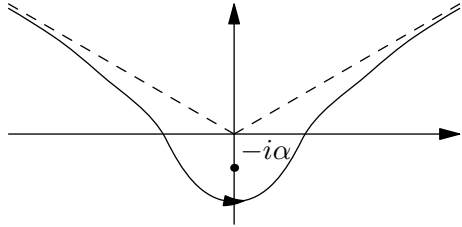}
  \caption{The contour of the integral in \eqref{eq:Calphadef}}
  \label{fig:limiting_contour_above}
\end{figure}

Let us define the limiting distributions that appear in the case when $\acc= \frac12 V'(\redge)$. 
For $\alpha \in \realR$, define the function (see \cite[Formula (15)]{Baik-Ben_Arous-Peche05} and \cite[Formula (18)]{Baik-Wang10a})
\begin{equation}\label{eq:Calphadef}
	C_{\alpha}(\xi) := \frac1{2\pi} \int e^{i\frac13 z^3+i\xi z} \frac{dz}{\alpha+iz},
\end{equation}
where the contour is from $\infty e^{5\pi i/6}$ to $\infty e^{\pi i/6}$ and the pole $z=-i\alpha$ lies above the contour in the complex plane (Figure \ref{fig:limiting_contour_above}).
Set 
for $T\in\R$ and $s\in \compC$
\begin{equation} \label{eq:defn:of_F_1(T;alpha)}
	\FGOE(T;\alpha;s) := \FGUE(T; s)\cdot 
	\bigg(1 - s\langle (1-s\chi_{[T,\infty)} K_{\Airy} \chi_{[T, \infty)})^{-1} C_{\alpha},  \chi_{[T,\infty)}\Ai  \rangle\bigg),
\end{equation}
where $\langle f,g \rangle$ denotes the real inner product over $\realR$, $\int_{\realR} f(x)g(x)dx$.  
For $k\ge 1$ and distinct real parameters $\alpha_1, \cdots, \alpha_k$, define 
\begin{equation} \label{eq:defination_of_generalized_TW_dist_rank_k}
	F_k(T; \alpha_1, \cdots, \alpha_k;s ) := F_{\TW}(T;s) 
	\frac{\det \left[ (\alpha_i +\frac{d}{d T})^{j-1} \frac{F_1(T; \alpha;s)}{F_{\TW}(T;s)} \right]_{1 \leq i,j \leq k}}{\prod_{1 \leq i < j \leq k} (\alpha_j-\alpha_i)}.
\end{equation}
When $s=1$, the function $\FGOE(T;\alpha;1)$ was defined in \cite[Definition 1.3]{Baik-Ben_Arous-Peche05} 
and $F_k(T; \alpha_1, \cdots, \alpha_k;1)$ was introduced in \cite[Theorem 1.1]{Baik06}. They are known to be distribution functions and can be expressed in terms of Painlev\'e II equation and its Lax pair equations. It is also known that 
$\FGOE(T;0;1)$ is the square of the GOE Tracy--Widom distribution  (see \cite[Formula (24)]{Baik-Ben_Arous-Peche05}). (The  function $F_k(T; \alpha_1, \cdots, \alpha_k;1)$ is shown to be the limiting distribution of the largest eigenvalue in the spiked  model of rank $k$ at the critical case
for the potentials  $V(x)= ((1+c)x-c\log x)\chi_{(0,\infty)}(x)$ and $V(x)=x^2/2$ in  \cite{Baik-Ben_Arous-Peche05} and \cite{Peche05},
respectively. 
It is easy to check from the determinantal point process structure that  
\begin{equation} \label{eq:defn_TW_j-th_distribution_generalized}
  	F^{(j)}_k(T; \alpha_1, \cdots, \alpha_k) := 
  \sum_{i=0}^{j-1} \frac{(-1)^i}{i!} \frac{d^i}{d s^i} \bigg|_{s=1} F_k(T; \alpha_1, \cdots, \alpha_k;s )
\end{equation}
is the limiting distribution of the $j$-th largest eigenvalue in these potentials even though 
this was not discussed in \cite{Baik-Ben_Arous-Peche05, Peche05}.)

\begin{thm}[critical case 1: continuous transition]\label{thm:critical1}
Suppose that $V$ is a potential such that $\acc= \frac12 V'(\redge)$ and $\acc\notin\mathcal{J}_V$. Suppose that 
\begin{equation} \label{eq:defn_of_ext_eigenvalues_critical_1}
	\aaa_k= \acc+ \frac{\beta \alpha_k}{n^{1/3}}, \qquad k=1,\cdots, \m, 
\end{equation}
for distinct real numbers $\alpha_1, \cdots ,  \alpha_\m$. 
Then for each $T\in\R$, 
\begin{equation} \label{eq:result_of_thm:critical1}
	\lim_{n\to\infty} \cE_n(\aaa_1, \cdots, \aaa_\m; \Int; s)= F_{\m}(T; -\alpha_1, \cdots, -\alpha_{\m};s)
\end{equation}
uniformly in $s$ which is close to $1$.
\end{thm}

\bigskip

We now consider the case with potential $V$ such that $\acc <\frac{1}{2}V'(\redge)$. 
As usual, we assume $\acc\notin\mathcal{J}_V$. 
As in Theorem~\ref{thm:superjump} above, we also assume,  for the ease of statement,  that the support of the equilibrium measure of $V$ consists of a single interval. 
An analogue of Remark~\ref{rmk:multiinterval} applies to the multiple interval case.

Define (cf.~\eqref{eq:Mjindepn}), for $j\in \intZ$, 
\begin{equation}\label{eq:Mjindepntilde}
	\tilde{\M}_j(z)=  \sqrt{\frac{2}{\pi(\redge-\ledge)}} \frac{\gamma(z)-\gamma(z)^{-1}}{-2i}
	\bigg( \frac{\gamma(z)-\gamma(z)^{-1}}{\gamma(z)+\gamma(z)^{-1}} \bigg)^{-j}, 
	\qquad z\in \compC\setminus(-\infty, \redge].
\end{equation}
Note that $-i\tilde{\M}_{j}(x) > 0$ for $x \in (\redge, \infty)$. 
For $j = 0, 1, \cdots, \m$ and $a, b \in (\redge, \infty)$, we define, similar to \eqref{eq:superj1030}, the matrix
\begin{equation}\label{eq:superj1030-1}
  \frakP^{(b, j)}_{(a,\m-j)} := 
  \begin{bmatrix} 
    \tilde{\M}_1(a) & \tilde{\M}_1'(a) & \cdots & \tilde{\M}_1^{(\m-j-1)}(a)& \M_1(b) 
    & \M'_1(b) & \cdots & \M_1^{(j-1)}(b) \\
    \vdots& \vdots& \ddots& \vdots& \vdots& \vdots& \ddots& \vdots \\
    \tilde{\M}_\m(a) & \tilde{\M}_\m'(a) & \cdots & \tilde{\M}_\m^{(\m-j-1)}(a) & \M_\m(b) 
    & \M'_\m(b) & \cdots & \M_\m^{(j-1)}(b)  \\
  \end{bmatrix}.
\end{equation}
Also recall the matrix $\frakQ_{(0,\m-j),(c,j)}(\alpha_1, \cdots, \alpha_{\m})$ defined in~\eqref{eq:superj1040}

Recall the point $x_0(\acc)$ which is well defined when $\acc< \frac12 V'(\redge)$ from Section~\ref{subsec:assumptions_of_V_and_preliminary_notations}.
Note that $x_0(\acc)>c(\acc)>\redge$. 
The point $x_0(\acc)$ is unique when $\acc\notin\mathcal{J}_V$.

\begin{thm}[critical case 2: jump transition]\label{thm:critical2}
 Let $V$ be a potential such that $\acc< \frac12 V'(\redge)$. Assume that $\acc\notin\mathcal{J}_V$ and $\Gfn''(x_0(\acc); \acc)\neq 0$. Fix $m \in \{ 1, 2, \cdots, \m \}$ and suppose that 
\begin{equation}\label{eq:superj000-1}
	\aaa_k=  \acc- \tilde{q}_m\frac{\log (\tilde{K}_m n)}{n}+ \frac{\alpha_k}{n}, \quad k = 1, \cdots, \m,
\end{equation}
for distinct $\alpha_1 > \cdots > \alpha_\m$, 
where
\begin{equation}
  \tilde{q}_m := \frac{\m-2m+1}{x_0(\acc)-c(\acc)}
\end{equation}
and
\begin{equation}\label{eq:superj005-1}
	\tilde{K}_m := \left( \frac{(\m-m)!}{(m-1)!} \frac{(\Hfn''(c(\acc)))^{\m-m+1/2} }{(-\Gfn''(x_0(\acc)))^{m-1/2}} \right)^{\frac{1}{\m-2m+1}}.
\end{equation}
Assume the support of the equilibrium measure of $V$ consists of a single interval.
Then for any $T \in \realR$, as $n \to \infty$ we have the following.
\begin{enumerate}[label=(\alph*)]
\item
  
For $k = 1, \cdots, m-1$,
  \begin{equation}
    \Prob^{(k)}_n(\aaa_1, \cdots, \aaa_\m; \Intx(x_0(\acc))) = \tilde{\bfR}_m \erf^{(k)}_{m-1}(T) +(1-\tilde{\bfR}_m) \erf^{(k)}_m(T) + o(1).
  \end{equation}

\item
  For $k = m$,
  \begin{align}
    \Prob^{(k)}_n(\aaa_1, \cdots, \aaa_\m; \Intx(x_0(\acc))) = & \tilde{\bfR}_m +(1-\tilde{\bfR}_m) \erf^{(m)}_m(T) + o(1), \\
    \Prob^{(k)}_n(\aaa_1, \cdots, \aaa_\m; \Int) = & \tilde{\bfR}_m F_{\TW}(T) + o(1).
  \end{align}
\end{enumerate}
Here
\begin{equation}
  \tilde{\bfR}_{m}:= \frac{\det[\tilde{\bfP}_{m-1}] \det[\tilde{\bfQ}_{m-1}]}{\det[\tilde{\bfP}_{m-1}] \det[\tilde{\bfQ}_{m-1}] + \det[\tilde{\bfP}_m] \det[\tilde{\bfQ}_m]}, \label{eq:defn_of_tilde_R^L_i} 
\end{equation}
where  $\tilde{\bfP}_{\ell} := (-i)^{\m-\ell}\frakP^{(x_0(\acc),\ell)}_{(c(\acc),\m-\ell)}$ and $\tilde{\bfQ}_{\ell} := \frakQ_{(0, \m-\ell),(x_0(\acc)-c(\acc), \ell)}(\alpha_1, \cdots, \alpha_{\m})$. For each $\ell$, $\det[\tilde{\bfP}_{\ell}]\det[\tilde{\bfQ}_{\ell}] > 0$, and hence $ \tilde{\bfR}_m\in (0,1)$. The distribution functions $\erf^{(\ell)}_k(T)$ are defined in \eqref{eq:defination_of_distr_of_j_eigen_in_k_GUE}, and $F_{\TW}(T)$ is the Tracy--Widom distribution in \eqref{eq:defn_of_F_TW}.
\end{thm}

Hence in this case, some eigenvalues are pulled off the edge of the equilibrium measure and fluctuate as the eigenvalues of Gaussian unitary ensemble. The largest eigenvalue which is not pulled off the edge of the equilibrium measure fluctuates as the Tracy--Widom distribution. 

It can also be shown that the $(m+j)$-th eigenvalue has either the Tracy--Widom $(j+1)$-th eigenvalue distribution, or the Tracy--Widom $j$-th eigenvalue distribution as its limiting distribution, with probability $\tilde{\bfR}_m$ and $1-\tilde{\bfR}_m$, respectively. 

The proof of this theorem is very close to the proof of Theorem~\ref{thm:superjump} and we skip it.

\bigskip

\begin{rmk}
In this paper we  state only limit theorems for individual top eigenvalues. The analysis of this paper can be modified to obtain the limit theorems of joint distribution of top eigenvalues. The result is what one would expect. We skip the detail.
\end{rmk}

\subsection{Comments on the proofs and organization of the paper}\label{sec:org}


We prove the asymptotic results in Section~\ref{sec:asym} based 
on Theorem~\ref{thm:alg} 
and extensions of the asymptotic results 
for the rank one case of \cite{Baik-Wang10a}. 
The latter are mostly routine and we do not give full details. 
In employing the above strategy to prove Theorems~\ref{thm:sub} and~\ref{thm:sup1}, 
we need to prove the limit of the determinant in the denominator of~\eqref{eq:alg}
is nonzero. 
This is done in Section~\ref {sec:proof_of_non_vanishing_property}. 
The proofs of Theorems~\ref{thm:sup2},~\ref{thm:superjump} and~\ref{thm:critical2} 
are more complicated since the denominator of~\eqref{eq:alg} converges to zero and hence 
we need to show that the numerator and the denominator have the same vanishing factors in their asymptotics. 
This requires careful linear algebraic manipulations.
After factoring out the vanishing term, the denominator converges to a certain determinant which 
we show again nonzero in Section~\ref {sec:proof_of_non_vanishing_property}. 
The proof of Theorem~\ref{thm:critical1} also follows this general strategy 
but we use a variation of Theorem \ref{thm:alg} and the analysis is more involved. 
We skip the proof of Theorem~\ref{thm:critical2} since it is very close to that of Theorem~\ref{thm:superjump}.

\medskip

The rest of the paper is organized as follows.
Theorems \ref{thm:sub}--\ref{thm:critical1} are proved in the Sections~\ref{sec:proof_of_sub_and_sup_1}--\ref{higher_rank_critical_case}, respectively. 
As mentioned in the above paragraph, in the proofs in these sections, we need to show that a certain determinant is nonzero. 
This is done in Section~\ref{sec:proof_of_non_vanishing_property} in a unifying way. 
The algebraic theorem, Theorem~\ref{thm:alg}, that relates the  higher rank case to the  rank $1$ case is proved in Section \ref{sec:higher}. 


\section{Proof of Theorem~\ref{thm:sub}: sub-critical case} \label{sec:proof_of_sub_and_sup_1}

Recall that the sub-critical case is when 
\begin{equation}
	\max\{\aaa_1, \cdots, \aaa_\m\} < \acc.
\end{equation}
We also assume that $\aaa_j$'s are all positive, distinct and fixed. 
In order to use Theorem \ref{thm:alg}, we need the asymptotics of $\bfGamma_{n-j}(\aaa_k;n)$ and $\bar{\cE}_{n-j+1, n}(\aaa_k; \Int; s)=\frac{\cE_{n-j+1,n}(\aaa_k;\Int;s)}{ \cE_{n-j+1}(\Int;s)}$. 

\subsubsection*{Asymptotics of $\bfGamma_{n-j}(\aaa_k;n)$}

From \cite[Formula (92)]{Baik-Wang10a} we have for all $a \in (0, \acc)$
\begin{equation} \label{eq:evaluation_of_bfGamma_sub}
  \bfGamma_{n-j}(a; n) = \tilde{C}(a) \tilde{\M}_{j,n}(c(a))(1 + o(1)),
\end{equation}
where
\begin{equation}
    \tilde{C}(a) := -i\sqrt{\frac{2\pi}{n}} e^{-n\ell/2} \frac{e^{n\Hfn(c(a);a)}}{\sqrt{\Hfn''(c(a);a)}}
\end{equation}
and $\tilde{\M}_{j,n}(z)$ is a generalization to $\tilde{\M}_j(z)$ in~\eqref{eq:Mjindepntilde} when the 
support of the equilibrium measure is a multi-interval (i.e. $N>0$). 
The function $\tilde{\M}_{j,n}(z)$ can be found in terms the solution to a global Riemann--Hilbert problem 
in the analysis of orthogonal polynomials and is given explicitly in terms of a Riemann theta function, see \cite[Formula (312)]{Baik-Wang10a}. 
Unless $N=0$, $\tilde{\M}_{j,n}(z)$ depends on $n$. However it is uniformly bounded in $n$, together with its derivatives, in any compact subset of $z$.

\subsubsection*{Asymptotics of $\bar{\cE}_{n-j+1, n}(\aaa_k; \Int; s)$}

Recall from the usual invariant ensemble theory  that (see \eg, \cite{Tracy-Widom98})
\begin{equation} \label{eq:the_generating_function_well_known}
  \cE_{n-j+1,n}(E; s) = \det(1 - s\chi_E K_{n-j+1, n} \chi_E),
\end{equation}
for any set $E\subset\realR$, 
where $K_{n-j+1, n}$ is the standard Christoffel--Darboux kernel for the weight $e^{-nV}$ (see \eg, \cite[Formula (69)]{Baik-Wang10a}). 
For $E=\Int$, the asymptotic result on invariant ensemble implies that 
$\chi_{\Int} K_{n-j+1, n} \chi_{\Int}$ converges to $\chi_{[T,\infty)}K_{\Airy}\chi_{[T,\infty)}$ in trace norm for each fixed $j$ (\cite{Deift-Kriecherbauer-McLaughlin-Venakides-Zhou99}, \cite{Deift-Gioev07a}, see also \cite[Corollary 6.3]{Baik-Wang10a}). Thus
\begin{equation} \label{eq:null_TW_prob}
  \begin{split}
     \det(1 - s\chi_{\Int}K_{n-j,n}\chi_{\Int}) 
    =  \det(1 - s\chi_{[T,\infty)}K_{\Airy}\chi_{[T,\infty)}) (1+ o(1))
    \end{split}
  \end{equation}
and by \eqref{eq:defn_of_TW_distr_s_version}
\begin{equation} \label{eq:null_TW_prob001}
     \cE_{n-j+1,n}(\Int;s) = F_{\TW}(T;s)(1+o(1)).
  \end{equation}
  
For the rank $1$ case, the analogue of~\eqref{eq:the_generating_function_well_known} is
(see \cite[Formula (73)]{Baik-Wang10a} (Only the $s=1$ case is given in \cite{Baik-Wang10a} but the same proof works for general $s\neq 1$.))
\begin{equation}
  \cE_{n-j+1,n}(a;E;s) = \det(1-s\chi_E \tilde{K}_{n-j+1,n}\chi_E)
\end{equation}
where $\tilde{K}_{n-j+1,n}=K_{n-j,n}+ \tilde{\psi}_{n-j}\otimes \psi_{n-j}$.
Here $\psi_{n-j}$ is the orthogonal polynomial times $e^{-\frac{n}2 V}$
and 
\begin{equation}
	\tilde{\psi}_{\ell} (x)= \tilde{\psi}_{\ell}(x;a;n) 
	:= \frac1{\bfGamma_\ell(a;n)} \bigg( e^{n(ax-V(x)/2)} - \int_{\R} K_{\ell, n}(x,y) e^{ay-V(y)} dy \bigg).
\end{equation}
This implies that, for $E\subset \R$, if $1-s\chi_E K_{n-j,n}\chi_E$ is invertible, 
\begin{multline} \label{eq:handy_formula_of_Prob_with_s}
  	\cE_{n-j+1,n}(a;E;s) = \det(1 - s\chi_E K_{n-j,n}\chi_E) [1 - s\langle \tilde{\psi}_{n-j}, \chi_E \psi_{n-j} \rangle \\
  - s^2 \langle (1-s\chi_E K_{n-j,n}\chi_E)^{-1} \chi_E K_{n-j,n} \chi_E \tilde{\psi}_{n-j}, \chi_E \psi_{n-j} \rangle].
\end{multline}
For $E=\Int$, $1-s\chi_E K_{n-j,n}\chi_E$ is invertible for all $s$ close enough to $s=1$. (This is because $\chi_E K_{n-j,n}\chi_E$ converges to $\chi_{[T, \infty)} K_{\Airy}\chi_{[T, \infty)}$ in operator norm when $E=\Int$ and 
since $\chi_{[T, \infty)} K_{\Airy}\chi_{[T, \infty)}$ has its spectrum in $[0,1)$. 
This appears in several places in the subsequence sections and we do not repeat this remark.)
When $E=\Int$ and $a\in (0,\aaa_c)$, the asymptotics of~\eqref{eq:handy_formula_of_Prob_with_s} for $s=1$ was obtained in \cite[Section 3.3]{Baik-Wang10a} 
by analyzing each term on the right-hand side. 
It was shown that both inner products are $O(n^{1/3})$ (see \cite[Formulas (128), (131) and (332)]{Baik-Wang10a}).
Hence from~\eqref{eq:null_TW_prob} 
we find that  
\begin{equation} \label{eq:gap_prob_of_ranl_1_sub}
   \cE_{n-j+1,n}(a;\Int;s)    = \det(1 - s\chi_{[T,\infty)}K_{\Airy}\chi_{[T,\infty)}) (1+ o(1)) = F_{\TW}(T;s)(1+o(1)).
\end{equation}
Combining \eqref{eq:null_TW_prob001} and \eqref{eq:gap_prob_of_ranl_1_sub} we obtain, for $a\in (0, \aaa_c)$,  
\begin{equation}\label{eq:gap_prob_of_ranl_1_sub_bar} 
  \bar{\cE}_{n-j+1,n}(a; \Int; s) = 1+o(1),
\end{equation}
uniformly in $s$ which is close to $1$.

\bigskip

Inserting \eqref{eq:gap_prob_of_ranl_1_sub_bar}, \eqref{eq:evaluation_of_bfGamma_sub} and~\eqref{eq:null_TW_prob001} into the formula~\eqref{eq:alg}, we obtain 
  \begin{equation} \label{eq:ratio_of_det_formula_sub}
      \cE_n(\aaa_1, \cdots, \aaa_{\m}; \Int; s) 
      = \frac{\det[ \tilde{\M}_{j,n}(c(\aaa_k))+ o(1) ]^{\m}_{j,k = 1}}{\det[ \tilde{\M}_{j,n}(c(\aaa_k)) + o(1) ]^{\m}_{j,k = 1}} F_{\TW}(T;s).
  \end{equation}
Since $\tilde{\M}_{j,n}(c(\aaa_k))$ are bounded uniformly in $n$ and the reciprocal of $\det[\tilde{\M}_{j,n}(c(\aaa_k))]$ is bounded uniformly in $n$ by Proposition \ref{prop:non-vanishing}\ref{enu:prop:non-vanishing:a}, 
the ratio of the two determinants in~\eqref{eq:ratio_of_det_formula_sub} is $1+o(1)$. Hence we obtain Theorem~\ref{thm:sub}.


\section{Proof of Theorem~\ref{thm:sup1}: super-critical case 1, separated external sources eigenvalues} \label{sec:proof_of_sub_and_sup_1.1}

Recall that we assume that there exists $p\in \{1, \cdots, \m\}$ such that the positive, distinct and fixed numbers
\begin{equation}\label{eq:sup1Pr}
\begin{split}
	&\text{$\aaa_j>\acc$ for $j=1, \cdots, p$,} \\
	&\text{$\aaa_j<\acc$ for $j=p+1, \cdots, \m$. }
\end{split}
\end{equation}
We assume, without loss of generality, that 
$\aaa_1 > \aaa_2 > \cdots > \aaa_p$. 
Furthermore, we assume that $\aaa_j \notin \mathcal{J}_V$ and $\Gfn''(x_0(\aaa_j)) \neq 0$ for each $j = 1, \cdots, p$, 
where $x_0(a)$ is defined in the paragraph between~\eqref{eq:mathcalJV} and~\eqref{eq:condition_of_V_4}  in Section~\ref{subsec:assumptions_of_V_and_preliminary_notations}.

To use Theorem~\ref{thm:alg}, we need the asymptotics of $\bfGamma_{n-j}(\aaa_k;n)$, $\cE_{n-j+1,n}(\aaa_k;E;s)$, and $\cE_{n-j+1}(E;s)$ 
for $E=\Int$ and $E=\Intx(x_*)$ with $x_*>\redge$. 

\subsubsection*{Asymptotics of $\bfGamma_{n-j}(a;n)$}

By Baik and Wang \cite[Formula (93)]{Baik-Wang10a}, we find that if $\acc < a < \frac{1}{2}V'(\redge)$ and $a \not\in \mathcal{J}_V$, then
  \begin{equation} \label{eq:evaluation_of_bfGamma_super}
    \bfGamma_{n-j}(a;n) = C(a) \M_{j,n}(x_0(a)) (1 + o(1)),
  \end{equation}
  where
  \begin{equation} \label{eq:evaluation_of_bfGamma_super_C}
    C(a) = \sqrt{\frac{2 \pi}{n}} e^{-n\ell/2}\frac{e^{n \Gfn(x_0(a);a)}}{\sqrt{-\Gfn''(x_0(a);a)}}
  \end{equation}
and $\M_{j,n}(z)$ is a generalization of $\M_j(z)$ in~\eqref{eq:Mjindepn} when the 
support of the equilibrium measure is multi-interval (\ie, $N>0$). 
This again can be expressed explicitly in terms of a Riemann theta function;
see \cite[Formula (311)]{Baik-Wang10a}. 
Note that $\M_{j,n}(z)$ depends on $n$ but is uniformly bounded in $n$, together with its derivatives, in any compact subset in $z$. 
For $a>\frac12 V'(\redge)$, we have the asymptotics \cite[Formula (188)]{Baik-Wang10a} of $\bfGamma_{n-j}(a;n)$ which is an intermediate step toward the formula~\eqref{eq:evaluation_of_bfGamma_super}.
It is easy to further compute the formula \cite[Formula (188)]{Baik-Wang10a} asymptotically (using the asymptotics of $\varphi_{n-j}(y)$) and we find that~\eqref{eq:evaluation_of_bfGamma_super} also holds for $a>\frac12 V'(\redge)$. 
The same applies to the case when $a= \frac12 V'(\redge) > \acc$ from the remark in the first paragraph of Baik and Wang \cite[Section 5]{Baik-Wang10a}. In conclusion, the formula \eqref{eq:evaluation_of_bfGamma_super} is  valid for all $a > \acc$ and $a \not\in \mathcal{J}_V$.

\subsubsection*{Asymptotics of $\cE_{n-j+1,n}(\aaa_k;\Intx(x_*);s)$ and $\cE_{n-j+1}(\Intx(x_*);s)$ for $x_*>\redge$.}

Since $x_*>\redge$, the Christoffel--Darboux kernel restricted on $\Intx(x_*)$ converges to $0$ rapidly (see \eg\ \cite[Formula (346)]{Baik-Wang10a}.) Hence we find from~\eqref{eq:the_generating_function_well_known} that 
\begin{equation} \label{eq:null_TW_prob_stretched001}
  	\cE_{n-j+1, n}(\Intx(x_*);s)  = 1 + o(1)
\end{equation}
for all $x_*>\redge$, $T \in \realR$ and $s$ close to $1$.

We now evaluate $\cE_{n-j+1,n}(a; \Intx(x_*); s)$ for $a= \aaa_j$. 
From the assumption~\eqref{eq:sup1Pr}, there are two cases. 
The first is when $a>\acc$ and $a\notin\mathcal{J}_V$ and the second is when $a<\acc$.
The formula~\eqref{eq:handy_formula_of_Prob_with_s} is the starting point.

Let $a>\acc$ and satisfy $a\notin\mathcal{J}_V$. 
For $x_*=x_0(a)$, 
the asymptotics \cite[Formula (137)]{Baik-Wang10a} 
implies that $\langle \tilde{\psi}_{n-j}, \chi_{\Intx(x_0(a))} \psi_{n-j} \rangle$  
equals $1-\erf(T)+o(1)$ where $\erf(T)$ is the \cdf\ of the normal distribution~\eqref{eq:erfdef}. 
The estimates \cite[Formulas (139) and (333)]{Baik-Wang10a} implies that the other inner product is $o(1)$. 
Therefore, we find that 
\begin{equation} \label{eq:gap_prob_of_ranl_1_super}
  	\cE_{n-j+1,n}(a; \Intx(x_0(a)); s) = 1 - s(1-\erf(T)) + o(1)
\end{equation}
uniformly in $s$ close to $1$.
The estimates \cite[Formulas (137) and (139)]{Baik-Wang10a} can be extended straightforwardly to the set $\Intx(x_*)$ for $x_*$ not equal to $x_0(a)$ but still in $(\redge, \infty)$. 
Hence we obtain 
\begin{equation} \label{eq:gap_prob_of_ranl_1_super001}
  	\cE_{n-j+1,n}(a; \Intx(x_*); s) = 
	\begin{cases} 1  + o(1), \qquad &x_*>x_0(a), \\
	1-s+o(1), \qquad  &x_*\in (\redge, x_0(a)),
	\end{cases}
\end{equation}
uniformly in $s$ close to $1$. This is what is expected from~\eqref{eq:gap_prob_of_ranl_1_super} by taking $T=\infty$ for the first case and taking $T=-\infty$ the second case. 
Recall that these asymptotics are for $a>\acc$ such that $a\notin\mathcal{J}_V$.
The asymptotics~\eqref{eq:gap_prob_of_ranl_1_super} and \eqref{eq:gap_prob_of_ranl_1_super001} apply to $\aaa_1, \cdots, \aaa_p$.

On the other hand, for $a<\acc$, 
an estimate similar to~\eqref{eq:gap_prob_of_ranl_1_sub}  implies that 
(note that $\Intx(x_*)\subset\Int$ for any $x_*>\redge$,~\eqref{eq:gap_prob_of_ranl_1_sub} and $F_{\TW}(T;s) \to 1$ as $T \to \infty$) 
\begin{equation} \label{eq:easy_estimate_by_monoticity}
  \cE_{n-j+1,n}(a;\Intx(x_*);s) = 1 + o(1)
\end{equation}
for $s$ close to $1$. 
This asymptotics applies to $\aaa_{p+1}, \cdots, \aaa_\m$.


\bigskip

Now inserting the asymptotics~\eqref{eq:evaluation_of_bfGamma_sub},~\eqref{eq:evaluation_of_bfGamma_super},~\eqref{eq:null_TW_prob_stretched001},~\eqref{eq:gap_prob_of_ranl_1_super},~\eqref{eq:gap_prob_of_ranl_1_super001} and~\eqref{eq:easy_estimate_by_monoticity} 
into~\eqref{eq:alg}, we obtain, for each $k=1, \cdots, p$,  
\begin{equation} \label{eq:ratio_of_det_formula_super_1}
  	\cE_n(\aaa_1, \cdots, \aaa_{\m}; \Intx(x_0(\aaa_k)); s) 
  	= \frac{(1-s+s\erf(T)) (1-s)^{k-1} 
	\det[ \frakP + o(1)]}{\det[ \frakP + o(1)]},
\end{equation}
where
\begin{equation} \label{eq:defn_of_frakP_simple-super1}
	\frakP
	:= \begin{bmatrix}
	\M_{1,n}(x_0(\aaa_1)) &\cdots & \M_{1,n}(x_0(\aaa_p)) &
	\tilde{\M}_{1,n}(c(\aaa_{p+1})) &\cdots & \tilde{\M}_{1,n}(c(\aaa_{\m}))  \\ 
	\vdots & \ddots & \vdots & \vdots & \ddots &\vdots \\
	\M_{\m,n}(x_0(\aaa_1)) &\cdots & \M_{\m,n}(x_0(\aaa_p)) &\tilde{\M}_{\m,n}(c(\aaa_{p+1})) &\cdots & \tilde{\M}_{\m,n}(c(\aaa_{\m}))  \\ 
	\end{bmatrix}.
\end{equation}
The matrix $\frakP$ is the $\m \times \m$ matrix defined in~\eqref{eq:defn_of_frakP_simple} up to column changes
and hence the reciprocal of $\det[\frakP]$ is bounded uniformly in $n$ by Proposition \ref{prop:non-vanishing}\ref{enu:prop:non-vanishing:a}. 
Since the entries of $\frakP$ are bounded uniformly in $n$, 
we find that 
\begin{equation} \label{eq:genrating_function_separate_ext_eigenvalue_case}
  \cE_n(\aaa_1, \cdots, \aaa_{\m}; \Intx(x_0(\aaa_k)); s) = (1-s+s\erf(T)) (1-s)^{k-1} + o(1)
\end{equation}
uniformly in $s$ close to $1$. 

For each $j \geq 0$,
 \begin{equation} \label{eq:derivative_formula_of_genrating_function_separate_ext_eigenvalue_case}
   \left. \frac{(-1)^j}{j!} \frac{d^j}{ds^j}\right\rvert_{s=1} (1-s+s\erf(T))(1-s)^{k-1}  =
   \begin{cases}
     0 & \textnormal{if $j \neq k-1, k$,} \\
     \erf(T) & \textnormal{if $j = k-1$,} \\
     1 - \erf(T) & \textnormal{if $j = k$.} 
   \end{cases}
 \end{equation}
Since the left-hand side of \eqref{eq:genrating_function_separate_ext_eigenvalue_case} is analytic in $s$, we obtain~\eqref{eq:thmsup1} by taking derivatives and using~\eqref{eq:jthlarge}.

\subsubsection*{Asymptotics of $\cE_{n-j+1,n}(\aaa_k;\Int;s)$ and $\cE_{n-j+1}(\Int;s)$ }

To prove~\eqref{eq:thmsup2}, we repeat the above computation with $\Intx(x_*)$ replaced by 
$\Int$. 
Let $a>\acc$ and assume $a \notin \mathcal{J}_V$.
Using asymptotics \cite[Formulas (330) and (331)]{Baik-Wang10a} of $\psi$, asymptotics \cite[Formulas (106) and (135)]{Baik-Wang10a} of $\tilde{\psi}$ and asymptotics \cite[Corollary 6.3]{Baik-Wang10a} of $K_{n-j,n}$, by \eqref{eq:handy_formula_of_Prob_with_s} we have that for $s$ close to $1$ (cf.~\eqref{eq:gap_prob_of_ranl_1_super001} with $x_*<x_0(\redge)$)
\begin{equation}
  \cE_{n-j+1,n}(a; \Int; s) = \det(1 - s\chi_{\Int} K_{n-j,n}\chi_{\Int}) (1 - s + o(1)) = F_{\TW}(T;s)  (1 - s + o(1)).
\end{equation}
Hence by \eqref{eq:null_TW_prob001} we have for $s$ close to $1$
\begin{equation} \label{eq:bar_cE_super_but_near}
  \bar{\cE}_{n-j+1,n}(a; \Int; s)=1-s+o(1).
\end{equation}
This asymptotics and \eqref{eq:evaluation_of_bfGamma_super} are for $a=\aaa_1, \cdots, \aaa_p$. 

For $a=\aaa_{p+1}, \cdots, \aaa_\m$ which are all less than $\acc$, we can use the asymptotics~\eqref{eq:evaluation_of_bfGamma_sub} and~\eqref{eq:gap_prob_of_ranl_1_sub_bar}. 

Inserting them into~\eqref{eq:alg}, we obtain for $s$ close to $1$
\begin{equation} \label{eq:ratio_of_det_formula_super_1_123}
 \cE_n(\aaa_1, \cdots, \aaa_{\m}; \Int; s) 
  = \frac{ (1-s)^{p} \det[ \frakP  + o(1)]}{\det[ \frakP + o(1)]} F_{\TW}(T;s) 
  =  (1-s)^p F_{\TW}(T;s) + o(1)
 \end{equation}
where $\frakP$ is same as~\eqref{eq:defn_of_frakP_simple-super1}.
We obtain~\eqref{eq:thmsup2} by taking derivatives on both sides of \eqref{eq:ratio_of_det_formula_super_1_123}.


\section{Proof of Theorem~\ref{thm:sup2}: super-critical case 2, clustered external source eigenvalues}\label{subsection:multiple_interwining}

Let $a$ be a fixed number such that $a > \acc$ and $a \notin \mathcal{J}_V$. 
Recall the definition $x_0(a)$ given in the paragraph between~\eqref{eq:mathcalJV} and~\eqref{eq:condition_of_V_4}.
As usual we assume that $\Gfn''(x_0(a)) \neq 0$. Set 
\begin{equation}\label{eq:aaaksca100}
	\aaa_k = a+ \sqrt{-\Gfn''(x_0(a))}\frac{\alpha_k}{\sqrt{n}} , \qquad k=1, \cdots, \m,
\end{equation}
for fixed distinct $\alpha_1, \cdots, \alpha_\m$. 
To prove Theorem~\ref{thm:sup2}, we first evaluate the denominator of~\eqref{eq:alg} asymptotically and then the numerator
when $E=\Intx(x_0(a))$.

\subsection{Evaluation of $\det\big[ \bfGamma_{n-j}(\aaa_k;n)\big]_{j,k=1}^{\m}$}\label{sec:sup2Ga}

The asymptotics of $\bfGamma_{n-j}(a)$ were evaluated in \cite[Formula (93)]{Baik-Wang10a} when $a$ is a constant. It is easy to see from the proof that \cite[Formula (93)]{Baik-Wang10a}  is indeed uniform for $a$ in a compact subset of $(\acc, \infty)$ which is especially applicable when $\aaa_k$ given by~\eqref{eq:aaaksca100}.
It is clear that the leading-order asymptotics of $\bfGamma_{n-j}(\aaa_k)$ are same for all $k=1, \cdots, \m$. 
This implies that the determinant $\det[ \bfGamma_{n-j}(\aaa_k) ]_{j,k=1}^{\m}$ converges to zero.  
Therefore, we need to evaluate the sub-leading terms in the asymptotics of $\bfGamma_{n-j}(\aaa_k)$ in order to determine the asymptotics of $\det[ \bfGamma_{n-j}(\aaa_k) ]_{j,k=1}^{\m}$.

The formula of $\bfGamma_{n-j}(\aaa_k)$ in~\eqref{eq:gadef} is in terms of an integral over $\R$. 
This integral can be written as a sum of two integrals, one over a contour $\Gamma_{\pm}$ in a complex plane and the other over the segment $(c, \infty)$ for any constant $c>\redge$ (see \cite[Formula (85)]{Baik-Wang10a}).
For the case at hand, the integral over the contours $\Gamma_{\pm}$ was shown to be 
exponentially smaller than the integral over $(c,\infty)$, and
the main contribution to the integral over $(c, \infty)$ comes from a small neighborhood of the critical point $x=x_0(a)$. 
Hence from \cite[Formula (90)]{Baik-Wang10a} we have, for any $\epsilon>0$,
\begin{equation}\label{eq:Gammanjak}
\begin{split}
	\bfGamma_{n-j}(\aaa_k)
	&= e^{-n\ell/2} \int_{x_0(a)-\epsilon}^{x_0(a)+\epsilon} M_{j,n}(y)e^{n\Gfn(y;\aaa_k)}dy (1+O(e^{-\delta n}))
\end{split}
\end{equation}
for some $\delta>0$. 
Here $\Gfn(y;a)$ is the function defined in~\eqref{eq:definition_of_GH} and 
$M_{j,n}(z)$ is an analytic function in a neighborhood of $z=x_0(a)$.

By Taylor expansion, 
\begin{equation}\label{eq:Mjerrorz}
	M_{j,n}(z)= \sum_{i=1}^{\m} \frac1{(i-1)!} M_{j,n}^{(i-1)}(x_0(a)) (z-x_0(a))^{i-1} + O(|z-x_0(a)|^{\m})
\end{equation}
uniformly for $z$ in a neighborhood of $x_0(a)$, where $M_{j,n}^{(i-1)}$ is the $(i-1)$th derivative.
As $n\to\infty$, 
$M_{j,n}(z)=\M_{j,n}(z)(1+O(n^{-1}))$ uniformly in $z$ in the same neighborhood
for another analytic function $\M_{j,n}(z)$
which depends on quasi-periodically in $n$. (See \cite[(319)]{Baik-Wang10a}. 
This $\M_{j,n}(z)$ is the same $\M_{j,n}(z)$ appearing in \eqref{eq:evaluation_of_bfGamma_super}.) 
A key property for our purpose is that a certain determinant involving $\M_{j,n}$ and its derivatives is nonzero, which is proved in Proposition~\ref{prop:non-vanishing}\ref{enu:prop:non-vanishing:b} later.
This is used when we consider $\det[\hat{\calP}]$ and $\det[\calP]$ below. 
Note that 
\begin{equation}\label{eq:Mjerrorz0}
	M_{j,n}^{(i-1)}(x_0(a))= \M_{j,n}^{(i-1)}(x_0(a)) (1+O(n^{-1}))
\end{equation}
since both functions are analytic. 
Each of $M_{j,n}^{(i-1)}(x_0(a))$ and $\M_{j,n}^{(i-1)}(x_0(a))$ are $O(1)$. 

Inserting~\eqref{eq:Mjerrorz} into~\eqref{eq:Gammanjak}, we find
\begin{equation}\label{eq:Gammanjak1}
  \bfGamma_{n-j}(\aaa_k) = e^{n\Gfn(x_0(a); \aaa_k) - n\ell/2} 
  \left( \sum_{i=1}^{\m} \frac{\left( -n\Gfn''(x_0(a); a) \right)^{-i/2}}{(i-1)!}  M^{(i-1)}_{j,n}(x_0(a)) Q(i,k) + n^{-\frac{\m+1}{2}} R(j,k) \right),
\end{equation}
where
\begin{equation}\label{eq:Gammanjak1.1}
  Q(i,k):= \left( -n\Gfn''(x_0(a); a) \right)^{i/2} \int_{x_0(a)-\epsilon}^{x_0(a)+\epsilon} (y-x_0(a))^{i-1} e^{n(\Gfn(y;\aaa_k) - \Gfn(x_0(a); \aaa_k))} dy  
\end{equation}
and
\begin{multline}\label{eq:Gammanjak1.2}
	R(j,k) = O\bigg( n^{\frac{\m+1}{2}} \int_{x_0(a)-\epsilon}^{x_0(a)+\epsilon} |y-x_0(a)|^\m e^{n(\Gfn(y;\aaa_k) - \Gfn(x_0(a); \aaa_k))} dy \bigg) \\
	+ \int_{x_0(a)-\epsilon}^{x_0(a)+\epsilon} M_{j,n}(y)e^{n(\Gfn(y;\aaa_k) - \Gfn(x_0(a); \aaa_k))} dy \cdot O(n^{\frac{\m+1}{2}}  e^{-\delta n}).
\end{multline}
From the definition of $\Gfn$ and~\eqref{eq:aaaksca100}, 
\begin{equation}\label{eq:Gammanjak1.3}
  \Gfn(y;\aaa_k)=\Gfn(y;a) + \alpha_k y \sqrt{\frac{-\Gfn''(x_0(a);a)}{n}}.
\end{equation}
Hence Laplace's method yields 
\begin{equation}\label{eq:Gammanjak3}
  	Q(i,k) = \int_{-\infty}^\infty \xi^{i-1} e^{-\frac12 \xi^2 +\alpha_k\xi} d\xi \big(1+o(1)\big).
\end{equation}
Similarly, we find 
\begin{equation}\label{eq:Gammanjak1.2}
	R(j,k) = O\bigg( n^{\frac{\m+1}{2}} \int_{x_0(a)-\epsilon}^{x_0(a)+\epsilon} |y-x_0(a)|^\m e^{n\Gfn(y;\aaa_k)}dy \bigg) = O(1).
\end{equation}

Denote the $\m \times \m$ matrices
\begin{equation}
  \hat{\calP} = \left[ M^{(i-1)}_{j,n}(x_0(a)) \right]^{\m}_{j,i = 1}, \quad \hat{\calQ} = \left[ Q(i,k) \right] ^{\m}_{i,k = 1}, \quad \calR = \left[ R(j,k) \right]^{\m}_{j,k = 1}.
\end{equation}
Note that all entries of these matrices are $O(1)$. 
We also set $\calN$ to be an $\m \times \m$ diagonal matrix with entries 
\begin{equation}
  (\calN)_{ii} = \frac{1}{(i-1)!} \left( -n\Gfn''(x_0(a); a) \right)^{-i/2}.
\end{equation}
 From \eqref{eq:Gammanjak1}, we have
\begin{equation}\label{eq:invhacaQP}
\begin{split}
  \det \left[ \ga_{n-j}(\aaa_k;n) \right]^\m_{j,k=1} 
  = & \left( \prod^{\m}_{k=1} e^{n\Gfn(x_0(a); \aaa_k) - n\ell/2} \right) \det[ \hat{\calP} \calN \hat{\calQ} + n^{-\frac{\m+1}{2}}\calR ] \\
  = & \left( \prod^{\m}_{k=1} e^{n\Gfn(x_0(a); \aaa_k) - n\ell/2} \det[ \calN ] \right) \det[ \hat{\calP} \hat{\calQ} + n^{-\frac{\m+1}{2}} \hat{\calP}\calN^{-1}\hat{\calP}^{-1} \calR ]
\end{split}
\end{equation}
if $\hat{\calP}$ is invertible.

We now replace $\hat{\calP}$ and $\hat{\calQ}$ by matrices with entries given by the leading terms given in~\eqref{eq:Mjerrorz0} and~\eqref{eq:Gammanjak3}.
Define the $\m \times \m$ matrix
\begin{equation}\label{eq:calPPa}
  	\calP = \left[ \M^{(i-1)}_{j,n}(x_0(a)) \right]^{\m}_{j,i = 1}.
\end{equation}
The entries of this matrix are the leading term of the entries of the matrix $\hat{\calP}$ (see~\eqref{eq:Mjerrorz0}).
From the general result Proposition~\ref{prop:non-vanishing}\ref{enu:prop:non-vanishing:b} 
and  by noting that 
$\calP=\frakP^{(x_0(a), \m)}$ in the notation of Section~\ref{sec:proof_of_non_vanishing_property},
we find that
$\det[\calP]$ is nonzero.
Moreover, $1/\det[\calP]$, which depends on $n$, is uniformly bounded. 
This nonvanishing property is easy to check directly using the formula~\eqref{eq:Mjindepn} when $N=0$ but is complicated when $N>0$. 
Now as $\hat{\calP}= \calP + o(1)$, we find that all entries of $\hat{\calP}^{-1}$ are $O(1)$. 
Hence noting that the explicit dependence on $n$ of $\calN$, we find that 
all entries of $n^{-\frac{\m+1}{2}} \hat{\calP}\calN^{-1}\hat{\calP}^{-1} \calR$ are $O(n^{-1/2})$. 

We also define
\begin{equation}\label{eq:calPPa-1}
	\calQ = \left[ \int_{-\infty}^\infty \xi^{i-1} e^{-\frac12 \xi^2 +\alpha_k\xi} d\xi \right]^{\m}_{i,k=1}.
\end{equation}
Then $\hat{\calQ} = \calQ + o(1)$. 
Therefore, we find from~\eqref{eq:invhacaQP} that 
\begin{equation} \label{eq:det_of_bfGamma_super_cluster}
  \det \left[ \ga_{n-j}(\aaa_k;n) \right]^\m_{j,k=1} = \prod^{\m}_{k=1} e^{n\Gfn(x_0(a); \aaa_k) - n\ell/2} \det[ \calN ](\det[ \calP ]\det[ \calQ ] + o(1)).
\end{equation}
It is straightforward to check that 
\begin{equation}
  \det[ \calQ ] = \int_{\R^n} \det [ e^{\alpha_k\xi_j}]  \prod_{i<j} (\xi_j-\xi_i) \prod_{j=1}^n e^{-\frac12 \xi_j^2} d\xi_1 \cdots d\xi_{\m}
\end{equation}
and this is nonzero when all $\alpha_k$'s are distinct.

\subsection{Evaluation of $\det\big[ \bfGamma_{n-j}(\aaa_k;n) \bar{\cE}_{n-j+1,n} (\aaa_k;\Intx(x_0(a));s)\big]_{j,k=1}^{\m}$}  
\label{sec:sup2e}

Note that since $\aaa_k> \acc$, the asymptotics~\eqref{eq:null_TW_prob_stretched001} applies. 
Hence from the definition~\eqref{eq:defn_of_bar_E_a_dots_E_s_beginning} of $\bar{\cE}$,
we find that 
\begin{equation} \label{eq:supersim}
  \begin{split}
    & \det \left[ \bfGamma_{n-j}(\aaa_k) \bar{\cE}_{n-j+1,n} (\aaa_k; \Intx(x_0(a)); s) \right]_{j,k=1}^{\m} \\
    &= \det \left[ \bfGamma_{n-j}(\aaa_k) \frac{\cE_{n-j+1,n} (\aaa_k; \Intx(x_0(a)); s)}{\cE_{n-j,n}(\Intx(x_0(a)); s)} \right]_{j,k=1}^{\m} 
    \prod_{j=1}^{\m} \frac{\cE_{n-j,n}(\Intx(x_0(a)); s)}{\cE_{n-j+1,n}(\Intx(x_0(a)); s)} \\
    &= \det \left[ \bfGamma_{n-j}(\aaa_k) \frac{\cE_{n-j+1,n} (\aaa_k; \Intx(x_0(a)); s)}{\cE_{n-j,n}(\Intx(x_0(a)); s)} \right]_{j,k=1}^{\m} (1 + o(1)).
  \end{split}
\end{equation}
We focus on the new determinant. 
As in the previous subsection, the determinant converges to zero and hence we need to find the leading asymptotics.


From \eqref{eq:handy_formula_of_Prob_with_s}, and \eqref{eq:the_generating_function_well_known}, we find
\begin{multline} \label{eq:algebraic_trans_form_of_Prob(aaa_k)}
    \frac{\cE_{n-j+1, n}(\aaa_k; E; s) }{\cE_{n-j,n}(E; s)}= 
    1 - s\langle \tilde{\psi}_{n-j}(x;\aaa_k;n), \psi_{n-j} \rangle_{E} \\
    - s^2\langle (1 - s\chi_{E} K_{n-j,n} \chi_{E})^{-1} 
\chi_{E} K_{n-j,n} \chi_{E} \tilde{\psi}_{n-j}(x;\aaa_k;n), \psi_{n-j} \rangle_{E}, 
\end{multline}
with $E=\Intx(x_0(a))$. 
Note that the interval $E=\Intx(x_0(a))$ is associated to $a$ and $\tilde{\psi}_{n-j}$ is associated to $\aaa_k$. 
When $a$ and $\aaa_k$ are the same and larger than $\redge$, the second inner product was shown to be exponentially small in  \cite{Baik-Wang10a}: precisely in \cite[Formula (139)]{Baik-Wang10a} assuming $a< \frac12V'(\redge)$ and in \cite[Sections 4 and 5]{Baik-Wang10a} when $a>\frac12 V'(\redge)$ and $a=\frac12 V'(\redge)$.
In our case, $\aaa_k$ and $a$ are different, but note that the difference is $O(n^{-1/2})$ by assumption~\eqref{eq:aaaksca100}. 
This corresponds to a small change in the domain of the inner product which does not change the exponential decay of the inner product. 
Thus we have
\begin{equation} \label{eq:complicated_term_equal_e^delta_prime}
  \frac{\cE_{n-j+1, n}(\aaa_k; \Intx(x_0(a)); s)}{\cE_{n-j,n}(\Intx(x_0(a)); s)} = 
  1 - s\langle \tilde{\psi}_{n-j}(x;\aaa_k;n), 
  \psi_{n-j}(x;n) \rangle_{\Intx(x_0(a))} + O(e^{-\delta' n}).
\end{equation}
for some $\delta'>0$ uniformly for $s$ close to $1$.

Now we evaluate the remaining inner product in~\eqref{eq:complicated_term_equal_e^delta_prime}.
We actually evaluate the inner product multiplied by $\bfGamma_{n-j}(\aaa_k;n)$, which is 
what we need in view of~\eqref{eq:algebraic_trans_form_of_Prob(aaa_k)}.
The leading-order asymptotics of this quantity was evaluated in \cite[Section 3.4]{Baik-Wang10a}.
Here we need the sub-leading terms
and this follows from a simple extension of the analysis for the leading term as follows. 
First, for all $x \in \Intx(x_0(a))$, we have from \cite[Formulas (106) and (330)]{Baik-Wang10a}
that 
\begin{equation} \label{eq:exponential_approx_of_tilde_psi_psi}
  \bfGamma_{n-j}(\aaa_k;n) \tilde{\psi}_{n-j}(x;\aaa_k;n)\psi_{n-j}(x;n) 
  = e^{-n\ell/2}M_{j,n}(x)e^{n\Gfn(x;\aaa_k)} (1 + O(e^{-\delta''n})).
\end{equation}
for some $\delta'' > 0$. 
This is same as \cite[Formula (136)]{Baik-Wang10a} (after substituting the asymptotics of $\bfGamma_{n-j}$) where the error term is  written only as $o(1)$ instead of an exponentially small term. 
Recalling that $\Gfn(x;a)$, $x\in (\redge, \infty)$, takes its unique maximum at $x=a$ by the assumption $a>\acc$ and $a\notin\mathcal{J}_V$, 
and noting that  $\Gfn(x;\aaa_k)$ is close to $\Gfn(x;a)$ (see~\eqref{eq:Gammanjak1.3}), we find that 
for any $\epsilon > 0$
\begin{equation}\label{eq:Inner_product_of_tilde_psi_psi_bfGamma}
  \begin{split}
    	& \bfGamma_{n-j}(\aaa_k;n) \langle \tilde{\psi}_{n-j}(x;\aaa_k;n), 
  \psi_{n-j}(x;n) \rangle_{\Intx(x_0(a))}  \\
     &= \left( e^{-n\ell/2} \int_{\Intx(x_0(a))} M_{j,n}(y)e^{n\Gfn(y;\aaa_k)}dy \right) (1+O(e^{-\delta'' n})) \\
     &= \left( e^{-n\ell/2} \int_{E_{T, \epsilon}(x_0(a))} M_{j,n}(y)e^{n\Gfn(y;\aaa_k)}dy ) \right) (1+O(e^{-\delta''' n})) 
  \end{split}
\end{equation}
for some $\delta''' > 0$
where $E_{T, \epsilon}(x_0(a))$ is the interval 
\begin{equation}
	E_{T, \epsilon}(x_0(a)) := 
	\left(x_0(a) + \frac{T}{\sqrt{-\Gfn''(x_0(a))n}} ,\ x_0(a)+\epsilon \right).
\end{equation}

We now find from~\eqref{eq:complicated_term_equal_e^delta_prime},~\eqref{eq:Gammanjak}, and~\eqref{eq:Inner_product_of_tilde_psi_psi_bfGamma} that 
\begin{multline} \label{eq:Gammanjak11}
  \bfGamma_{n-j}(\aaa_k) \frac{\cE_{n-j+1,n} (\aaa_k; \Intx(x_0(a)); s)}{\cE_{n-j,n}(\Intx(x_0(a)); s)} =  \\
  e^{-n\ell/2} \int^{x_0(a)+\epsilon}_{x_0(a)-\epsilon} M_{j,n}(y)e^{n\Gfn(y;\aaa_k)} ( 1 - s\chi_{E_{T, \epsilon}(x_0(a))}(y)) dy (1+O(e^{-\tilde{\delta} n})),
\end{multline}
where $\tilde{\delta} = \min\{ \delta, \delta''' \} $. 
This formula is completely analogous to~\eqref{eq:Gammanjak} in the previous subsection
except that there is the term $( 1 - s\chi_{E_{T, \epsilon}(x_0(a))}(y))$ in the integrand. 
We can now proceed exactly as in the previous subsection to evaluate the determinant~\eqref{eq:supersim}. 
The new term in the integrand only changes that $Q(i,k)$ in~\eqref{eq:Gammanjak1.1} contains 
the term $(1 - s\chi_{(T, \infty)}(\xi))$ in the integrand. 
Therefore, we obtain, similarly to~\eqref{eq:det_of_bfGamma_super_cluster}, 
\begin{multline} \label{eq:calculation_of_almost_top_clustered_normal}
  \det\bigg[ \bfGamma_{n-j}(\aaa_k) \frac{\cE_{n-j+1,n} (\aaa_k, \Intx(x_0(a); s)}{\cE_{n-j,n}(\Intx(x_0(a)); s)} \bigg]^{\m}_{j,k = 1} \\
  = \prod^{\m}_{k=1} e^{n\Gfn(x_0(a); \aaa_k) - n\ell/2} \det[ \calN ](\det[ \calP ]\det[ \calQ_{T;s} ] + o(1)),
\end{multline}
where $\calQ_{T;s}$ is the $\m \times \m$ matrix with entries
\begin{equation}
  \calQ_{T;s}(i,k)= \int^{\infty}_{-\infty} \xi^{i-1} e^{-\frac12 \xi^2 +\alpha_k\xi} (1 - s\chi_{(T, \infty)}(\xi)) d\xi .
\end{equation}

\bigskip

Combining~\eqref{eq:det_of_bfGamma_super_cluster}, \eqref{eq:supersim}, \eqref{eq:calculation_of_almost_top_clustered_normal} and Theorem \ref{thm:alg}, we find that (recall~\eqref{eq:defination_of_spiked_GUE_dist_rank_k} for the definition of $\erf_k$)
\begin{equation}\label{eq:Gammanjak14}
  \begin{split}
    \cE_n(\aaa_1, \cdots, \aaa_\m; \Intx(x_0(a)); s) = & \frac{\det[ \calQ_{T;s} ]}{\det[ \calQ ]} (1+o(1)) = \erf_k(T; \alpha_1, \cdots, \alpha_k; s) + o(1).
  \end{split}
\end{equation}
Hence Theorem \ref{thm:sup2} is proved.

\section{Proof of Theorem~\ref{thm:superjump}: secondary critical case}\label{sec:superjump}

We assume that the support of the equilibrium measure associated to $V$ consists of one interval. 
Let $a\in\mathcal{J}_V \setminus \{ \acc \}$. 
Then $\Gfn(x; a)$ attains its maximum in $(c(a), \infty)$ at more than one point. 
We assume that the maximum is achieved at two points, which we 
denote by  $x_1(a) < x_2(a)$. 
We write $x_1(a)$ as $x_1$ and $x_2(a)$ as $x_2$ for notational convenience if there is no confusion. 
Set 
\begin{equation}\label{eq:Gmax}
	\Gfn_{\max}:= \Gfn(x_1; a)=\Gfn(x_2;a).
\end{equation}
We assume, as usual, that $\Gfn''(x_1; a)\neq 0$ and $\Gfn''(x_2; a)\neq 0$. 

Throughout this section, we fix $m \in \{1, 2, \cdots, \m\}$. 
Recall the definitions 
\begin{equation}\label{eq:superj005-sec5}
\begin{split}
  	q_\mm := \frac{\m-2\mm+1}{x_2(a)-x_1(a)},
	\qquad 
	K_\mm := \left( 
        \frac{(\m-m)!}{(m-1)!} \frac{(-\Gfn''(x_1(a)))^{\m-m+1/2} }{(-\Gfn''(x_2(a)))^{m-1/2}}
	 \right)^{\frac{1}{\m-2m+1}}.
\end{split}
\end{equation}
Set 
\begin{equation}
  a':= a- q_m \frac{\log (K_m n)}{n},
\end{equation}
and we assume that 
\begin{equation}\label{eq:superj000-sec5}
	\aaa_k=  a- q_m\frac{\log (K_\mm n)}{n}+ \frac{\alpha_k}{n}
	= a' + \frac{\alpha_k}{n}, \qquad k = 1, \cdots, \m,
\end{equation}
for fixed distinct $\alpha_1 > \cdots > \alpha_\m$.

\subsection{Evaluation of $\det\big[ \bfGamma_{n-j}(\aaa_k;n)\big]_{j,k=1}^{\m}$}\label{sec:supjump1}

The goal here is to  prove the asymptotic formula~\eqref{eq:super228.7}  given at the end of this subsection.

Analysis in this section is similar to that of Section~\ref{sec:sup2Ga} but with the change that the main 
contribution to the integral formula of $\bfGamma_{n-j}(\aaa_k;n)$ comes from 
two intervals (near $x_1(a)$ and $x_2(a)$) instead of one interval as in~\eqref{eq:Gammanjak}.
This is because of~\eqref{eq:Gmax}.
We have a small enough $\epsilon > 0$ and a corresponding $\delta > 0$ such that
\begin{equation}\label{eq:super201}
	e^{n\ell/2} \bfGamma_{n-j}(\aaa_k;n)
	=\bigg(  \int_{E_1} M_{j,n}(y)e^{n\Gfn(y;\aaa_k)}dy  
	+ \int_{E_2} M_{j,n}(y)e^{n\Gfn(y;\aaa_k)}dy \bigg) (1+O(e^{-\delta n})).
\end{equation}
where
\begin{equation} \label{eq:defn_of_E_epsilon}
  E_1 := (x_1 - \epsilon, x_1 + \epsilon), \quad E_2 := (x_2 - \epsilon, x_2 + \epsilon).
\end{equation}
Since the two integrals are asymptotically of same order, the evaluation of the determinant $\det\big[ \bfGamma_{n-j}(\aaa_k;n)\big]_{j,k=1}^{\m}$ is more complicated.

Using the \Andreief's formula in random matrix theory (see \eg\ \cite{Tracy-Widom98}), we have
\begin{equation} \label{eq:super210}
\begin{split}
 & \left( \prod^{\m}_{k=1} e^{n\Gfn(x_1; \aaa_k)} \right)^{-1} \det \left[ \int_{E_1 \cup E_2} M_{j,n}(y) e^{n\Gfn(y;\aaa_k)} dy \right]^{\m}_{j,k = 1} \\
 &= \det \left[ \int_{E_1 \cup E_2} M_{j,n}(y) e^{n(\Gfn(y;\aaa_k)-\Gfn(x_1;\aaa_k))} dy \right]^{\m}_{j,k = 1} \\
  &=  \frac{1}{\m!} \int_{(E_1 \cup E_2)^{\m}} \det[M_{j,n}(y_k)] \det[e^{n(\Gfn(y_k; \aaa_j)-\Gfn(x_1;\aaa_k))}] dy_1 \cdots dy_{\m}.
\end{split}
\end{equation}
For each variable $y_k$, the integral in $y_k$ is over  $E_1 \cup E_2$. 
Using the symmetry of the integrand  in $y_k$ in the last line of~\eqref{eq:super210} is symmetric in $y_k$,~\eqref{eq:super210} equals 
\begin{equation}\label{eq:super210.5.5}
	\frac1{\m!}\sum_{\ell=0}^{\m}  \binom{\m}{\ell} \Ii_\ell
\end{equation}
where, for $\ell=0, \cdots, \m$, 
\begin{equation}\label{eq:super210.5}
  \Ii_\ell := 
  \int_{E^{\m-\ell}_1} dy_1 \cdots dy_{\m-\ell} \int_{E^{\ell}_2} dy_{\m-\ell+1} \cdots dy_{\m} \det[M_{j,n}(y_k)] \det[e^{n (\Gfn(y_k; \aaa_j) - \Gfn(x_1; \aaa_j))}].
\end{equation}

We now evaluate the leading asymptotics of $\Ii_\ell$ for each $\ell$. 
For $t = (t_1, \cdots, t_j)$, let $\Delta_j(t) := \prod_{1 \leq k < \ell \leq j} (t_{\ell} - t_k) $ 
denote the Vandermonde determinant. 
For each $j$, set
\begin{equation} \label{eq:superj006}
  \bfZ_j := \int_{\realR^j} \lvert \Delta_j(t) \rvert^2 \prod^j_{k=1} e^{-\frac{1}{2}t^2_k} dt_k = (2\pi)^{n/2} \prod^j_{k=1} k!.
\end{equation}
This is the partition function of the $j$-dimensional GUE.
We have the following lemma. 

\begin{lemma}\label{lem:Ii} 
For each $\ell=0,1, \cdots, \m$, we have
\begin{equation}\label{eq:super228}
	\Ii_\ell =
	\frac{\error^{\m-\ell}\errort^ \ell}{((K_{m}n)^{q_{m}(x_2-x_1)})^ \ell} 
	\prod_{k=0}^{\m-1-\ell} \bigg(\frac{\error^k}{k!} \bigg)^2  \prod_{k=0}^{\ell-1} \bigg(\frac{\errort^k}{k!} \bigg)^2
	  \cdot \det[\bfP_ \ell] \det[\bfQ_ \ell] 
\bfZ_{\m-\ell} \bfZ_{\ell} (1+o(1))
\end{equation}
where $\bfZ_\ell$ is defined in~\eqref{eq:superj006}, $\bfP_\ell $ and $\bfQ_\ell $ are defined in Theorem \ref{thm:superjump}, and 
\begin{equation}\label{eq:super225.5}
\begin{split}
	\omega_j:= (-n\Gfn(x_j;a))^{-1/2}, \qquad j=1,2.
\end{split}
\end{equation}
\end{lemma}

\begin{proof}
Since $\Gfn(y; \aaa_j)-\Gfn(x_1; \aaa_j)= \Gfn(y; a')-\Gfn(x_1;a') + \frac{\alpha_j}{n}(y-x_1)$ and 
$\Gfn(x_1; a)= \Gfn(x_2;a)$,~\eqref{eq:super210.5} equals 
\begin{equation} \label{eq:super210.5_new}
  \Ii_\ell = \int_{E_1^{\m-\ell} \times E_2^{\ell}} \det[ M_{j,n}(y_k)] \det[ Q_j(y_k)] \prod_{k=1}^{\m} e^{nD(y_k)} dy_k,
\end{equation}
where 
\begin{equation}\label{eq:super208}
\begin{split}
	Q_k(y):= \begin{cases}
	e^{\alpha_k(y-x_1)}, & y\in E_1, \\
	\frac{e^{\alpha_k(x_2-x_1)}}{(K_{m}n)^{q_{m}(x_2-x_1)}}  e^{\alpha_k(y-x_2)}, 
	 &y\in E_2, 
	\end{cases}
\end{split}
\end{equation}
and
\begin{equation}\label{eq:super209}
\begin{split}
	D(y):= \begin{cases}
	\Gfn(y; a')-\Gfn(x_1;a'),  & y\in E_1, \\
	\Gfn(y; a')-\Gfn(x_2;a'),  & y\in E_2. 
	\end{cases}
\end{split}
\end{equation}

Note that first $\ell$ of $y_k$ are in $E_1$ and the rest $y_k$ are in $E_2$. 
Using the Taylor's expansion, we have for $k=1, \cdots, \m-\ell $, 
  \begin{equation}\label{eq:super212}
    M_{j,n}(y_k) =  \sum_{i=1}^{\m-\ell} \frac{M_{j,n}^{(i-1)}(x_1)}{(i-1)!} (y_k-x_1)^{i-1} +O(|y_k-x_1|^{\m-\ell}),  
  \end{equation}
and for $k=\m-\ell +1, \cdots, \m$, 
  \begin{equation}\label{eq:super213}
    M_{j,n}(y_k) =  \sum_{i=1}^{\ell} \frac{M_{j,n}^{(i-1)}(x_2)}{(i-1)!} (y_k-x_2)^{i-1} +O(|y_k-x_2|^{\ell}). 
  \end{equation}
Let $\hat{\bfP}_\ell$ be the matrix defined similar to $\bfP_{\ell}= \frakP^{(x_1(a),\m - \ell),(x_2(a), \ell)}$ 
in the statement of the theorem but with $\M_j$ replaced by $M_{j,n}$ in the entries. 
By Baik and Wang \cite[Proposition 6.1]{Baik-Wang10a} we have that $M_{j,n}^{(i)}(x)= \M_j^{(i)}(x)(1+O(n^{-1})$ and 
they are bounded uniformly in $n$ for any $x$ in a compact subset of $(\redge, \infty)$. (This was also discussed in the previous section in~\eqref{eq:Mjerrorz0}.)
Since $(-1)^{\m(\m-1)/2}\det[\bfP_{\ell}]>0$  and $1/\det[\bfP_{\ell}]$ are bounded uniformly in $n$
(by Proposition \ref{prop:non-vanishing}\ref{enu:prop:non-vanishing:b}; see~\eqref{eq:superj1030}), 
$\det [ \hat{\bfP}_\ell] = \det[\bfP_\ell] (1+o(1))$ and the entries of $\hat{\bfP}^{-1}$ is $O(1)$ . Thus \eqref{eq:super213} implies that 
\begin{equation}\label{eq:super214}
\begin{split}
  \det[ M_{j,n}(y_k)]^{\m}_{j,k = 1} = & \det[\hat{\bfP}_\ell \bfV_\ell +\bfEo ] = \det[\bfP_\ell]  \cdot \det[ \bfV_\ell + \hat{\bfP}^{-1}_\ell\bfEo ]  (1+o(1))
\end{split}
\end{equation}
where 
\begin{equation}\label{eq:super215.5}
\begin{split}
	\bfV_\ell:= \begin{bmatrix} 
	1 & \cdots &1 & 0 &\cdots &0 \\
	(y_1-x_1) & \cdots & (y_{\m-\ell}-x_1) &0 &\cdots &0 \\
	\vdots& \ddots& \vdots& \vdots &\ddots &\vdots\\
	\frac{(y_1-x_1)^{\m-\ell-1}}{(\m-\ell-1)!} & \cdots & \frac{(y_{\m-\ell}-x_1)^{\m-\ell-1}}{(\m-\ell-1)!} & 0 &\cdots &0\\
	0 & \cdots & 0 & 1 & \cdots &1 \\
	0 & \cdots & 0 & 	(y_{\m-\ell+1}-x_2) & \cdots & (y_{\m}-x_2) \\
	\vdots &\ddots&\vdots  & 	\vdots& \ddots& \vdots\\
	0 & \cdots & 0 & 	\frac{(y_{\m-\ell+1}-x_2)^{\ell-1}}{(\ell-1)!} & \cdots & \frac{(y_{\m}-x_2)^{\ell-1}}{(\ell-1)!}  \\
	\end{bmatrix}
\end{split}
\end{equation}
and $\bfEo$ is a matrix with entries satisfying, for each $j=1, \cdots, \m$,
\begin{equation} \label{eq:super216}
  \bfEo_{j,k}= 
\begin{cases}
	O(|y_k-x_1|^{\m-\ell-1}), &k=1, \cdots, \m-\ell, \\
	 O(|y_k-x_2|^{\ell-1}), &k=\m-\ell+1, \cdots, \m.
\end{cases}
\end{equation}

Similarly, we have the Taylor expansions 
\begin{equation}\label{eq:super218}
\begin{split}
	Q_j(y_k)= \sum_{i=1}^{\m-\ell} \frac{\alpha_j^i(y_k-x_1)^i}{i!} + O(|y_k-x_1|^{\m-\ell-1})
\end{split}
\end{equation}
for $k=1, \cdots, \m-\ell$,  and
\begin{equation}\label{eq:super219}
  Q_j(y_k)= \frac{e^{\alpha_k(x_2-x_1)}}{(K_m n)^{q_m(x_2-x_1)}} \bigg(\sum_{i=1}^{\ell} \frac{\alpha_j^i(y_k-x_2)^i}{i!} + O(|y_k-x_2|^{\ell-1}\bigg)
\end{equation}
for $k=\m-\ell+1, \cdots, \m$. 
From the arguments above we find that 
\begin{equation}\label{eq:super220}
  \det[Q_j(y_k)]^{\m}_{j,k =1} = \frac{1}{(K_m n)^{q_m(x_2-x_1)\ell}} \det[\bfQ_\ell] \cdot \det[ \bfV_\ell+ \bfQ^{-1}_\ell\bfEt],
\end{equation}
where $\bfV_\ell$ is same as in~\eqref{eq:super215.5}, $\bfQ_\ell$ is defined in Theorem \ref{thm:superjump} 
and the error matrix $\bfEt$ satisfies the same estimate~\eqref{eq:super216} of $\bfEo$.
Note that $(-1)^{\m(\m-1)/2}\det[\bfQ_{\ell}] > 0$ for $\alpha_1 > \cdots > \alpha_{\m}$.

We now evaluate the integral in $\Ii_\ell$ using the Laplace's method with the change of variables
\begin{equation}\label{eq:super222}
\begin{split}
	y_k= x_1 + \frac{s_k}{\sqrt{-n\Gfn(x_1;a)}}, \qquad & k=1, \cdots, \m-\ell, \\
	y_{\m-\ell+k}= x_2 + \frac{t_k}{\sqrt{-n\Gfn(x_2;a)}}, \qquad & k=1, \cdots, \ell.
\end{split}
\end{equation}
Note that under this change of variables, with the notations $\omega_j$ defined in~\eqref{eq:super225.5}, $\bfV_\ell$ equals the diagonal matrix 
$\diag(1, \error, \cdots, \frac{\error^{\m-\ell-1}}{(\m-\ell-1)!}, 1, \omega_2, \cdots, \frac{\omega_2^{\ell-1}}{\ell-1)!})$
times the matrix
 \begin{equation}\label{eq:super224}
\begin{split}
	\Delta_{\m-\ell, \ell}(s,t):=
	\begin{bmatrix} 
	1 & \cdots &1 & 0& \cdots & 0  \\
	s_1 & \cdots &s_{\m-\ell} & 0& \cdots & 0  \\
	\vdots & \ddots & \vdots & \vdots & \ddots & \vdots\\
	s_1^{\m-\ell-1} & \cdots & s_{\m-\ell}^{\m-\ell-1} & 0& \cdots & 0  \\
	0& \cdots & 0 & 1 & \cdots &1\\
	0& \cdots & 0 & t_1 & \cdots & t_\ell  \\
	\vdots & \ddots & \vdots & \vdots & \ddots & \vdots  \\
	0& \cdots & 0 & t_1^{\ell-1}& \cdots & t_\ell^{\ell-1}
	\end{bmatrix} .
\end{split}
\end{equation}
Also for each $j=1, \cdots, \m$, \eqref{eq:super216} implies that 
\begin{equation}\label{eq:super225}
  \bfEo_{j,k}=  
\begin{cases}
	O( n^{-(\m-\ell)/2} \cdot |s_j|^{\m-\ell}), & k=1, \cdots, \m-\ell, \\
	O( n^{-\ell/2} \cdot |t_j|^{\ell}), & k=\m-\ell+1, \cdots, \m. 
\end{cases}
\end{equation}
Hence
\begin{equation}\label{eq:super226}
	\det[\bfV_\ell +\hat{\bfP}^{-1}_\ell \bfEo]
	=  
	\prod_{k=0}^{\m-\ell-1} \frac{\error^{k} }{k!} \prod_{k=0}^{\ell-1} \frac{\omega_2^k}{k!}
	\cdot \det \bigg[ \Delta_{\m-\ell, \ell}(s,t) 
	+ O(\frac{ \max_k |s_k|^{\m-\ell}+\max_k |t_k|^{\ell}}{\sqrt{n}})
	\bigg].
\end{equation}
The determinant of $\bfV_\ell+ \bfQ^{-1}_\ell \bfEt$ has the same asymptotics. 
Also note that for $k=1, \cdots, \m-\ell$, the term $e^{nD(y_k)}dy_k$ in the integral~\eqref{eq:super210.5_new} becomes, under the change of variables~\eqref{eq:super222}, 
\begin{equation}\label{eq:super227.01}
	e^{nD(y_k)}dy_k = e^{n(\Gfn(y_k;a)-\Gfn(x_1;a)-q_{m} \frac{\log(K_{m} n)}{n} (y-x_1))} dy_k= \error e^{-\frac12 s_k^2 +O(|s_k|^3/\sqrt{n})+O(\frac{\log n}{\sqrt{n}} |s_k|)  } ds_k.
\end{equation}
The term for $k=\m-\ell+1, \cdots, \m$ is also similar. 
Therefore, we find by applying the Laplace's method to~\eqref{eq:super210.5_new} and using~\eqref{eq:super214},~\eqref{eq:super220},~\eqref{eq:super226} that 
\begin{equation}\label{eq:super227}
\begin{split}
	\Ii_\ell
	=&  \frac{\error^{\m-\ell}\errort^\ell}{(K_mn)^{q_m(x_2-x_1)\ell}} 
	\prod_{k=0}^{\m-\ell-1} \bigg( \frac{\error^{k} }{k!}\bigg)^2 \prod_{k=0}^{\ell-1}  \bigg( \frac{\omega_2^k}{k!} \bigg)^2 \cdot \det[\bfP_\ell] \det[\bfQ_\ell] \\
	&\times 
	\bigg[\int_{\R^{\m-\ell}} |\Delta_{\m-\ell}(s)|^2 \prod_{j=1}^{\m-\ell} e^{-\frac12 s_j^2} ds_j  \bigg]
	\bigg[ \int_{\R^{\ell}} |\Delta_\ell(t)|^2  e^{-\frac12 t_k^2}dt_k \bigg] (1+o(1)) \\
        = & \frac{\error^{\m-\ell}\errort^\ell}{(K_mn)^{q_m(x_2-x_1)\ell}} 
	\prod_{k=0}^{\m-\ell-1} \bigg( \frac{\error^{k} }{k!}\bigg)^2 \prod_{k=0}^{\ell-1}  \bigg( \frac{\omega_2^k}{k!} \bigg)^2 \cdot \det[\bfP_\ell] \det[\bfQ_\ell] \cdot \bfZ_{\m-\ell} \bfZ_{\ell} (1+o(1)).
\end{split}
\end{equation}
\end{proof}

We now evaluate the asymptotics of the sum in~\eqref{eq:super210.5.5} using the above lemma. 
Since $\sqrt{n}\omega_i$, $(\sqrt{n}\omega_i)^{-1}$,  $\det[\bfQ_{\ell}] \det[\bfP_{\ell}]$ and its reciprocal are all $ O(1)$, we find from Lemma~\ref{lem:Ii} that 
\begin{equation}\label{eq:super228.1}
  \frac{\Ii_j}{\Ii_k} = O \left( \frac{n^{-\frac{1}{2}(\m^2 - 2\m j + 2j^2) - q_m(x_2 - x_1)j}}{n^{-\frac{1}{2}(\m^2 - 2\m k + 2k^2) - q_m(x_2 - x_1)k}} \right) =  O(n^{(k - m + \frac{1}{2})^2 - (j - m + \frac{1}{2})^2})
\end{equation}
for all $j,k\in \{0, \cdots, \m\}$, 
by using the definition~\eqref{eq:superj000} of $q_m$. 
Hence $\Ii_{m-1}$ and $\Ii_m$ are of same order and the other $\Ii_\ell$ are of smaller orders (at least by factor $n^{2}$).
This implies that~\eqref{eq:super210.5.5} becomes
\begin{equation}\label{eq:super228.4}
	\frac1{\m!} \bigg[ \binom{\m}{m-1} \Ii_{m-1} + \binom{\m}{m} \Ii_m \bigg] (1+o(1)).
\end{equation}
Thus, from~\eqref{eq:super201} and~\eqref{eq:super210}, we have 
\begin{equation}
  \det\big[ \bfGamma_{n-j}(\aaa_k;n)\big]_{j,k=1}^{\m} =
  \left( \prod^{\m}_{k=1} e^{n\Gfn(x_1; \aaa_k)} \right) \frac1{\m!} \bigg[ \binom{\m}{m-1} \Ii_{m-1} + \binom{\m}{m} \Ii_m \bigg] (1+o(1)).
\end{equation}

Therefore, using the asymptotics~\eqref{eq:super228} of $\Ii_\ell$, the value of $\bfZ_j$ in \eqref{eq:superj006}, and the definition of $K_{m}$ in~\eqref{eq:superj005}, we obtain
\begin{equation}\label{eq:super228.7}
  \det [ \bfGamma_{n-j}(\aaa_k;n) ]_{j,k=1}^{\m} = 
  \left( \prod_{k=1}^{\m} e^{n\Gfn(x_1; \aaa_k)-n\ell/2} \right) \frac{\binom{\m}{m-1} }{\m!} \Ii_{m-1} \bigg[ 1+  \frac{\det[\bfP_{m}] \det[\bfQ_{m}]}{\det[\bfP_{m-1}] \det[\bfQ_{m-1}]} \bigg] (1+o(1)).
\end{equation}

\subsection{Evaluation of $\det\big[ \bfGamma_{n-j}(\aaa_k;n) \bar{\cE}_{n-j+1,n} (\aaa_k;\Intx(x_i);s)\big]_{j,k=1}^{\m}$, $i=1,2$}\label{sec:supjump0}

The analysis in this subsection is similar to Section~\ref{sec:sup2e} 
but as in Section~\ref{sec:supjump1} the main contribution to involved integrals comes from neighborhoods of two points $x_1$ and $x_2$. 

We consider the interval $\Intx(x_2)$ first.
From~\eqref{eq:supersim} in Section~\ref{sec:sup2e}, 
\begin{multline}
	\det\big[ \bfGamma_{n-j}(\aaa_k;n) \bar{\cE}_{n-j+1,n} (\aaa_k; \Intx(x_2);s )\big]^{\m}_{j,k = 1} \\
	=\det\bigg[ \bfGamma_{n-j}(\aaa_k;n) \frac{\cE_{n-j+1,n} (\aaa_k; \Intx(x_2);s)}{\cE_{n-j,n}(\Intx(x_2); s)} \bigg]^{\m}_{j,k = 1}  (1+o(1)). 
\end{multline}
Now~\eqref{eq:Gammanjak11} is changed to, as it happened to~\eqref{eq:super201}, 
\begin{multline}\label{eq:super601}
  e^{n\ell/2} \bfGamma_{n-j}(\aaa_k) \frac{\mathcal{E}_{n-j+1}(\aaa_k; \Intx(x_2); s)}{\mathcal{E}_{n-j}(\Intx(x_2); s)} = \bigg[ \int_{E_1} M_{j,n}(y) e^{n\Gfn(y;\aaa_k)}dy \\
  + \int_{E_2}  M_{j,n}(y) e^{n\Gfn(y;\aaa_k)} (1-s\chi_{E_{T,\epsilon}(x_2)}(y))dy \bigg] (1+O(e^{-\delta n}))
\end{multline}
for some $\epsilon>0$ and $\delta>0$, where $E_i$ are in \eqref{eq:defn_of_E_epsilon} and 
$E_{T,\epsilon}(x_2)$ is the interval 
\begin{equation}
	E_{T,\epsilon}(x_2):= (x_2 + T/\sqrt{-\Gfn''(x_2)n},\ x_2+\epsilon).
\end{equation}
Now the analysis of Section~\ref{sec:supjump1} goes through with the change that the measure $dy$ is changed to $ (1-s\chi_{E_{T,\epsilon}(x_2)}(y))dy$ for $y \in E_2$. 
Thus we find (cf.~\eqref{eq:super210.5.5} and~\eqref{eq:super210.5})
\begin{multline} \label{eq:super210.5.5_x_2_variation}
  \det \bigg[ \int_{E_1} M_{j,n}(y) e^{n\Gfn(y;\aaa_k)}dy + \int_{E_2}  M_{j,n}(y) e^{n\Gfn(y;\aaa_k)} (1-s\chi_{E_{T,\epsilon}(x_2)}(y))dy \bigg]^{\m}_{j,k = 1} = \\
  \left( \prod^{\m}_{k=1} e^{n\Gfn(x_1; \aaa_k)} \right) \frac1{\m!}\sum_{\ell=0}^{\m}  \binom{\m}{\ell} \Ii_\ell(T;s)
\end{multline}
where 
\begin{equation}\label{eq:super210.5.0.0.1}
	\Ii_\ell(T; s) := \int_{E_1^{\m-\ell} \times E_2^{\ell}} \det[M_{j,n}(y_k)] \det[ Q_j(y_k)] \prod_{k=1}^{\m} e^{nD(y_k)} d\mu(y_k)
\end{equation}
with
\begin{equation}\label{eq:super210.5.0.0.2}
\begin{split}
	d\mu(y_k) =\begin{cases}
	dy_k, \qquad &k=1, \cdots, \m-\ell, \\
	(1- s \chi_{E_{T,\epsilon}(x_2)}(y_k))dy_k, \qquad & k=\m-\ell+1, \cdots, \m.
	\end{cases}
\end{split}
\end{equation}
The analysis that yields Lemma~\ref{lem:Ii} applies with trivial modifications and we obtain (cf.~\eqref{eq:super228})
\begin{equation}\label{eq:super228.0.0.2}
  \Ii_\ell (T; s) = \frac{\error^{\m-\ell}\errort^ \ell}{((K_{m}n)^{q_{m}(x_2-x_1)})^ \ell} \prod_{k=0}^{\m-1-\ell} \bigg(\frac{\error^k}{k!} \bigg)^2  \prod_{k=0}^{\ell-1} \bigg(\frac{\errort^k}{k!} \bigg)^2 \cdot \det[\bfP_ \ell] \det[\bfQ_ \ell] 
  \bfZ_{\m-\ell} \cdot \bfZ_{\ell}(T; s) (1+o(1))
\end{equation}
where
\begin{equation}\label{eq:super602}
\begin{split}
	\bfZ_i(T; s):= \int_{\R^i} |\Delta_i(t)|^2 \prod_{j=1}^i e^{-\frac12 t_j^2} (1-s\chi_{(T,\infty)}(t_j))dt_j  .
\end{split}
\end{equation}
Therefore, we obtain,  as in~\eqref{eq:super228.7},
\begin{multline}\label{eq:super602.01}
  \det \left[ \bfGamma_{n-j}(\aaa_k;n) \frac{\cE_{n-j+1,n}(\aaa_k; \Intx(x_2);s)}{\cE_{n-j,n}(\Intx(x_2);s)} \right]^{\m}_{j,k = 1} = \prod_{k=1}^{\m} e^{n\Gfn(x_1; \aaa_k)-n\ell/2} \\
  \times \frac{\binom{\m}{m-1} }{\m!} \Ii_{m-1}(T;s) \bigg[ 1+ \frac{\bfZ_{m-1}}{\bfZ_{m}}\frac{\bfZ_{m}(T; s)}{\bfZ_{m-1}(T; s)} \frac{\det[\bfP_{m}] \det[\bfQ_{m}]}{\det[\bfP_{m-1}] \det[\bfQ_{m-1}]} \bigg] (1+o(1)).
\end{multline}

\bigskip

For the interval $\Intx(x_1)$,~\eqref{eq:super601} is changed to 
\begin{multline}\label{eq:super604}
  e^{n\ell/2} \bfGamma_{n-j}(\aaa_k) \frac{\mathcal{E}_{n-j+1}(\aaa_k; \Intx(x_1); s)}{\mathcal{E}_{n-j}(\Intx(x_1); s)} = 
  \bigg[ \int_{E_1} M_{j,n}(y) e^{n\Gfn(y;\aaa_k)}(1-s\chi_{E_{T,\epsilon}(x_1)}(y))dy + \\
   \int_{E_2}  M_{j,n}(y) e^{n\Gfn(y;\aaa_k)} (1-s)dy \bigg]  (1+O(e^{-\delta n})).
\end{multline}
Corresponding to \eqref{eq:super210.5.5} and \eqref{eq:super210.5.5_x_2_variation}, we have
\begin{multline} \label{eq:super210.5.5_x_1_variation}
  \det \bigg[ \int_{E_1} M_{j,n}(y) e^{n\Gfn(y;\aaa_k)}(1-s\chi_{E_{T,\epsilon}(x_1)}(y))dy + 
  \int_{E_2}  M_{j,n}(y) e^{n\Gfn(y;\aaa_k)} (1-s)dy \bigg] ^{\m}_{j,k = 1} = \\
  \left( \prod^{\m}_{k=1} e^{n\Gfn(x_1; \aaa_k)} \right) \frac1{\m!}\sum_{\ell=0}^{\m}  \binom{\m}{\ell} \tilde{\Ii}_\ell(T;s)
\end{multline}
where $\tilde{\Ii}_\ell(T;s)$ are defined analogously to $\Ii_{\ell}$ in \eqref{eq:super210.5_new} and \eqref{eq:super210.5.0.0.1},
and it is straightforward to obtain
\begin{multline}\label{eq:super228.0.0.2.01}
  \tilde{\Ii}_\ell (T;s) = \frac{\error^{\m-\ell}\errort^ \ell}{((K_{m}n)^{q_{m}(x_2-x_1)})^ \ell} \prod_{k=0}^{\m-1-\ell} \bigg(\frac{\error^k}{k!} \bigg)^2  \prod_{k=0}^{\ell-1} \bigg(\frac{\errort^k}{k!} \bigg)^2 \cdot \det[\bfP_ \ell] \det[\bfQ_ \ell] \\
  \times \bfZ_{\m-\ell}(T;s) \cdot (1-s)^{\ell}\bfZ_{\ell}(1+o(1)).
\end{multline}
The difference from~\eqref{eq:super228.0.0.2} for $\Ii_{\ell}(T;s)$ is that we now have $\bfZ_{\m-\ell}(T; s)$ and$(1-s)^\ell\bfZ_\ell$ in place of $\bfZ_{\m-\ell}$ and $\bfZ_\ell(T;s)$, respectively.
Therefore, we obtain (cf. \eqref{eq:super602.01})
\begin{multline}\label{eq:super602.01.01}
  \det \left[ \bfGamma_{n-j}(\aaa_k;n) \frac{\cE_{n-j+1,n}(\aaa_k; \Intx(x_1);s)}{\cE_{n-j,n}(\Intx(x_1);s)} \right]^{\m}_{j,k = 1} = \prod_{k=1}^{\m} e^{n\Gfn(x_1; \aaa_k)-n\ell/2} \\
  \times \frac{\binom{\m}{m-1} }{\m!} \tilde{\Ii}_{m-1}(T; s) \bigg[ 1+ (1-s)\frac{\bfZ_{\m-m+1}}{\bfZ_{\m-m}} \frac{\bfZ_{\m-m}(T; s)}{\bfZ_{\m-m+1}(T; s)} \frac{\det[\bfP_{m}] \det[\bfQ_{m}]}{\det[\bfP_{m-1}] \det[\bfQ_{m-1}]} \bigg] \\
  \times (1+o(1)).
\end{multline}

\bigskip

Combining~\eqref{eq:super602.01} and~\eqref{eq:super228.7}, we obtain 
\begin{multline}\label{eq:603}
  \mathcal{E}_n(\aaa_1, \cdots \aaa_\m; \Intx(x_2); s) = 
  \frac{\det[\bfP_{m-1}]\det[\bfQ_{m-1}] \frac{\bfZ_{m-1}(T; s)}{\bfZ_{m-1}} + \det[\bfP_{m}]\det[\bfQ_{ m}] \frac{\bfZ_{m}(T; s)}{\bfZ_{m}} } {\det[\bfP_{m-1}]\det[\bfQ_{m-1}] + \det[\bfP_m]\det[\bfQ_{m}] } + o(1)
\end{multline}
and from~\eqref{eq:super602.01.01} and~\eqref{eq:super228.7} we obtain
\begin{multline}\label{eq:605}
  \mathcal{E}_n(\aaa_1, \cdots \aaa_\m; \Intx(x_1); s) = \\
  \frac{\det[\bfP_{m-1}]\det[\bfQ_{m-1}] \frac{\bfZ_{\m-m+1}(T; s)}{\bfZ_{\m-m+1}} (1-s)^{m-1} + \det[\bfP_m]\det[\bfQ_m] \frac{\bfZ_{\m-m}(T; s)}{\bfZ_{\m-m}} (1-s)^{m}} {\det[\bfP_{m-1}]\det[\bfQ_{m-1}] + \det[\bfP_{m}]\det[\bfQ_{m}] } 
  +o(1),
\end{multline}
where we use the convention that $\frac{\bfZ_{0}(T; s)}{\bfZ_{0}}=1$. It is easy to check that the convergences in~\eqref{eq:603} and~\eqref {eq:605} are all uniform for $s$ which is close to $1$. Since $\frac{\bfZ_\ell(T;s)}{\bfZ_\ell} = \erf_\ell(T; 0, \cdots, 0;s)$ (see~\eqref{eq:defination_of_spiked_GUE_dist_rank_k}), Theorem~\ref{thm:superjump} follows from~\eqref{eq:defination_of_distr_of_j_eigen_in_k_GUE}.

\section{Proof of Theorem~\ref{thm:critical1}: critical case 1, continuous transition} \label{higher_rank_critical_case}

We assume that the critical value $\acc= \frac12 V'(\redge)$ and suppose that 
$\acc\notin \mathcal{J}_V$. Let 
\begin{equation}
	\aaa_k = \frac{1}{2} V'(\redge) + \frac{\beta \alpha_k}{n^{1/3}}, 
	\qquad k=1, \cdots, \m,
\end{equation}
for fixed, distinct real numbers $\alpha_1, \cdots, \alpha_\m$.
Here $\beta$ is a positive constant defined in~\eqref{eq:defn_of_beta}.
The proof of this critical case is more involved  than other cases. 
We first need to perform some algebraic manipulations of the determinant in Theorem~\ref{thm:alg} 
to make it asymptotically easy to evaluate.

We start with a formula that is equivalent to but slightly different from Theorem~\ref{thm:alg}. 
From Lemma~\ref{lem:Dmiddle}, which is an intermediate step toward the proof of Theorem~\ref{thm:alg},  
\begin{multline}\label{eq:afom2}
  	\frac{\cE_n(\aaa_1, \cdots, \aaa_{\m}; E;s)}{\cE_n(E;s)} 
  	= \frac{1}{\det[\ga_{n-j}(\aaa_k)]^{\m}_{j,k = 1}} \det \bigg[  \langle \psi_{n-j},  \hv^t- s\chi_E(1-s\chi_E K_{n,n} \chi_E)^{-1}\hw^t \rangle \bigg]_{j=1}^{\m}
\end{multline}
where $\langle, \rangle$ is the real inner product on $\realR$ and  $K_{\ell,n}(x,y)=(p_0(x)p_0(y)+\cdots p_{\ell-1}(x)p_{\ell-1}(y))e^{-\frac{n}{2}(V(x)+V(y))}$ is the usual Crhistroffel--Darboux kernel.
Here $p_\ell(x)=p_\ell(x;n)$ is the orthonormal polynomial with respect to the (varying) measure $e^{-nV(x)}dx$ on $\R$, 
and 
$\psi_\ell(x):= p_\ell(x)e^{-\frac{n}2V(x)}$.
The column vector $\hv(x):= (v_1(x), \cdots, v_\m(x))^t$ is defined by 
\begin{equation} 
	v_k(x) :=e^{n(\aaa_k x-V(x)/2)},
\end{equation}
and the column vector $\hw(x):= (w_1(x), \cdots, w_\m(x))^t$ is given by 
\begin{equation}\label{eq:-3}
   	w_k(x):= ((1-K_{n,n}) v_k)(x).
\end{equation}

Note that the kernel $K_{n,n}$ in~\eqref{eq:afom2} is independent of $j$. 
This is the difference from the formula~\eqref{eq:alg}: in terms of $\psi_{\ell}$ and $K_{\ell, n}$, the formula~\eqref{eq:alg} becomes~\eqref{eq:last_eq_in_proof_8.1} in which the Christoffel--Darboux kernel appears as $K_{n-j+1,n}$, depending on the row index $j$. 
This change makes the following computation easier. 

\medskip

We use  the three-term recurrence relation of orthonormal polynomials (see \eg, \cite{Szego75}) repeatedly below. In terms of $\psi_\ell(x)=p_\ell(x) e^{-\frac{n}{2}V(x)}$, 
\begin{equation}\label{eq:threeterm}
	x \psi_{\ell}(x)= b_{\ell} \psi_{\ell+1}(x) + a_{\ell} \psi_{\ell}(x) + b_{\ell-1} \psi_{\ell-1}(x), \qquad \ell \ge 1
\end{equation}
for some constants $a_\ell$ and for positive constants
\begin{equation}
	b_{\ell} = \frac{\gamma_{\ell}}{\gamma_{\ell+1}},
\end{equation}
where $\gamma_{\ell}$ is the leading coefficient of $p_{\ell}(x)$.

\subsection{Evaluation of $\det\big[ \bfGamma_{n-j}(\aaa_k) \big]_{j,k=1}^{\m}$}\label{sec:numcri}

Using the notations above, we have (see~\eqref{eq:gadef})
\begin{equation}
  \det\big[ \bfGamma_{n-j}(\aaa_k) \big]_{j,k=1}^{\m} = \det\big[ \langle \psi_{n-j}, \hv^t \rangle \big]_{j=1}^{\m}.
\end{equation}
By taking a linear combination of the last three rows and using the three-term recurrence relation~\eqref{eq:threeterm}, we can replace the last row in the above matrix by the vector 
\begin{equation}
	\frac1{b_{n-\m}} \langle (x-\redge) \psi_{n-\m+1}(x), \hv^t(x) \rangle.
\end{equation}
We then can replace the $(\m-1)$-th row  similarly by using the two rows above. By repeating this process up to the third row, we obtain
\begin{equation}
  \begin{split}
    \bigg( \prod_{\ell=n-\m}^{n-3} b_\ell \bigg) \det\big[ \bfGamma_{n-j}(\aaa_k) \big]_{j,k=1}^{\m} 
    = \det
    \begin{bmatrix}
      \langle \psi_{n-1}, \hv^t \rangle \\
      \langle \psi_{n-2}, \hv^t \rangle  \\
      \langle (x-\redge) \psi_{n-2}(x), \hv^t \rangle \\
      \vdots \\
      \langle (x-\redge) \psi_{n-\m+1}(x), \hv^t \rangle
    \end{bmatrix}.
  \end{split}
\end{equation}
Now we can change the last row of this new matrix to 
\begin{equation}
	\frac1{b_{n-\m+1}} \langle (x-\redge)^2 \psi_{n-\m+2}(x), \hv^t \rangle
\end{equation}
without changing the determinant, 
by using a linear combination of the last three rows and the three-term recurrence relation again. We repeat this process up to the fifth row and obtain 
\begin{equation}
    \bigg( \prod_{\ell=n-\m+1}^{n-4} b_\ell \bigg)  \bigg( \prod_{\ell=n-\m}^{n-3} b_\ell \bigg) \det\big[ \bfGamma_{n-j}(\aaa_k) \big]_{j,k=1}^{\m} 
    = \det
        \begin{bmatrix}
          \langle \psi_{n-1}, \hv^t \rangle \\
          \langle \psi_{n-2}, \hv^t \rangle  \\
	\langle (x-\redge) \psi_{n-2}(x), \hv^t \rangle \\
        \langle (x-\redge) \psi_{n-3}(x), \hv^t \rangle \\ 
	\langle (x-\redge)^2 \psi_{n-3}(x), \hv^t \rangle \\
        \vdots \\
        \langle (x-\redge)^2 \psi_{n-\m+1}(x), \hv^t \rangle
	\end{bmatrix}.
\end{equation}
We repeat the process and obtain, for even $\m$, 
\begin{equation}\label{eq:GammatoR}
  \bigg( \prod_{\ell=1}^{[\m/2]-1} (b_{n-2-\ell}b_{n-\m-1+\ell})^\ell \bigg) \det\big[ \bfGamma_{n-j}(\aaa_k) \big]_{j,k=1}^{\m} = \det
  \begin{bmatrix}
    R_1  \\
    \vdots \\
    R_{[\m/2]} 
  \end{bmatrix}.
\end{equation}
where each $R_\ell$ is a $2\times \m$ matrix defined by  
\begin{equation}
  R_\ell =
  \begin{bmatrix}
   \langle (x-\redge)^{\ell-1} \psi_{n-\ell}(x), \hv^t \rangle \\
   \langle (x-\redge)^{\ell-1} \psi_{n-\ell-1}(x), \hv^t \rangle
  \end{bmatrix}.
\end{equation}
The determinant in~\eqref{eq:GammatoR} is unchanged if 
we add the second row of $R_\ell$ by a constant multiple of the first row. 
Hence we can change the matrices $R_\ell$ in~\eqref{eq:GammatoR} to  
\begin{equation}\label{eq:-2}
  R_\ell =
  \begin{bmatrix}
    \langle (x-\redge)^{\ell-1} \psi_{n-\ell}(x), \hv^t \rangle \\
    \langle  (x-\redge)^{\ell-1} (\psi_{n-\ell-1}(x) - \frac{B_{\ell+1, n}(\redge)}{B_{\ell, n}(\redge)} \psi_{n-\ell}(x)), \hv^t \rangle
  \end{bmatrix}
\end{equation}
where $B_{j,n}(\redge)$ is the value of $B_{j,n}(z)$ at $z = \redge$, and 
the function $B_{j,n}(z)$ is a function defined in \cite[Proposition 6.1(b)]{Baik-Wang10a} which appears in the asymptotic of orthonormal polynomial near the edge $\redge$ of the support of the equilibrium measure. 
From the asymptotics of $B_{j,n}(z)$ \cite[Formulas (323) and (313)]{Baik-Wang10a}, it was shown that $B_{j,n}(z)$ and its reciprocal are uniformly bounded in a neighborhood of $z=\redge$.

When $\m$ is odd, we need to add an extra row $\langle (x-\redge)^{[\m/2]} \psi_{n-[\m/2]-1}(x), \hv^t \rangle$ to the matrix to the right-hand side of~\eqref{eq:GammatoR} and the extra term $b_{n-[\m/2]-1}^{[\m/2]}$ needs to be multiplied on the left-hand side. In the remaining part of this section, we consider only even $\m$ since the odd $\m$ case can be solved by the same method.

We now evaluate the asymptotics of $R_{\ell+1}$ for each $\ell=0,1,\cdots$. 
First consider  the top row of $R_{\ell+1}$.
We consider a slightly more general quantity 
\begin{equation}\label{eq:QPO1}
	\langle (x-\redge)^{\ell} \psi_{n-j}(x), e^{n(\aaa_k x-V(x)/2)} \rangle
\end{equation}
for a later use. 
Observe that $\psi_{n-\ell}$ is changed to $\psi_{n-j}$. 
Note from~\eqref{eq:gadef} with $a=\aaa_k$, 
$\bfGamma_{n-j}(\aaa_k)= \langle \psi_{n-j}(x), e^{n(\aaa_k x-V(x)/2)} \rangle$. 
The asymptotics of this inner product at the critical case was obtained in \cite[Section 5.1]{Baik-Wang10a}
and the asymptotics of~\eqref{eq:QPO1} is very similar. 
Namely setting 
\begin{equation}
	\varphi_{n-j}(x) := \psi_{n-j}(x) e^{-\frac{n}{2}V(x)},
\end{equation}
we see that~\eqref{eq:QPO1} equals 
\begin{equation}\label{eq:QPO2}
	\int_{\Sigma_+\cup \Sigma_-} (C\varphi_{n-j})(z) (z-\redge)^\ell e^{n\aaa_k z} dz
	+ \int^{\infty}_{\redge} \varphi_{n-j}(z)(z-\redge)^{\ell} e^{n\aaa_k y}dz 
\end{equation}
where $(C\varphi_{n-j})(z)$ is the Cauchy transform of $\varphi_{n-j}(x)$
and $\Sigma_+$ and $\Sigma_-$ are certain contours from $\redge$ to $\infty$ lying in 
$\compC_+$ and $\compC_-$, respectively (see \cite[Figure 9]{Baik-Wang10a}).
When $\ell=0$, it was shown in \cite[Section 5.1]{Baik-Wang10a} that, under the criticality assumption the main contribution to the above integrals comes from a neighborhood of the point $z=\redge$
and an appropriate change of variable is 
$\xi= (z-\redge)\beta n^{2/3}$. 
The presence of the term $(z-\redge)^\ell$ above does not change the analysis except for an inclusion of an extra term $(\frac{\xi}{\beta n^{2/3}})^\ell$ in \cite[Formula (222)]{Baik-Wang10a}
and we obtain
\begin{multline} \label{eq:evaluation_of_likeGamma_l_j}
  	\langle (x-\redge)^{\ell} \psi_{n-j}(x), e^{n(\aaa_k x-V(x)/2)} \rangle \\
  = \frac{Q_n(\aaa_k)}{\beta\sqrt{n} (\beta n^{2/3})^{\ell} }  \left(B_{j,n}(\redge) \int^{\infty}_0 \xi^\ell \Ai(\xi)(e^{\alpha_k\xi}+ e^{\omega\alpha_k\xi}+e^{\omega^2\alpha_k\xi}) d\xi + o(1) \right).
\end{multline}
where 
(see \cite[Formula (206)]{Baik-Wang10a})
\begin{equation} \label{eq:190}
	Q_n(\aaa_k) = e^{n(-\frac{1}2V(\redge) + \aaa_k \redge)}=e^{n(\Hfn(\redge; \aaa_k)-\ell/2)} = e^{n(\Gfn(\redge; \aaa_k)-\ell/2)}.
\end{equation}
(The result of \cite{Baik-Wang10a} involves $\B_{j,n}(\redge)$ in place of $B_{j,n}(\redge)$. But as $B_{j,n}(z)= \B_{j,n}(z)(1+O(n^{-1}))$ from \cite[Formula (323)]{Baik-Wang10a}
and $1/\B_{j,n}(\redge)$ is uniformly bounded, the above statement follows.) 
Now observe that (see \cite[Formula (223)]{Baik-Wang10a})
\begin{equation}\label{eq:bfGak6}
\begin{split}
	 \int_0^\infty  \Ai(\xi) (e^{\alpha \omega \xi}
	+  e^{\alpha \omega^2 \xi} +e^{\alpha\xi}) d\xi = e^{\alpha^3/3}.
\end{split}
\end{equation}
By taking the derivatives of this identity with respect to $\alpha$, we find that
\begin{equation}\label{eq:bfGak7}
\begin{split}
	\int_0^\infty \xi^\ell \Ai(\xi) (\omega^\ell e^{\alpha \omega \xi}
	+ \omega^{2\ell} e^{\alpha \omega^2 \xi} +e^{\alpha\xi}) d\xi 
	= \bigg( \frac{d}{d\alpha}\bigg)^\ell e^{\alpha^3/3}. 
\end{split}
\end{equation}
Hence we obtain
\begin{equation}\label{eq:cr101}
  \langle (x-\redge)^{\ell} \psi_{n-j}(x), e^{n(\aaa_k x-V(x)/2)} \rangle 
	= \frac{Q_n(\aaa_k)}{\beta\sqrt{n} (\beta n^{2/3})^{\ell} } B_{j,n}(\redge) \bigg( \frac{d}{d\alpha_k}\bigg)^{\ell} e^{\alpha_k^3/3}  (1+ o(1) ).
\end{equation}

\bigskip

Now we consider the second row of $R_{\ell+1}$. Again we consider a slightly more general quantity 
\begin{equation}\label{eq:QPR3}
	\langle (x-\redge)^{\ell} (\psi_{n-j-1}(x) - \frac{B_{j+1,n}(\redge)}{B_{j,n}(\redge)} \psi_{n-j}(x)), e^{n(\aaa_k x-V(x)/2)} \rangle.
\end{equation}
This can be written as the sum of the integrals~\eqref{eq:QPO2} with the terms $(C\varphi_{n-j})(z)$ and $\varphi_{n-j}(z)$ replaced by 
$\varphi_{n-j-1}(z) - \frac{B_{j+1,n}(\redge)}{B_{j,n}(\redge)} \varphi_{n-j}(z)$
and $\varphi_{n-j-1}(z) - \frac{B_{j+1,n}(\redge)}{B_{j,n}(\redge)} \varphi_{n-j}(z)$, respectively.
Then again the main contribution to the integrals come near $z=\redge$. 
The precise behaviors of the integrands near $z$ are well known (see  \cite[Formula (322)]{Baik-Wang10a}).
First, 
\begin{equation} \label{eq:separation_of_B_part_and_D_part}
  	\varphi_{n-j-1}(z) - \frac{B_{j+1,n}(\redge)}{B_{j,n}(\redge)} \varphi_{n-j}(z) = 
  \left( n^{1/6} \Ai(\Phi(z)) c_1(z) + n^{-1/6} \Ai'(\Phi(z)) c_2(z) \right) e^{-\frac{n}{2}V(z)}
\end{equation}
for $z$ near $\redge$. Secondly, 
\begin{multline} \label{eq:separation_of_B_part_and_D_part_Cauchy_trans_upper}
  (C\varphi_{n-j-1})(z) - \frac{B_{j+1,n}(\redge)}{B_{j,n}(\redge)} (C\varphi_{n-j})(z) = \\
  e^{\frac{\pi i}{3}} \left( n^{1/6} \Ai(\omega^2\Phi(z)) c_1(z) + n^{-1/6} \omega^2\Ai'(\omega^2\Phi(z)) c_2(z) \right) e^{-\frac{n}{2}V(z)}
\end{multline}
for $z$ near $\redge$ with $z\in  \compC_+$. Finally, 
\begin{multline} \label{eq:separation_of_B_part_and_D_part_Cauchy_trans_lower}
  (C\varphi_{n-j-1})(z) - \frac{B_{j+1,n}(\redge)}{B_{j,n}(\redge)} (C\varphi_{n-j})(z) = \\
  -e^{\frac{\pi i}{3}}  \left( n^{1/6} \omega^2 \Ai(\omega\Phi(z)) c_1(z) + n^{-1/6} \Ai'(\omega\Phi(z)) c_2(z) \right) e^{-\frac{n}{2}V(z)}
\end{multline}
for $z$ near $\redge$ with $z\in \compC_-$.
Here $\Phi(z)$ is a function that satisfies $\Phi(z)= \beta n^{2/3} (z-\redge) (1+O(|z-\redge|))$ as $z\to \redge$ and is defined by Baik and Wang \cite[Formula (309)]{Baik-Wang10a}).
The functions $c_1(z)$ and $c_2(z)$ are given by 
\begin{equation} \label{eq:Bc1c2}
  	c_1(z) := B_{j+1,n}(z) - \frac{B_{j+1,n}(\redge)}{B_{j,n}(\redge)} B_{j,n}(z),  \quad
	c_2(z) := D_{j+1,n}(z) - \frac{B_{j+1,n}(\redge)}{B_{j,n}(\redge)} D_{j,n}(z).
\end{equation}
As $B_{\ell, n}(z)$ is analytic at $z=\redge$ and its reciprocal is uniformly away from zero in a neighborhood of $\redge$, we have 
\begin{equation}\label{eq:q010}
	c_1(z)=O(z-\redge) \qquad \text{for $z$ near $\redge$.}
\end{equation} 
Also $B_{j,n}(z)D_{j+1,n}(z) - B_{j+1,n}(z)D_{j,n}(z)$ is shown to be independent of $z$ (see \cite[Formula (329)]{Baik-Wang10a}) and equal to $\frac{\gamma_{n-j}}{\gamma_{n-j-1}}$ where 
$\gamma_\ell$ is the leading coefficient of $p_{\ell}(z)$.
Hence
\begin{equation} \label{eq:D_part_at_e}
  	c_2(\redge)  
	 = - \frac{\kappa_{j,n}}{B_{j,n}(\redge)}, \qquad
	 \kappa_{j,n}:= - \frac{\gamma_{n-j}}{\gamma_{n-j-1}} .
\end{equation}
We note that from the explicit asymptotics \cite[Formula (303), (304)]{Baik-Wang10a}
of $\gamma_{n-j}$, $ \kappa_{j,n}$ and its reciprocal are bounded uniformly in $n$.

From this we can find the asymptotics of~\eqref{eq:QPR3} 
in a  similar form as~\eqref{eq:evaluation_of_likeGamma_l_j}.
The resulting formula contains two integrals, one involving $\Ai$ and the other $\Ai'$, 
since each of~\eqref{eq:separation_of_B_part_and_D_part},~\eqref{eq:separation_of_B_part_and_D_part_Cauchy_trans_upper}, and~\eqref{eq:separation_of_B_part_and_D_part_Cauchy_trans_lower} contains such two terms. 
Now notice that due to~\eqref{eq:q010} and the change of variables 
$\xi= (z-\redge)\beta n^{2/3}$, 
the integral involving $\Ai$ is smaller than the integral involving $\Ai'$ by the factor  $O(n^{-1/3})$. 
Thus we find that 
\begin{multline} \label{eq:temtem}
  	\langle (x-\redge)^{\ell}  (\psi_{n-j-1}(x)-\frac{B_{j+1,n}(\redge)}{B_{j,n}(\redge)}\psi_{n-j}(x)), e^{n(\aaa_k x-V(x)/2)} \rangle \\
  	= - \frac{Q_n(\aaa_k)\kappa_{j,n} }{\beta n^{5/6} (\beta n^{2/3})^{\ell} B_{j,n}(\redge)}  
	\left( \int^{\infty}_0 \xi^\ell \Ai'(\xi)(e^{\alpha_k\xi}+ \omega^2\omega^{\ell} e^{\omega\alpha_k\xi} + \omega\omega^{2\ell}e^{\omega^2\alpha_k\xi}) d\xi + o(1) \right).
\end{multline}

The integral above can be simplified by the identity
\begin{equation}\label{eq:bfGak9}
	\int_0^\infty \xi^\ell \Ai'(\xi) (\omega^2 \omega^\ell e^{\alpha \omega \xi}
	+ \omega \omega^{2\ell} e^{\alpha \omega^2 \xi} +e^{\alpha\xi}) d\xi 
	= \bigg( \frac{d}{d\alpha}\bigg)^\ell (-\alpha e^{\alpha^3/3}).
\end{equation}
This identity is obtained by taking derivatives with respect to $\alpha$ of the identity 
\begin{equation}\label{eq:bfGak8}
\begin{split}
	 \int_0^\infty  \Ai'(\xi) (\omega^2 e^{\alpha \omega \xi}
	+  \omega e^{\alpha \omega^2 \xi} +e^{\alpha\xi}) d\xi =-\alpha  e^{\alpha^3/3},
\end{split}
\end{equation}
which follows from~\eqref{eq:bfGak6} after integrating by parts. 
Hence we obtain 
\begin{multline}\label{eq:-1}
  	\langle (x-\redge)^{\ell} (\psi_{n-j-1}(x)-\frac{B_{j+1,n}(\redge)}{B_{j,n}(\redge)}\psi_{n-j}(x)), e^{n(\aaa_k x-V(x)/2)} \rangle \\
  =  \frac{Q_n(\aaa_k)\kappa_{j,n} }{\beta n^{5/6} (\beta n^{2/3})^{\ell} B_{j,n}(\redge)}   
  \bigg( \frac{d}{d\alpha_k}\bigg)^{\ell} (\alpha_k e^{\alpha_k^3/3}) + o(1). 
\end{multline}

\bigskip

Inserting~\eqref{eq:cr101} 
and~\eqref{eq:-1} 
(with $\ell\mapsto \ell-1$ and $j=\ell$) into~\eqref{eq:-2}, we obtain that~\eqref{eq:GammatoR} 
equals, when $\m$ is even, 
\begin{equation} \label{eq:big_formula_of_det_Gamma_critical_1}
  	\prod_{\ell=1}^{[\m/2]}
	\frac{\kappa_{\ell, n}}{(\beta n^{2/3})^{2\ell}} 
	\prod_{k=1}^{\m} Q_n(\aaa_k) 
  	\left( \det
  \begin{bmatrix} 
    e^{\alpha_1^3/3} &  \cdots & e^{\alpha_\m^3/3}  \\
    \alpha_1 e^{\alpha_1^3/3} & \cdots &\alpha_\m e^{\alpha_\m^3/3} \\
    \frac{d}{d\alpha_1} e^{\alpha_1^3/3} & \cdots & \frac{d}{d\alpha_\m} e^{\alpha_\m^3/3}  \\
    \frac{d}{d\alpha_1} (\alpha_1 e^{\alpha_1^3/3})  & \cdots & \frac{d}{d\alpha_\m} (\alpha_\m e^{\alpha_\m^3/3}) \\
    \big( \frac{d}{d\alpha_1}\big)^2 e^{\alpha_1^3/3}  & \cdots &\big( \frac{d}{d\alpha_\m}\big)^2 e^{\alpha_\m^3/3}  \\
    \big( \frac{d}{d\alpha_1}\big)^2 (\alpha_1 e^{\alpha_1^3/3}) & \cdots &\big( \frac{d}{d\alpha_\m}\big)^2 (\alpha_\m e^{\alpha_\m^3/3}) \\
    \vdots &\vdots & \vdots  
  \end{bmatrix}
  +o(1) \right).
\end{equation}

Note that the $(j,k)$ entry of the determinant on the right-hand side of \eqref{eq:big_formula_of_det_Gamma_critical_1} is of the form $P_j(\alpha_k)e^{\alpha_k^3/3}$ for some polynomial $P_j(x)$ of degree $j-1$ with leading coefficient $1$, (\ie, $P_j(x)=x^{j-1} + \cdots$),
which are defined by the conditions 
\begin{equation}\label{eq:defpolyPj}
  P_j(\alpha) =
  \begin{cases}
    e^{-\alpha^3/3} \big( \frac{d}{d\alpha} \big)^{i} e^{\alpha^3/3} & \textnormal{if $j = 2i$,} \\
    e^{-\alpha^3/3} \big( \frac{d}{d\alpha} \big)^{i} (\alpha e^{\alpha^3/3}) & \textnormal{if $j=2i+1$.}
  \end{cases}
\end{equation}
Therefore, by elementary row operations we find that the determinant is same as the determinant 
of the matrix $(\alpha_k^{j-1} e^{\alpha_k^3/3})_{j,k=1}^{\m}$. 
The determinant of this matrix is $\prod_{1\le j<k\le \m} (\alpha_k-\alpha_j) \prod_{k=1}^{\m} e^{\alpha_k^3/3}$ and this is nonzero. 
Therefore, when $\m$ is even, 
\begin{multline}\label{eq:-10}
  	\bigg[ \prod_{\ell=1}^{[\m/2]-1}  (b_{n-2-\ell}b_{n-\m+\ell})^\ell \bigg] \det\big( \bfGamma_{n-j}(\aaa_k) \big)_{j,k=1}^{\m}  \\
  	= \prod_{\ell=1}^{[\m/2]}
	\frac{\kappa_{\ell, n}}{(\beta n^{2/3})^{2\ell} } 
	\prod_{k=1}^{\m} Q_n(\aaa_k) e^{\alpha_k^3/3} 
	 \prod_{1\le j<k\le \m} (\alpha_k-\alpha_j) (1+o(1)).
\end{multline}
We have a similar result when $\m$ is odd.

\subsection{Evaluation of $\det \big[  \langle \psi_{n-j},  \hv^t- s\chi_{\Int}(1-s\chi_{\Int} K_{n,n} \chi_{\Int})^{-1}\hw^t \rangle \big]_{j=1}^{\m}$}

We now evaluate the numerator of~\eqref{eq:afom2} when $E=\Int$. Note that $\psi_{n-j}$ is the only term that depends on $j$.  Hence by using the same row operations as in Section~\ref{sec:numcri} that lead to~\eqref{eq:GammatoR} and \eqref{eq:-2}, we find that, when $\m$ is even, 
\begin{equation}\label{eq:GammatoS}
  \bigg[ \prod_{\ell=1}^{[\m/2]-1} (b_{n-2-\ell}b_{n-\m+\ell})^\ell \bigg] \det \bigg[  \langle \psi_{n-j},  \hv^t - s\chi_{\Int}(1-s\chi_{\Int} K_{n,n} \chi_{\Int})^{-1}\hw^t \rangle \bigg]_{j=1}^{\m} = 
  \det
  \begin{bmatrix}
    S_1  \\
    \vdots \\
    S_{[\m/2]} 
  \end{bmatrix},
\end{equation}
where $S_\ell$ is a $2\times \m$ matrix given by  
\begin{equation}\label{eq:Slle-0101}
  \begin{split}
    S_\ell = &
    \begin{bmatrix}
      \langle (x-\redge)^{\ell-1} \psi_{n-\ell}(x), \hv^t- s\chi_{\Int}(1-s\chi_{\Int} K_{n,n} \chi_{\Int})^{-1}\hw^t \rangle \\
      \langle (x-\redge)^{\ell-1} (\psi_{n-\ell-1}(x) - \frac{B_{\ell+1, n}(\redge)}{B_{\ell, n}(\redge)} \psi_{n-\ell}(x)), \hv^t- s\chi_{\Int}(1-s\chi_{\Int} K_{n,n} \chi_{\Int})^{-1}\hw^t \rangle
    \end{bmatrix}.
  \end{split}
\end{equation}
We can write this as 
\begin{equation}\label{eq:Slle}
  \begin{split}
    S_\ell = R_{\ell} - \begin{bmatrix} \bfU_{\ell-1, \ell} \\  \bfV_{\ell-1, \ell} \end{bmatrix}
  \end{split}
\end{equation}
where $\bfU_{\ell, j} = (\bfU_{\ell, j}(\aaa_1), \cdots, \bfU_{\ell, j}(\aaa_\m))$ 
and $\bfV_{\ell, j} = (\bfV_{\ell, j}(\aaa_1), \cdots, \bfV_{\ell, j}(\aaa_\m))$ with
\begin{equation}
\begin{split}
	\bfU_{\ell, j}(\aaa_k) &:= \langle (x-\redge)^{\ell} \psi_{n-j}(x), s\chi_{\Int}(1- s\chi_{\Int} K_{n,n} \chi_{\Int})^{-1} w_k \rangle \\
	\bfV_{\ell, j}(\aaa_k) &:=  \langle (x-\redge)^{\ell}  (\psi_{n-j-1}(x) - \frac{B_{\ell+1, n}(\redge)}{B_{\ell, n}(\redge)} \psi_{n-j}(x)), s\chi_{\Int}(1- s\chi_{\Int} K_{n,n} \chi_{\Int})^{-1} w_k \rangle .
\end{split}	
\end{equation}

The asymptotics of 
$\tilde{\psi}_n(x;\aaa_k) = \frac{w_k(x)}{\bfGamma_n(\aaa_k)}$ were obtained \cite[Lemma 5.2]{Baik-Wang10a}. 
From this the asymptotics of $\bfU_{\ell, j}/\bfGamma_n(\aaa_k)$ when $\ell=0$ and $s=1$
were obtained  in \cite[Section 5.1.3]{Baik-Wang10a}. 
It is straightforward to extend this to other $\ell$ and $s$ as  
in the previous subsection.
We can follow the arguments in \cite[Sections 5.1.2 and 5.1.3]{Baik-Wang10a} almost verbatim and find 
\begin{equation}
\begin{split}
  	&\frac{\bfU_{\ell, j}(\aaa_k)}{\bfGamma_n(\aaa_k)}=  \langle (x - \redge)^{\ell} \psi_{n-j}(x), s\chi_{E_{T, \epsilon}} (1 - \chi_{\Int}K_{n,n}\chi_{\Int})^{-1} \tilde{\psi}_{n-j}(x; \aaa_k) \rangle (1 + o(1)) \\
  = & \frac1{(\beta n^{2/3})^\ell}  \langle \xi^\ell \Ai(\xi),  s\chi_{[T, \infty)}(1- s \chi_{[T, \infty)}K_{\Airy} \chi_{[T, \infty)})^{-1} C_{-\alpha_k}(\xi) \rangle (1+o(1)),
\end{split}
\end{equation}
where $E_{T,\epsilon} = \Int \setminus (\redge+\epsilon, \infty)$ with a small enough constant $\epsilon$ 
and
\begin{equation}
	C_{-\alpha}(\xi)= \frac1{2\pi} \int e^{i\frac13 z^3+i\xi z} \frac{dz}{-\alpha+iz}
\end{equation}
is defined in~\eqref{eq:Calphadef}. 
Using the asymptotics \cite[Formula (197)]{Baik-Wang10a} of $\Gamma_n(\aaa_k)$ (It was given in terms of $\B_{j,n}(\redge)$ but we can change it to $B_{j,n}(\redge)$. See text in parenthesis below equation \eqref{eq:190}.), 
this implies that  
\begin{equation}\label{eq:-5}
 	\bfU_{\ell, j}(\aaa_k)=   \frac{Q_n(\aaa_k)B_{j,n}(\redge)}{\beta\sqrt{n} (\beta n^{2/3})^{\ell} }   e^{\alpha_k^3/3} 
        \langle \xi^\ell \Ai(\xi), s\chi_{[T, \infty)} (1- s\chi_{[T, \infty)} K_{\Airy} \chi_{[T, \infty)})^{-1} C_{-\alpha_k} \rangle  (1+ o(1) ).
\end{equation}
Similarly, as in the argument for the asymptotics~\eqref{eq:temtem}, we find that 
\begin{equation}\label{eq:-6}
	\bfV_{\ell,j}(\aaa_k) = - \frac{Q_n(\aaa_k)\kappa_{j,n}}{\beta n^{5/6}(\beta n^{2/3})^{\ell} 	B_{j,n}(\redge)  }  e^{\alpha_k^3/3}  \langle \xi^\ell \Ai'(\xi), s (1- s\chi_{[T, \infty)} K_{\Airy} \chi_{[T, \infty)})^{-1} C_{-\alpha_k} \rangle  (1+ o(1) ).
\end{equation}

Recall the polynomials $P_j(\alpha)$ defined in~\eqref{eq:defpolyPj}. We claim that  
\begin{equation}
  \xi^{i} \Ai(\xi) =  P_{2i}(- \frac{d}{d\xi}  ) \Ai (\xi),  \qquad 
  -\xi^{i} \Ai'(\xi) =  P_{2i+1}(- \frac{d}{d\xi}  ) \Ai (\xi). \label{eq:Aipwer2}
\end{equation}
To see this, note that successive integrations by parts 
of the integral representation of the Airy function $\Ai(\xi) = \frac1{2\pi \sqrt{-1}} \int_{\infty e^{-\pi i/3}}^{\infty e^{\pi i/3}} e^{- \xi s+ \frac13 s^3} ds$ imply that, for any $i$, 
\begin{equation}
\begin{split}
	\Ai(\xi) 
	&=  \frac1{2\pi \sqrt{-1}} \int_{\infty e^{-\pi i/3}}^{\infty e^{\pi i/3}}
	\frac1{\xi^i} e^{-\xi s} \bigg[ \bigg( \frac{d}{ds}\bigg)^i e^{\frac13 s^3} \bigg] ds= \frac1{2\pi  \sqrt{-1}} \int_{\infty e^{-\pi i/3}}^{\infty e^{\pi i/3}}
	\frac1{\xi^i} e^{-\xi s} P_{2i}(s) e^{\frac13 s^3}  ds.
\end{split}
\end{equation}
Hence
\begin{equation}
\begin{split}
	\xi^i \Ai(\xi) 
	&=   \frac1{2\pi  \sqrt{-1}} \int_{\infty e^{-\pi i/3}}^{\infty e^{\pi i/3}}
	 P_{2i}(s) e^{-\xi s}  e^{\frac13 s^3}  ds\\
	 &=   \frac1{2\pi  \sqrt{-1}} \int_{\infty e^{-\pi i/3}}^{\infty e^{\pi i/3}}
	 \bigg[ P_{2i}(- \frac{d}{d\xi}  ) e^{-\xi s}  \bigg] e^{\frac13 s^3}  ds 
	 =P_{2i}(- \frac{d}{d\xi}  ) \Ai (\xi),
\end{split}
\end{equation}
which proves the first identity of~\eqref{eq:Aipwer2}. Similarly, for any $i$, 
\begin{equation}
\begin{split}
	-\xi^i \Ai'(\xi) 
	&=   \xi^i \frac1{2\pi \sqrt{-1}} \int_{\infty e^{-\pi i/3}}^{\infty e^{\pi i/3}}
	s e^{-\xi s+\frac13 s^3}  ds =   \frac1{2\pi \sqrt{-1}} \int_{\infty e^{-\pi i/3}}^{\infty e^{\pi i/3}}
	e^{-\xi s} \bigg[ \bigg( \frac{d}{ds}\bigg)^i (s e^{\frac13 s^3}) \bigg] ds \\
	&= \frac1{2\pi  \sqrt{-1}} \int_{\infty e^{-\pi i/3}}^{\infty e^{\pi i/3}}
	e^{-\xi s} P_{2i+1}(s) e^{\frac13 s^3}  ds  
	= P_{2i+1}(- \frac{d}{d\xi}  ) \Ai (\xi),
\end{split}
\end{equation}
which proves the second identity. 

We insert~\eqref{eq:Aipwer2} into~\eqref{eq:-5} and~\eqref{eq:-6}. 
This gives the asymptotics of the second matrix in the definition of $S_\ell$. For $R_\ell$, we use the asymptotics~\eqref{eq:cr101} and~\eqref{eq:-1} and insert~\eqref{eq:defpolyPj}. 
Then we obtain, when $\m$ is even, 
\begin{equation}\label{eq:-7}
  \begin{split}
    & (-1)^{[\m/2]} \prod_{\ell=1}^{[\m/2]}\frac{\kappa_{\ell,n}}{(\beta n^{2/3})^{2\ell}} \prod_{k=1}^{\m} Q_n(\aaa_k)e^{\alpha_k^3/3}  \\
    &\times \det
    \left[ P_{j-1}(\alpha_k) - \langle P_{j-1}(- \frac{d}{d\xi}  )  \Ai(\xi), s\chi_{[T, \infty)} (1-s\chi_{[T, \infty)} K_{\Airy} \chi_{[T, \infty)})^{-1} C_{-\alpha_k} \rangle + o(1) \right]_{j,k=1}^{\m}.
\end{split}
\end{equation}
Simple row operations then imply that  the last determinant, without $o(1)$ term, equals 
\begin{equation}\label{eq:-8}
  \begin{split}
    &\det \left[ \alpha_k^{j-1}  - \langle  (- \frac{d}{d\xi}  )^{j-1}  \Ai(\xi), s \chi_{[T, \infty)}(1-s\chi_{[T, \infty)} K_{\Airy} \chi_{[T, \infty)})^{-1} C_{-\alpha_k} \rangle \right]_{j,k=1}^{\m} . 
\end{split}
\end{equation}

\subsection{Proof of Theorem~\ref{thm:critical1}}

From~\eqref{eq:afom2},~\eqref{eq:-10}, and~\eqref{eq:-8}, we find that 
\begin{equation} \label{eq:PT1}
\begin{split}
	&\frac{\cE_n(\aaa_1, \cdots, \aaa_\m; \Int; s)}{F_0(x)} 
	=
	 \frac{\det [ D_\m(s) ] }
	 {	 \prod_{1\le j<k\le \m} (\alpha_k-\alpha_j)}  +o(1)
\end{split}
\end{equation}
where
\begin{equation} \label{eq:PT2}
\begin{split}
	D_\m(s) :=  \left[ (-\alpha_k)^{j-1}  - \langle (\frac{d}{d\xi}  )^{j-1}  \Ai(\xi), s\chi_{[T, \infty)} (1-s\chi_{[T, \infty)} K_{\Airy} \chi_{[T, \infty)})^{-1} C_{-\alpha_k} \rangle \right]_{j,k=1}^{\m} .
\end{split}
\end{equation}

When $s=1$, $D_\m(s)$ is precisely the matrix $M$ defined in \cite[Formula (3.36)]{Baik06} with $w_k=-\alpha_k$ (see \cite[Formula (3.9)]{Baik06} for the definition of $E_w$ and \cite[Formulas (3.4) and (1.10)]{Baik06} for the definition of $T_1$).  
A different formula of $\det (M)$ was then obtained \cite[Formula (3.46)]{Baik06} in terms of function $f(x;w)$.
Comparing with the case of $k=1$ of \cite[Formula (1.16)]{Baik06}, this function 
$f(x,w)= \frac{F_1(x; w)}{F_0(x)}$ and this implies that 
\begin{equation}\label{eq:e0}
\begin{split}
	\det[D_m(1)]= \det \left[ (-\alpha_k + \frac{d}{d T} )^{j-1} \frac{F_1(T;-\alpha_k;1)}{F_0(T)} \right].
\end{split}
\end{equation}

When $s\neq 1$, the only difference of $D_m(s)$ from $M$ is that 
the function $E_{w}$ (defined in \cite[Formula (3.9)]{Baik06}) is changed to 
$E_{w}^s (u):= s\big( \frac1{1-s\mathbb{A}_x}\widetilde{C}_w\big) (u)$. 
The proof of Baik \cite[Formula (3.46)]{Baik06} goes through without any changes and
we obtain 
\begin{equation}\label{eq:e1}
\begin{split}
	\det[D_m(s)]= \det \left[ (-\alpha_k + \frac{d}{d T} )^{j-1} \frac{F_1(T;\alpha_k;s)}{F_0(T; s)} \right].
\end{split}
\end{equation}
From the definition of $F_k$~\eqref{eq:defination_of_generalized_TW_dist_rank_k}, 
we obtain~\eqref{eq:result_of_thm:critical1}.

\section{Non-vanishing property of some determinants}
\label{sec:proof_of_non_vanishing_property}

As discussed in Section~\ref{sec:org}, 
in each of the Sections~\ref{sec:proof_of_sub_and_sup_1}-\ref{sec:superjump} 
we need the fact that the determinant of a certain matrix is nonzero and is uniformly bounded away from zero.  
Specifically, we need this property for the following four matrices $[\tilde{\M}_{j,n}(c(\aaa_k))]_{j,k=1}^{\m}$ in~\eqref{eq:ratio_of_det_formula_sub} of Section~\ref{sec:proof_of_sub_and_sup_1}, $\frakP$ in~\eqref{eq:defn_of_frakP_simple-super1} of
Section~\ref{sec:proof_of_sub_and_sup_1.1}, 
$\calP = [ \M^{(i-1)}_{j,n}(x_0(a)) ]^{\m}_{j,i = 1}$
in~\eqref{eq:calPPa} of
Section~\ref{subsection:multiple_interwining},
and $\bfP_{\ell}=\frakP^{(a,\m - j),(b,j)}$ in~\eqref{eq:superj1030} (see also the discussion after~\eqref{eq:super213} in Section~\ref{sec:superjump}).

In this section, we prove that the determinants of these matrices 
are uniformly away from zero in a unifying way.
This was obtained  by considering 
a more general matrix which includes the above matrices as special cases. 
We can show the nonvanishing property from  
a direct algebraic manipulation of the determinant when the support of the equilibrium measure consists of a single interval (i.e. $N=0$) since in this case the entries of the matrix do not depend on $n$ 
and are expressed in terms of a simple rational function.
However, when the support consists of multiple intervals (i.e. $N>0$), the entries involve a Riemann theta function and it is not easy to check directly that the determinant is nonzero. 

Instead we proceed as follows. 
The entries of the desired matrix are expressed in terms of the solution of 
the so-called ``global parametrix''  Riemann--Hilbert problem (RHP) for orthogonal polynomials.
Using this, we show that the desired determinant itself can be expressed as a product of the 
solutions of different RHPs, which are a Darboux-type transformation of the above global parametrix RHP. 
We exploit a relationship between the original RHP and its transformation in order to prove the nonvanishing property.

\bigskip

We now introduce the general matrix which we analyze.
Let $\M_{j,n}(z)$ and $\tilde{\M}_{j,n}(z)$ be defined in \cite[Formula (311) and (312)]{Baik-Wang10a}.
They are  expressed in terms of the solution to the global parametrix RHP, 
see~\eqref{eq:relation_of_cal_M_and_normal_M0} and~\eqref{eq:first_simple_RHP} 
for the explicit formula. 
We note that they are analytic, in particular, for $z \in (\redge, \infty)$. 
We also note that when $N=0$ (see~\eqref{eq:Mjindepn})
\begin{equation}
	\M_{j,n}(z)=  \sqrt{\frac{2}{\pi(\redge-\ledge)}} \frac{\gamma(z)+\gamma(z)^{-1}}2
	\bigg( \frac{\gamma(z)-\gamma(z)^{-1}}{\gamma(z)+\gamma(z)^{-1}} \bigg)^j, 
	\qquad z\in \compC\setminus(-\infty, \redge]
\end{equation}
for $z$ in $\compC\setminus(-\infty, \redge]$. 

Let $\{c_1, \cdots, c_p\}$ be a set distinct real numbers in $(\redge, \infty)$ and let $\{d_1, \cdots, d_q\}$ be another set of distinct real numbers in $z \in (\redge, \infty)$ for some nonnegative integers $p$ and $q$. For each $n$, we define the $(p+q) \times (p+q)$ matrix
\begin{equation} \label{eq:defn_of_frakP_simple}
	\frakP^{c_1, \cdots, c_p}_{d_1, \cdots, d_q}
	:= \begin{bmatrix}
	\tilde{\M}_{1,n}(d_1) &\cdots & \tilde{\M}_{1,n}(d_q) & \M_{1,n}(c_1) &\cdots & \M_{1,n}(c_p) \\ 
	\vdots & \ddots & \vdots & \vdots & \ddots &\vdots \\
	\tilde{\M}_{p+q,n}(d_1) &\cdots & \tilde{\M}_{p+q,n}(d_q) & \M_{p+q,n}(c_1) &\cdots & \M_{p+q,n}(c_p) \\ 
	\end{bmatrix}.
\end{equation}
Special cases of this matrix appeared in the proofs of Theorem~\ref{thm:sub} in Section~\ref{sec:proof_of_sub_and_sup_1} and of Theorem~\ref{thm:sup1} in Section~\ref{sec:proof_of_sub_and_sup_1.1}.

We also consider a slight extension of the above matrix whose special cases appeared in the proofs of Theorem~\ref{thm:sup2} in Section~\ref{subsection:multiple_interwining} and of Theorem~\ref{thm:superjump} in Section~\ref{sec:superjump}. 
Let $m_1, \cdots, m_p$ and $n_1, \cdots, n_q$ be positive integers, and set $\mathbf{s} := m_1 + \cdots +m_p$ and $\mathbf{t} = n_1 + \cdots + n_q$. 
For each $n$, define the $(\mathbf{s}+\mathbf{t}) \times (\mathbf{s}+\mathbf{t})$ 
matrix  $\frakP^{(c_1, m_1), \cdots, (c_p, m_p)}_{(d_1,n_1), \cdots, (d_q,n_q)}$ by the entries, 
for each $j = 1, \cdots, \mathbf{s}+\mathbf{t}$, 
\begin{equation} \label{eq:defn_of_frakP_confluent}
  \left( \frakP^{(c_1, m_1), \cdots, (c_p, m_p)}_{(d_1,n_1), \cdots, (d_q,n_q)}\right)_{j,k} := 
  \begin{cases}
    \tilde{\M}^{(k-1)}_{j,n}(d_1) & \textnormal{for $k = 1, \cdots, n_1$,} \\
    \tilde{\M}^{(k-n_1-1)}_{j,n}(d_2) & \textnormal{for $k = n_1+1, \cdots, n_1+n_2$,} \\
    \vdots & \\
    \tilde{\M}^{(k - \mathbf{t}+n_q-1)}_{j,n}(d_q) & \textnormal{for $k = \mathbf{t}- n_q+1, \cdots, \mathbf{t}$,} \\
    \M^{(k-\mathbf{t}-1)}_{j,n}(d_1) & \textnormal{for $k = \mathbf{t}+1, \cdots, \mathbf{t} + m_1$,} \\
    \vdots & \\
    \M^{(k-\mathbf{s}-\mathbf{t}+m_p-1)}_{j,n}(c_p) & \textnormal{for $k = \mathbf{s}+\mathbf{t}-m_q+1, \cdots, \mathbf{s} + \mathbf{t}$.}
  \end{cases}
\end{equation}
Note that $\frakP^{c_1, \cdots, c_p}_{d_1, \cdots, d_q}$ is a special case of $\frakP^{(c_1, m_1), \cdots, (c_p, m_p)}_{(d_1,n_1), \cdots, (d_q,n_q)}$ when 
all $m_j=1$ and $n_j=1$. 
The main result of this section is the following proposition.

\begin{prop} \label{prop:non-vanishing}
  Let $p,q$ be nonnegative integers, $c_1, \cdots, c_p$ be a set of distinct real numbers in $(\redge, \infty)$ and $d_1, \cdots, d_q$ be another set of distinct real numbers in $(\redge, \infty)$. 
  \begin{enumerate}[label=(\alph*)]
  \item \label{enu:prop:non-vanishing:a}
 Let $\frakP^{c_1, \cdots, c_p}_{d_1, \cdots, d_q}$ be defined in \eqref{eq:defn_of_frakP_simple}. 
 Then for all positive integer $n$, $\det [\frakP^{c_1, \cdots, c_p}_{d_1, \cdots, d_q}] \neq 0$.
 Also  both $\det [\frakP^{c_1, \cdots, c_p}_{d_1, \cdots, d_q}]$ and its reciprocal are bounded uniformly in $n$. Moreover, if $c_1 < \cdots < c_p$ and $d_1 < \cdots < d_q$, then $(-1)^{pq+p(p-1)/2}(-i)^q \det [\frakP^{c_1, \cdots, c_p}_{d_1, \cdots, d_q}] > 0$.
  \item \label{enu:prop:non-vanishing:b}
    Let $m_1, \cdots, m_p$ and $n_1, \cdots, n_q$ be positive integers and 
    let $\frakP^{(c_1, m_1), \cdots, (c_p, m_p)}_{(d_1,n_1), \cdots, (d_q,n_q)}$ be defined by \eqref{eq:defn_of_frakP_confluent}. 
    Then for all positive integer $n$, $\det[\frakP^{(c_1, m_1), \cdots, (c_p, m_p)}_{(d_1,n_1), \cdots, (d_q,n_q)}] \neq 0$.
    Also both $\det [\frakP^{(c_1, m_1), \cdots, (c_p, m_p)}_{(d_1,n_1), \cdots, (d_q,n_q)}]$ and its reciprocal are bounded uniformly in $n$. Moreover, if $c_1 < \cdots < c_p$ and $d_1 < \cdots < d_q$, then $(-1)^{\mathbf{s}\mathbf{t} + \mathbf{s}(\mathbf{s}-1)/2}(-i)^{\mathbf{t}}\det [\frakP^{(c_1, m_1), \cdots, (c_p, m_p)}_{(d_1,n_1), \cdots, (d_q,n_q)}] > 0$, where $\mathbf{s} = m_1 + \cdots + m_p$ and $\mathbf{t} = n_1 + \cdots + n_q$.
  \end{enumerate}
\end{prop}

Even though Proposition \ref{prop:non-vanishing} \ref{enu:prop:non-vanishing:a} is a special case of Proposition \ref{prop:non-vanishing} \ref{enu:prop:non-vanishing:b}, we state these results separately for the ease of citation. 

 The idea of this proof is motivated by the paper \cite{Baik-Deift-Strahov03} which evaluates the orthogonal polynomials and their Cauchy transforms with respect to a weight which is a multiplication 
of a given weight  by a rational function. 
This procedure bears resemblance to the Darboux transformation in spectral theory.

\begin{rmk}
In this section, we use the abbreviation `$f_n  \asymp O(1)$ uniformly in $n$' to mean that for a sequence  $f_n$, both $f_n$ and $\frac1{f_n}$ are bounded uniformly in $n$. 
\end{rmk}

\subsection{Proof of Proposition \ref{prop:non-vanishing}}

We first prove part \ref{enu:prop:non-vanishing:a}.

Let $J := \bigcup^N_{j=0} (b_j, a_{j+1})$, $b_0 < a_1 < \cdots < a_{N+1}$, 
be the support of the equilibrium measure given in \eqref{eq:defn_of_J}.
From \cite[Formulas (311) and (312)]{Baik-Wang10a}, 
\begin{equation} \label{eq:relation_of_cal_M_and_normal_M0}
  \M_{k,n}(z) = C_{k,n} \cdot [\Mk_k ]_{11}(z) , \quad \tilde{\M}_{k,n}(z) = C_{k,n} \cdot [\Mk_k]_{12}(z)
\end{equation}
for any $z \in \compC \setminus (b_0, a_{N+1})$, where 
the constant $C_{k,n} :=\hat{\gamma}_{n-k} e^{n\ell/2}$ and the $2\times 2$ matrix $\Mk_k(z):= M_{k,n}^{(\infty)}(z)$ satisfy the following properties. 

First, the positive number $\hat{\gamma}_{n-k}$ is defined in \cite[Formula (304)]{Baik-Wang10a} 
in terms of the Riemann theta function $\theta$.
This particular theta function satisfies the property that $\theta(V)\neq 0$ for all real vector $V$ from (the proof of) \cite[Formula (3.38)]{Deift-Its-Zhou97}. 
Hence due to the periodicity of the Riemann theta function, $|\theta(V)|$ is uniformly bounded below and above for real vectors $V$. 
Since all the arguments of the Riemann theta functions in the definition of $\hat{\gamma}_{n-k}$ are real, we find that 
\begin{equation} \label{eq:hat_gamma_asym_O(1)}
  	C_{k,n} \asymp O(1)
\end{equation}
uniformly in $n$.

The matrix $\Mk_k(z):=M_{k,n}^{(\infty)}(z)$ is explicitly defined in 
\cite[Formulas (300) and (301)]{Baik-Wang10a}) in terms of a Riemann theta function. 
However, we do not use this formula; instead  we use the following Riemann--Hilbert characterization 
given in \cite[Formulas (295)--(297)]{Baik-Wang10a}.
Let $\vf(x)$  for $x\in \Sigma:=(b_0, a_{N+1})\setminus\{b_1, \cdots, b_N, a_1, \cdots, a_N\}$ 
be the jump matrix defined by 
\begin{equation} \label{eq:defn_of_v_f}
  	\vf(x) :=
  \begin{cases}
    \begin{bmatrix}
      e^{-in\Omega_j} & 0 \\
      0 & e^{in\Omega_j}
    \end{bmatrix}, &
    \textnormal{for $x \in (a_j, b_j)$, $j = 1, \cdots, N$,} \\
    \begin{bmatrix}
      0 & 1 \\
      -1 & 0
    \end{bmatrix}, & \textnormal{for $x \in J$,}
  \end{cases}
\end{equation}
where $\Omega_j$, $j=1, \cdots, N$, are real constants defined in \cite[Formula (1.21)]{Deift-Kriecherbauer-McLaughlin-Venakides-Zhou99} (see also \cite[Table 1]{Baik-Wang10a}).
Then for $k\in\mathbb{Z}$, $\Mk_k(z):=M_{k,n}^{(\infty)}(z)$ solves the following RHP:
\begin{equation}\label{eq:first_simple_RHP}
\begin{cases}
  	&\Mk_k(z) \textnormal{ is analytic in $\compC \setminus \overline{\Sigma}$ and is continuous up to the boundary},  \\
  	&[\Mk_k]_+(x) = [\Mk_{k}]_-(x) \vf(x) \textnormal{ for } x \in \Sigma, \\
  	&\Mk_k(z) = [I_{2 \times 2} + O(z^{-1})]
  \begin{bmatrix}
    z^{-k} & 0 \\
    0 & z^k
  \end{bmatrix} \textnormal{ as } z \to \infty.
\end{cases}
\end{equation}
Let us denote 
\begin{equation}
  	\Mk_{k}(z) =
  \begin{bmatrix}
    \xi^{(1)}_{k}(z) & \eta^{(1)}_{k}(z) \\
    \xi^{(2)}_{k}(z) & \eta^{(2)}_{k}(z)
  \end{bmatrix}
\end{equation}
so that~\eqref{eq:relation_of_cal_M_and_normal_M0} become
\begin{equation} \label{eq:relation_of_cal_M_and_normal_M}
 	 \M_{k,n}(z) = C_{k,n} \xi^{(1)}_k (z) , \quad \tilde{\M}_{k,n}(z) = C_{k,n} \eta^{(1)}_k(z).
\end{equation}
Note that even though we do not explicitly indicate it, 
$ \xi^{(1)}_k (z)$ and $\eta^{(1)}_k(z)$ depend on $n$. 

From the hypothesis of Proposition \ref{prop:non-vanishing}\ref{enu:prop:non-vanishing:a}, $c_1, \cdots, c_p$ are distinct real numbers and $d_1, \cdots, d_q$ is another set of distinct real numbers, all in $(\redge,\infty)$. 
For integers $0 \leq s \leq p$ and $0 \leq t \leq q$, 
let $(\Mk_{k})^{c_1, \cdots, c_s}_{d_1, \cdots, d_t}(z)$ be the $2 \times 2$ matrix   whose entries are defined by, for each $i = 1,2$, 
\begin{multline} \label{eq:the_first_column_of_big_M}
  [(\Mk_k)^{c_1, \cdots, c_s}_{d_1, \cdots, d_t}]_{i1}(z) := \prod^s_{j=1} \frac{1}{c_j - z} \\
  \times \det
  \begin{pmatrix}
    \xi^{(i)}_{k-s}(c_1) & \cdots & \xi^{(i)}_{k-s}(c_s) & \eta^{(i)}_{k-s}(d_1) & \cdots & \eta^{(i)}_{k-s}(d_t) & \xi^{(i)}_{k-s}(z) \\
    \xi^{(i)}_{k-s+1}(c_1) & \cdots & \xi^{(i)}_{k-s+1}(c_s) & \eta^{(i)}_{k-s+1}(d_1) & \cdots & \eta^{(i)}_{k-s+1}(d_t) & \xi^{(i)}_{k-s+1}(z) \\
    \vdots & & \vdots & \vdots & & \vdots & \vdots \\
    \xi^{(i)}_{k+t}(c_1) & \cdots & \xi^{(i)}_{k+t}(c_s) & \eta^{(i)}_{k+t}(d_1) & \cdots & \eta^{(i)}_{k+t}(d_t) & \xi^{(i)}_{k+t}(z)
  \end{pmatrix},
\end{multline}
and
\begin{multline} \label{eq:the_second_column_of_big_M}
  [(\Mk_k)^{c_1, \cdots, c_s}_{d_1, \cdots, d_t}]_{i2} (z) := \prod^t_{j=1} \frac{1}{d_j - z} \\
  \times \det
  \begin{pmatrix}
    \xi^{(i)}_{k-s}(c_1) & \cdots & \xi^{(i)}_{k-s}(c_s) & \eta^{(i)}_{k-s}(d_1) & \cdots & \eta^{(i)}_{k-s}(d_t) & \eta^{(i)}_{k-s}(z) \\
    \xi^{(i)}_{k-s+1}(c_1) & \cdots & \xi^{(i)}_{k-s+1}(c_s) & \eta^{(i)}_{k-s+1}(d_1) & \cdots & \eta^{(i)}_{k-s+1}(d_t) & \eta^{(i)}_{k-s+1}(z) \\
    \vdots & & \vdots & \vdots & & \vdots & \vdots \\
    \xi^{(i)}_{k+t}(c_1) & \cdots & \xi^{(i)}_{k+t}(c_s) & \eta^{(i)}_{k+t}(d_1) & \cdots & \eta^{(i)}_{k+t}(d_t) & \eta^{(i)}_{k+t}(z)
  \end{pmatrix}.
\end{multline}

Now we proceed to prove Proposition~\ref{prop:non-vanishing}\ref{enu:prop:non-vanishing:a} as follows. 
We consider only the case when $q > 0$. The proof is completely analogous when $q=0$. 
When $q>0$, from~\eqref{eq:relation_of_cal_M_and_normal_M} and~\eqref{eq:the_second_column_of_big_M}, we have 
  \begin{equation} \label{eq:relation_between_detQ_and_big_RHP_det}
    \det [\frakP^{c_1, \cdots, c_p}_{d_1, \cdots, d_q}] = (-1)^{pq} \left( \prod^{p+q}_{j=1} C_{j,n} \right) \left( \prod^{q-1}_{\ell=1} (d_{\ell} - d_q) \right) 
    ( (\Mk_{p+1})^{c_1, \cdots, c_p}_{d_1, \cdots, d_{q-1}} )_{12}(d_q).
  \end{equation}
To show that $\det [\frakP^{c_1, \cdots, c_p}_{d_1, \cdots, d_q}] \asymp O(1)$, 
we only need to prove that 
$[ (\Mk_{p+1})^{c_1, \cdots, c_p}_{d_1, \cdots, d_{q-1}} ]_{12}(d_q) \asymp O(1)$ uniformly in $n$
due to~\eqref{eq:hat_gamma_asym_O(1)}.
We prove this by showing that  
for all $s \in \{ 0, 1, \cdots, p \}$ and $t \in \{ 0, 1, \cdots, q \}$,
\begin{equation} \label{eq:M_k^upper_lower_asymp_O(1)}
  [ (\Mk_k)^{c_1, \cdots, c_s}_{d_1, \cdots, d_t} ]_{i,j}(x) \asymp O(1) 
\end{equation}
uniformly in $n$ for each $i,j = 1,2$, for each integer $k$,  and for each real number $x \in (\redge, \infty)$, using an induction in $s$ and $t$.

When $s = t = 0$, $\Mk_k(z)$ has an explicit formula  in terms of the Riemann theta function $\theta$ \cite[Formulas (300) and (301)]{Baik-Wang10a}. Note that for $x\in (\redge, \infty)$, $u_{\pm}(x)$ is a real vector by the construction of $u$ defined in \cite[Formula (1.29)]{Deift-Its-Zhou97} and \cite[Formula (1.29)]{Deift-Kriecherbauer-McLaughlin-Venakides-Zhou99}. Since all the arguments of the Riemann theta functions are real vectors, we find, as in the discussion above~\eqref{eq:hat_gamma_asym_O(1)}, that $[\Mk_k]_{ij}(x) \asymp O(1)$ 
for each $i,j = 1,2$, for each integer $k$,  and for each real number $x \in (\redge, \infty)$

\medskip

To complete the induction step, it is enough to show that 
if~\eqref{eq:M_k^upper_lower_asymp_O(1)} holds for $s = \ell$ and $t = \ell'$, 
then it  holds for $s = \ell+1, t = \ell'$, and also for $s = \ell, t = \ell'+1$.
Recall that the $(j,\ell)$-minor of a square matrix is the determinant of the matrix formed by removing the $j$-th row and $\ell$-th column from the original matrix. 
For  $i=1,2$, we denote the $(1, s+t+1)$-minor of the matrix on the right-hand side of \eqref{eq:the_first_column_of_big_M} by $\mn_{i1}$.
We also denote  the $(s+t+1, s+t+1)$-minor of the matrix on the right-hand side of \eqref{eq:the_second_column_of_big_M} by $\mn_{i2}$. 
It is easy to see from the definition that and when $t > 0$, 
\begin{align}
  	\mn_{i1} = & [ (\Mk_{k+1})^{c_1, \cdots, c_s}_{d_1, \cdots, d_{t-1}} ]_{i2}(d_t) \prod^{t-1}_{j=1} (d_j - d_t), \label{eq:defn_of_A_i1_second} \\
  	\mn_{i2} = & [ (\Mk_k)^{c_1, \cdots, c_s}_{d_1, \cdots, d_{t-1}} ]_{i2}(d_t)  \prod^{t-1}_{j=1} (d_j - d_t), \label{eq:defn_of_A_i2_second}
\end{align}
and when $t = 0$,
\begin{align}
  	\mn_{i1} = & [ (\Mk_k)^{c_1, \cdots, c_{s-1}} ]_{i1}(c_s) \prod^{s-1}_{j=1} (c_j - c_s), \label{eq:defn_of_A_i1_first} \\
  	\mn_{i2} = & [ (\Mk_{k-1})^{c_1, \cdots, c_{s-1}} ]_{i1}(c_s)  \prod^{s-1}_{j=1} (c_j - c_s).\label{eq:defn_of_A_i2_first}
\end{align}

Now let $(\Mkh_{k})^{c_1, \cdots, c_s}_{d_1, \cdots, d_t}(z)$ be the solution to the RHP~\eqref{eq:first_simple_RHP} where the jump matrix is changed to $\vfh(x)$ which is given by $\vfh(x)=\vf(x)$ for $x\in \Sigma\setminus J$, and   
\begin{equation} \label{eq:defn_of_v_fh}
  \vfh(x) =
  \begin{bmatrix}
    0 &  \frac{(c_1-x) \cdots (c_s-x)}{(d_1-x) \cdots (d_t-x)} \\
    - \frac{(d_1-x) \cdots (d_t-x)}{(c_1-x) \cdots (c_s-x)} & 0
  \end{bmatrix},
  \quad \textnormal{for $x \in J$.}
\end{equation}
The existence of the solution to this RHP is given in the next subsection. 
The uniqueness follows from the fact that $\det \vfh\equiv 1$.

\begin{lemma} \label{lemma:construction_of_RHP_big} 
  \begin{enumerate}[label={(\alph*)}]
  \item
    For $s>0$, if $((\Mk_k)^{c_1, \cdots, c_{s-1}})_{11}(c_s) \neq 0$ and $((\Mk_{k-1})^{c_1, \cdots, c_{s-1}})_{21}(c_s) \neq 0$, then 
    \begin{equation} \label{eq:lemma:construction_of_RHP_big:1.0}
      (\Mk_{k})^{c_1, \cdots, c_s}(z) = \diag(\mn_{11}, \mn_{22}) (\Mkh_{k})^{c_1, \cdots, c_s}(z),
    \end{equation}
    where $\mn_{11}$ and $\mn_{22}$ are given in \eqref{eq:defn_of_A_i1_first} and \eqref{eq:defn_of_A_i2_first}.
  \item \label{enu:assumption_a_of_A_ii}
    For  $t > 0$, if $((\Mk_{k+1})^{c_1, \cdots, c_s}_{d_1, \cdots, d_{t-1}})_{12}(d_t) \neq 0$ and $((\Mk_{k})^{c_1, \cdots, c_s}_{d_1, \cdots, d_{t-1}})_{22}(d_t) \neq 0$, then 
    \begin{equation} \label{eq:lemma:construction_of_RHP_big:1}
      (\Mk_{k})^{c_1, \cdots, c_s}_{d_1, \cdots, d_t}(z) = (-1)^t \diag( \mn_{11}, \mn_{22}) (\Mkh_{k})^{c_1, \cdots, c_s}_{d_1, \cdots, d_t}(z),
    \end{equation}
    where $\mn_{11}$ and $\mn_{22}$ are given in \eqref{eq:defn_of_A_i1_second} and \eqref{eq:defn_of_A_i2_second}.
  \end{enumerate}
\end{lemma}

\begin{proof}
This is straightforward to check. 
\end{proof}

From the RHP, we can show the nonvanishing property of the entries of $ (\Mkh_{k})^{c_1, \cdots, c_s}_{d_1, \cdots, d_t}$. 
 
\begin{lemma} \label{lemma:RHP}
  For any integer $k$, real number $x \in (\redge, \infty)$, $s \in \{ 0, 1, \cdots, p \}$ and $t \in \{ 0, 1,\cdots, q \}$,
\begin{equation} \label{eq:M_hat^upper_asym_O(1)}
  [ (\Mkh_{k})^{c_1, \cdots, c_s}_{d_1, \cdots, d_t} ]_{ij}(x) \asymp O(1), \quad i,j = 1,2,
\end{equation}
uniformly in $n$ and $[ (\Mkh_{k})^{c_1, \cdots, c_s}_{d_1, \cdots, d_t} ]_{11}(x)$, $[ (\Mkh_{k})^{c_1, \cdots, c_s}_{d_1, \cdots, d_t} ]_{22}(x)$, $[ -i(\Mkh_{k})^{c_1, \cdots, c_s}_{d_1, \cdots, d_t} ]_{12}(x)$ and $i[ (\Mkh_{k})^{c_1, \cdots, c_s}_{d_1, \cdots, d_t} ]_{21}(x)$ 
are all positive.
\end{lemma}

\begin{proof}
The proof will be given in Section \ref{subsec:solution_to_a_generalized_RHP}.
\end{proof}

Thus Lemmas \ref{lemma:construction_of_RHP_big} and \ref{lemma:RHP} imply that
if~\eqref{eq:M_k^upper_lower_asymp_O(1)} holds for $s = \ell$ and $t = \ell'$, 
and so does  for $s = \ell+1, t = \ell'$, and also for $s = \ell, t = \ell'+1$.
The induction step of the proof is complete and we obtain 
$\det [\frakP^{c_1, \cdots, c_p}_{d_1, \cdots, d_q}] \asymp O(1)$.

If we take inductive steps, in particular,  as $(s,t) = (0,0)$, $(1,0)$, $\cdots$, $(p,0)$, $(p,1)$, $\cdots$, $(p, q-1)$,
then we find an explicit formula of  
$[ (\Mk_{p+1})^{c_1, \cdots, c_p}_{d_1, \cdots, d_{q-1}} ]_{12}(d_q)$, 
which implies from~\eqref{eq:relation_between_detQ_and_big_RHP_det} that 
\begin{multline} \label{eq:ultimate_formula_of_det_Q}
  \frac{(-1)^{pq} \left( \prod^{p+q}_{j=1} C^{-1}_{j,n} \right) \det[\frakP^{c_1, \cdots, c_p}_{d_1, \cdots, d_q}]}{\left( \prod_{1 \leq j < k \leq p} (c_j - c_k) \right) \left( \prod_{1 \leq j < k \leq q-1} (d_k - d_j) \right)} \\
  = \prod^q_{t=1} [(\Mkh_{p+1})^{c_1, \cdots, c_p}_{d_1, \cdots, d_{t-1}}]_{12}(d_t) \prod^p_{s=1} [(\Mkh_s)^{c_1, \cdots, c_{s-1}}]_{11}(c_s).
\end{multline}
From this formula and the signs of $[(\Mkh_k)^{c_1, \cdots, c_s}_{d_1, \cdots, d_t}]_{i,j}(x)$ i
n Lemma \ref{lemma:RHP}, 
we find that if $c_j$ and $d_j$ are both in ascending orders, then $(-1)^{pq+p(p-1)/2}(-i)^q \det [\frakP^{c_1, \cdots, c_p}_{d_1, \cdots, d_q}] > 0$. 
This complete a proof of Proposition \ref{prop:non-vanishing}\ref{enu:prop:non-vanishing:a}.

\medskip
We now consider Proposition \ref{prop:non-vanishing}\ref{enu:prop:non-vanishing:b}.
Note that the identity~\eqref{eq:ultimate_formula_of_det_Q} is analytic in $c_j$'s and $d_j$'s. Hence if we take the limit so that some of $c_j$ are identical and some of $d_j$ are identical, 
then by \lHopital's rule, we obtain Proposition \ref{prop:non-vanishing}\ref{enu:prop:non-vanishing:b}.

\subsection{Evaluation of $(\Mkh_k)^{c_1, \cdots, c_s}_{d_1, \cdots, d_t}$}
\label{subsec:solution_to_a_generalized_RHP}

In this section we prove Lemma \ref{lemma:RHP}
by finding an explicit formula of $(\Mkh_k)^{c_1, \cdots, c_s}_{d_1, \cdots, d_t} (z)$, which is obtained by solving an
RHP in terms of a Riemann theta function


We can consider the following slightly more general RHP. 
Let $\Sigma:=(b_0, a_{N+1})\setminus\{b_1, \cdots, b_N, a_1, \cdots, a_N\}$ 
and $J = \bigcup^N_{j=0} (b_j, a_{j+1})$ as in \eqref{eq:first_simple_RHP}. 
Let $f(x)$ be a positive real analytic function on $\bar{J}$.
Let the $2 \times 2$ matrix  $\NN(z)$ be the solution to the following RHP: 
\begin{equation}\label{eq:first_simple_RHPq} 
\begin{cases}
  &\NN(z) \textnormal{ is analytic in $\compC \setminus \overline{\Sigma}$ and is continuous up to the boundary}, \\
  &\NN_+(x) = \NN_-(x)
  \begin{bmatrix}
    e^{-in\Omega_j} & 0 \\
    0 & e^{in\Omega_j}
  \end{bmatrix}
  \textnormal{ for } x \in \Sigma\setminus J, \\
  &\NN_+(x) = \NN_-(x)
  \begin{bmatrix}
    0 & f(x) \\
    -1/f(x) & 0
  \end{bmatrix}
  \textnormal{ for } x \in J, \\
  &\NN(z) = [I_{2 \times 2} + O(z^{-1})]
  \begin{bmatrix}
    z^{-k} & 0 \\
    0 & z^k
  \end{bmatrix}
  \textnormal{ as } z \to \infty.
\end{cases}
\end{equation}
The matrix $(\Mkh_k)^{c_1, \cdots, c_s}_{d_1, \cdots, d_t} (z)$ is the special case of $\NN(z)$ when $f(x)$ is rational. 

We now solve the the above RHP for $\NN(z)$ explicitly. This is done by finding an algebraic transformation of $\NN$ so that the jump matrix on $J$ becomes $\big( \begin{smallmatrix}
      0 & 1 \\
      -1 & 0
    \end{smallmatrix}\big)$ while the jump matrix on $\Sigma\setminus J$ remains similar to the original one except that each $\Omega_j$ changes to a different constant. The asymptotic condition as $z\to\infty$ is unchanged. 
The solution to the resulting RHP is well known \cite{Deift-Kriecherbauer-McLaughlin-Venakides-Zhou99}.
    
For constants $t_1, \cdots, t_N$, let $D(z)$ be a solution to the following scalar RHP:
\begin{equation}\label{eq:Drhp}
\begin{cases}
	&\text{$D(z)$ is analytic in $\compC\setminus \overline{\Sigma}$ and is continuous up to the boundary,}\\
	&\text{$D(z)\neq 0$ for all $z\in \compC\setminus \overline{\Sigma}$,}\\
  	&D_+(x) D_-(x) = f(x) \quad \textnormal{for } x \in J, \\
  	&D_+(x) = D_-(x) e^{it_j}  \quad \textnormal{for } x \in (a_j, b_j), \quad j = 1, \cdots, N, \\
  	&D(\infty):=\lim_{z \to \infty} D(z) \textnormal{ exists and is non-zero.} 
\end{cases}
\end{equation}
For most choices of $t_1, \cdots, t_N$, there is no solution to this RHP. Below we construct a (unique) array of $t_1, \cdots, t_N$ for which $D(z)$ exists. Note that $D(z)$ is unique if it exists. 

Set 
\begin{equation} \label{eq:G_by_D}
  	L(z) = \log D(z),\qquad z\in \compC\setminus\Sigma
\end{equation}
where $\log$ is defined on  the principal branch of logarithm. The RHP for $D(z)$ implies that 
$L(z)$ is a well-defined analytic function in $\compC \setminus \Sigma$. 
Set $q(z) := \prod^N_{j=0} (z-b_j)(z-a_{j+1})$. 
Define the square root $\sqrt{q(z)}$ to be analytic in $\compC \setminus \bar{J}$ and satisfy $\sqrt{q(z)} \sim z^N$ as $z \to \infty$. Then the function
\begin{equation} \label{eq:G_tilde_by_G}
  \tilde{L}(z) = \frac{L(z)}{\sqrt{q(z)}}
\end{equation}
satisfies the following scalar RHP: $\tilde{L}(z)$ is analytic in  $z\in \compC\setminus \overline{\Sigma}$, continuous up to the boundary of $\Sigma$, and 
$\tilde{L}(z)= O((z-b_j)^{-1/2}$ and $\tilde{L}(z)= O((z-a_{j+1}))^{-1/2}$ for $j=0, \cdots, N$, and also 
\begin{equation}\label{eq:Ltilde}
\begin{cases}
 	&\tilde{L}_+(x) = \tilde{L}_-(x) + \frac{\log f(x)}{(\sqrt{q(x)})_+} \quad \textnormal{for $x \in J$,}  \\
  	& \tilde{L}_+(x)= \tilde{L}_-(x) + it_j \quad \textnormal{for $x \in (a_j, b_j)$, $j = 1, \cdots, N$,} \\
 	&\tilde{L}(z)= O(z^{-N-1}) \quad \textnormal{as $z \to \infty$.}
\end{cases}
\end{equation}
The additive jump conditions imply that, from the Plemelj formula, 
\begin{equation} \label{eq:formula_of_tilde_L}
  	\tilde{L}(z) = \frac{1}{2\pi i} \left( \sum^N_{j=0} \int^{a_{j+1}}_{b_j} \frac{\log f(s)}{(\sqrt{q(s)}_+)} \frac{ds}{s-z} + \sum^N_{j=1} it_j \int^{b_j}_{a_j} \frac{ds}{s-z} \right) +E(z)
\end{equation}
for an entire function $E(z)$. Now in order to satisfy the asymptotic condition $\tilde{L}(z)= O(z^{-N-1})$ as $z\to \infty$, we must have that $E(z)\equiv 0$ and 
\begin{equation} \label{eq:system_of_linear_equations}
  \sum^N_{j=0} \int^{a_{j+1}}_{b_j} \frac{\log f(s)}{i(\sqrt{q(s)}_+)} s^k ds + \sum^N_{j=1} t_j \int^{b_j}_{a_j} s^k ds = 0 \quad \textnormal{for $k = 0, \cdots, N-1$.}
\end{equation}
We regard \eqref{eq:system_of_linear_equations} as a system of $N$ 
linear equations for $t_1, \cdots, t_N$. 
This system has 
a unique solution since its Jacobian is
\begin{equation} \label{eq:determinant_needed_to_be_nonzero}
  \det \left( \int^{b_j}_{a_j} s^{\ell-1} ds \right)^N_{j,\ell = 1} = \int^{b_1}_{a_1} ds_1 \cdots \int^{b_N}_{a_N} ds_N \det(s^{\ell-1}_j)^N_{j,\ell = 1},
\end{equation}
which is positive. 
For this particular $t_j$, the RHP \eqref{eq:Ltilde} has a solution, and accordingly the RHP \eqref{eq:Drhp} have a solution. 
Note that since $i(\sqrt{q(s)}_+)$ is real for $s\in J$, the above system of equations has real coefficients and hence the solution $t_j$ are real. 
From this and~\eqref{eq:formula_of_tilde_L}, we find that $D(x) > 0$ for all $x \in (\redge, \infty)$.

Set 
\begin{equation} \label{eq:defination_of_M_f,k_by_tilde_M_f,k}
  	\tilde{\NN}(z) := 
  \begin{bmatrix}
    D(\infty)^{-1} & 0 \\
    0 & D(\infty)
  \end{bmatrix} \NN(z)
  \begin{bmatrix}
    D(z) & 0 \\
    0 & D(z)^{-1}
  \end{bmatrix},
\end{equation}
for $z\in \Sigma\setminus J$. It is easy to check using~\eqref{eq:Drhp} that $\tilde{\NN}$ satisfies the same analytic and asymptotic condition as the RHP~\eqref{eq:first_simple_RHPq}
and it satisfies the jump condition  
\begin{equation}\label{eq:first_almost_simple_RHP} 
\begin{cases}
  	&\tilde{\NN}_+(x) = \tilde{\NN}_-(x)
        \begin{bmatrix}
      e^{-i(n\Omega_j-t_j)} & 0 \\
      0 & e^{i(n\Omega_j-t_j)}
    \end{bmatrix}
    \textnormal{ for } x \in \Sigma\setminus J, \\
  	&\tilde{\NN}_+(x) = \tilde{\NN}_-(x)
        \begin{bmatrix}
      0 & 1 \\
      -1 & 0
    \end{bmatrix}
    \textnormal{ for } x \in J.
\end{cases}
\end{equation}
This is the same RHP~\eqref{eq:first_simple_RHP} for $\Mk_k$ with the changes $\Omega_j\mapsto \Omega_j- \frac1{n}t_j$. 
Hence the solution $\tilde{\NN}(z)$ is given by the usual Riemann theta function construction. 
Since $t_j$'s are real, we find, as in the case of $\Mk_k$,  from the property of the theta function 
that for any $x \in (\redge, \infty)$, $[\tilde{\NN}]_{ij}(x) \asymp O(1)$ uniformly in $n$, $i,j = 1,2$. 
Also from the explicit formula of $\tilde{\NN}(z)$, we find that $[\tilde{\NN}]_{11}(x) > 0$, $-i[\tilde{\NN}]_{12}(x) > 0$. Since $D(x) > 0$ for $x \in (\redge, \infty)$, we find by \eqref{eq:defination_of_M_f,k_by_tilde_M_f,k} that $[\NN]_{ij}(x) \asymp O(1)$ uniformly in $n$, $i,j = 1,2$, and $[\NN]_{11}(x) > 0$, $-i[\NN]_{12}(x) > 0$. 

Thus, as a special case, we obtain Lemma \ref{lemma:RHP}.

\section{Proof of Theorem~\ref{thm:alg}}\label{sec:higher}

Theorem~\ref{thm:alg} is an algebraic relation that reduce the higher rank case to the rank one case. 
We give an elementary proof of  this theorem in this section. 
A different, more conceptual proof based on the integrable structure of the Hermitian matrix model with external source
can be found in \cite{Baik-Wang12}.


Since the proof is purely algebraic, 
we drop the dependence on $n$ in the density function~\eqref{eq:generalized_pdf} 
and consider the following matrix model. 
Let $W(x)$ is a nonnegative function on the real line such that $\log W(x)$ grows faster than any linear function as $|x|\to \infty$.
We also assume that the orthonormal polynomials $p_0(x), p_1(x), \cdots$ with respect to the weight $W(x)$ exist. 
Fix the matrix 
$A = \diag(a_1, \cdots, a_\ddd)$, and consider the following measure on the set $\HilH_\ddd$ of 
$\ddd \times \ddd$ Hermitian matrix $M$:
\begin{equation} \label{eq:unscaled_external_source_model}
	f_d (M)dM := \frac{1}{Z} \det(W(M)) e^{\Tr(AM)}dM.
\end{equation}
Here $W(M)$ is defined in terms of the continuous functional calculus of Hermitian matrices, and 
$Z$ is the normalization constant. 
We emphasize that  $\ddd$ is the dimension of both the random matrix $M$ 
and the external source matrix $A$. 

For a subset $E\subset \realR$ and $s \in \compC$, 
define
\begin{equation} \label{eq:defn_of_E_a_dots_E_s0}
  \begin{split}
    \fE_\ddd(a_1, \cdots, a_\ddd; E; s)  := 
    \int_{\HilH^{\ddd}} \prod^{\ddd}_{j=1} (1 - s\chi_E(\lambda_j)) f_d(M) dM,
  \end{split}
\end{equation}
where $\lambda_1, \cdots, \lambda_{\ddd}$ are the eigenvalues of $M$. 
When $A= \diag(a_1, \cdots, a_\m, \underbrace{0, \cdots, 0}_{d-\m})$ for some $\m \leq \ddd$, we 
suppress the zero eigenvalues of the external source matrix $A$ denote~\eqref{eq:defn_of_E_a_dots_E_s0} by
\begin{equation} \label{eq:defn_of_E_a_dots_E_s}
  \begin{split}
    \fE_\ddd(a_1, \cdots, a_\m; E; s) .
  \end{split}
\end{equation}
We also set
\begin{equation} \label{eq:defn_of_bar_E_a_dots_E_s}
  \bar{\fE}_\ddd(a_1, \cdots, a_\m; E; s) := \frac{\fE_\ddd(a_1, \cdots, a_\m; E; s)}{\fE_\ddd( E; s )}.
\end{equation}
where $\fE_\ddd( E; s )$ is~\eqref{eq:defn_of_E_a_dots_E_s0} when $A=0$. 
For a real number $a$, define (cf.~\eqref{eq:gadef})
\begin{equation}\label{eq:gatam}
  \ga_j(a):= \int_{\R} p_j (s) e^{a s} W(s) ds. 
\end{equation}

Theorem~\ref{thm:alg} follows from the following proposition when 
$W(x) = e^{-nV(x)}$ and $A=n\A_\ddd= \diag(n\aaa_1, \cdots, n\aaa_\m, 0, \cdots, 0)$.

\begin{prop}\label{thm:algE}
We have, assuming that $a_1, \cdots, a_\m$ are nonzero and distinct, 
\begin{equation} \label{eq:algE}
  \bar{\fE}_\ddd(a_1, \cdots,a_\m; E; s) 
  = \frac{ \det \big[ \ga_{\ddd-j}(a_k) \bar{\fE}_{\ddd-j+1}(a_k; E; s) \big]
  _{j,k=1}^{\m} }{\det[\ga_{\ddd-j}(a_k)]_{j,k=1}^{\m}}.
\end{equation}
\end{prop}

\begin{proof}
The density function of the unordered eigenvalues $\lambda_1, \cdots, \lambda_\ddd\in \R$ of $M$
induced from~\eqref{eq:unscaled_external_source_model} is a symmetric function and is given by 
\begin{equation} \label{eq:800}
	 \frac1{Z'} \det[\lambda_i^{j-1}]_{i,j=1}^d \frac{\det[ e^{a_i\lambda_j}]_{i,j=1}^d}{\det[a_i^{j-1}]_{i,j=1}^d} \prod_{i=1}^{\ddd} W(\lambda_i), \qquad (\lambda_1, \cdots, \lambda_d)\in \R^d
\end{equation}
where $Z'$ is the normalization constant. 
Here the ratio $\frac{\det[ e^{a_i\lambda_j}]}{\det[a_i^{j-1}]} $ is evaluated using \lHopital's rule
as $a_{\m+1}=\cdots = a_d=1$.  
Recall that $a_1, \cdots, a_\m$ are nonzero and distinct by assumption. 
From the usual random matrix theory, the eigenvalues form a
determinantal point process 
whose kernel $\Ka_d$ obtained from the bi-orthonormal system constructed from 
$\{1, x, x^2, \cdots, x^{\ddd-1}\}$ and $\{1, x, \cdots, x^{d-\m-1}, e^{a_1x}, \cdots, e^{a_\m x}\}$. 
Then 
\begin{equation}
	\fE_\ddd(a_1, \cdots, a_\m; E; s)= \det [1-sP\Ka_\ddd P]
\end{equation} 
where  $P$ denote the projection operator on the set $E$.

Let  $\Ko_\ddd(x,y)= \sum_{j=0}^{n-1} p_j(x)p_j(y) W(x)^{1/2}W(y)^{1/2}$
be the usual Christoffel--Darboux kernel when $A=0$. 
Hence $\fE_\ddd(E; s)= \det [1-sP\Ka P]$
and 
$\bar{\fE}_\ddd(a_1, \cdots, a_\m; E; s) =  \frac{\det [1-sP\Ka_\ddd P]}{\det[1-sP\Ko_\ddd P]}$.

Note that the first $d-\m$ terms of the bi-orthogonal functions for $\Ka_\ddd$ 
are the orthonormal polynomials $p_j(x)$. 
It was shown in \cite{Baik09} that $\Ka_\ddd$  is a rank $\m$ perturbation of $\Ko_\ddd$ as follows. 
Define the column vectors
\begin{equation}
\begin{split}
	\hat{\mathbf{t}}(x)  := {}& (p_{\ddd-\m}(x)W(x)^{1/2}, \cdots, p_{\ddd-1}(x)W(x)^{1/2})^t , \\
	\quad \hat{\mathbf{v}}(x)  := {}& (e^{a_1 x}W(x)^{1/2}, \cdots, e^{a_{\m} x}W(x)^{1/2})^t,
\end{split}
\end{equation}
and define an $\m \times \m$ matrix (the second equality follows from~\eqref{eq:gatam})
\begin{equation}\label{e:BBBB}
	\mathbf{B}
	:= \int_{\R} \hat{\mathbf{t}}(s) \hat{\mathbf{v}}(s)^t ds
	=  \big[ \ga_{\ddd-\m+j-1}(a_k) \big]_{j,k=1}^\m.
\end{equation}
Set 
\begin{equation}\label{eq:wintermv}
	\hat{\mathbf{w}}(z) := \hat{\mathbf{v}}(z)- \int_{\R} \Ko_\ddd(s,z) \hat{\mathbf{v}}(s) ds
	= \big[ (1-\Ko_\ddd)_{\R} \hat{\mathbf{v}} \big](z).
\end{equation}
Then (The formula is equivalent to \cite[Theorem 1]{Baik09} once one changes the monic orthogonal polynomials $\pi_j(x)$ to the orthonormal polynomials $p_j(x)$, and conjugate both sides of \cite[Formula (19)]{Baik09} by $W(x)^{1/2}=e^{-\frac12 V(x)}$.)
\begin{equation} \label{eq:kernel_for_external_source_model}
	\Ka_\ddd= \Ko_\ddd + \sum_{j=1}^{\m} 
	\hat{\mathbf{w}}_j \otimes (\mathbf{B}^{-1} \hat{\mathbf{t}})_j.
\end{equation}

We now derive~\eqref{eq:algE} from~\eqref{eq:kernel_for_external_source_model}.
Let $\langle, \rangle$ be the real inner product over $\R$.  
Then~\eqref{eq:kernel_for_external_source_model} implies that
\begin{equation}\label{eq:draK}
\begin{split}
	\bar{\fE}_\ddd(a_1, \cdots, a_\m; E; s) 
	= & \det\bigg[ 1-  s\sum_{j=1}^{\m} ((1-sP\Ko_\ddd P)^{-1}\hat{\mathbf{w}})_j \otimes (\mathbf{B}^{-1} \hat{\mathbf{t}})_jP\bigg] \\
	= & \det \bigg[ \mathbb{I} - s\langle (\mathbf{B}^{-1}\hat{\mathbf{t}}P, (1-sP\Ko_\ddd P)^{-1}\hat{\mathbf{w}}^t \rangle \bigg] \\
	= & \frac{1}{\det[\mathbf{B}]}  \det \bigg[ \mathbf{B} - s\langle \hat{\mathbf{t}},  P(1-sP\Ko_\ddd P)^{-1}\hat{\mathbf{w}}^t \rangle \bigg].
\end{split}
\end{equation}
Set $\psi_k(x):= p_k(x) W(x)^{1/2}$. Then $\hat{\mathbf{t}}_j= \psi_{\ddd-j}$. 
By  using the definition~\eqref{eq:wintermv} of $\hat{\mathbf{w}}$, we find that 
\begin{equation}\label{eq:algFredi-0}
	\bar{\fE}_\ddd(a_1, \cdots, a_\m; E; s) = 
	\frac{1}{\det[\mathbf{B}]}   \det \bigg[ \langle \psi_{\ddd-\m+j-1}, (1 -sP(1-sP\Ko_\ddd P)^{-1}(1-\Ko_\ddd))\hat{\mathbf{v}}_k \rangle \bigg]_{j,k=1}^{\m} 
\end{equation}
where $\hat{\mathbf{v}}_k$ is the $k$th component of $\hat{\mathbf{v}}$. 
By arranging the columns backward and using~\eqref{e:BBBB}, we obtain 

\begin{lemma}\label{lem:Dmiddle}
We have 
\begin{equation}\label{eq:algFredi}
	\frac{\fE_\ddd(a_1, \cdots, a_\m; E; s)}{\fE_\ddd( E; s )}
	= \frac{1}{\det[\ga_{\ddd-j}(a_k)]_{j,k=1}^{\m}}   \det \bigg[ \langle \psi_{\ddd-j}, (1 -sP(1-sP\Ko_\ddd P)^{-1}(1-\Ko_\ddd))\hat{\mathbf{v}}_k \rangle \bigg]_{j,k=1}^{\m} .
\end{equation}
\end{lemma}

Now the following lemma shows that the subscript $\ddd$ of $\Ko_\ddd$ in the right-hand side of~\eqref{eq:algFredi} can be replaced by $\ddd-j+1$.

\begin{lemma}\label{lem:Dp}
We have 
\begin{equation}\label{eq:dettemid}
\begin{split}
	&\det \bigg[ \langle \psi_{\ddd-j}, (1 -sP(1-sP\Ko_\ddd P)^{-1}(1-\Ko_\ddd ))\hat{\mathbf{v}}_k \rangle\bigg]_{j,k=1}^{\m} \\
	&= \det \bigg[ \langle \psi_{\ddd-j}, (1 -sP(1-sP\Ko_{\ddd-j+1}P)^{-1}(1-\Ko_{\ddd -j+1}))\hat{\mathbf{v}}_k \rangle \bigg]_{j,k=1}^{\m}.
\end{split}
\end{equation}
\end{lemma}

\begin{proof}[Proof of Lemma~\ref{lem:Dp}]
We first observe the following general identity: for an operator $A$, if $B=A+f\otimes f$ and if $P$ is a projection, then 
\begin{equation}\label{eq:generalmid0}
\begin{split}
	&(1-sPAP)^{-1}(1-A) - (1-sPBP)^{-1}(1-B)\\
	&= (1-sPAP)^{-1}f\otimes f + ((1-sPAP)^{-1}-(1-sPBP)^{-1})(1-B) \\
	&= 	(1-sPAP)^{-1}f\otimes f - s(1-sPAP)^{-1}Pf\otimes fP(1-sPBP)^{-1}(1-B) \\
	&= 	(1-sPAP)^{-1}f\otimes f(1-sP(1-sPBP)^{-1}(1-B)).
\end{split}
\end{equation}
Hence for any square integrable functions $g$ and $h$, 
\begin{multline}\label{eq:generalmid}
	\langle g, (1-sPAP)^{-1}(1-A) h \rangle = \\
        \langle g,  (1-sPBP)^{-1}(1-B) h \rangle
        +\langle g, (1-sPAP)^{-1}f \rangle\langle f, (1-sP(1-sPBP)^{-1}(1-B))h \rangle.
\end{multline}

Also observe that since $\Ko_k= \psi_0\otimes \psi_0+ \cdots +\psi_{k-1}\otimes \psi_{k-1}$, we have 
$\Ko_{\ddd-j+1}= \Ko_{\ddd-j}+\psi_{\ddd-j}\otimes \psi_{\ddd-j}$.

We denote the matrix on each side of the identity~\eqref{eq:dettemid} as $L$ and $R$. 
Consider $R_{ij}$. 
Applying~\eqref{eq:generalmid} to $A=\Ko_{\ddd-j+1}$, $B=\Ko_{d-j+2}$, $g=\psi_{\ddd-j}$ and $h= \hat{\mathbf{v}}_k$, we obtain 
\begin{equation}
\begin{split}
	R_{jk} &= \langle \psi_{\ddd-j}, (1 -sP(1-sP\Ko_{\ddd-j+2}P)^{-1}(1-\Ko_{\ddd-j+2}))\hat{\mathbf{v}}_k \rangle\\
	&\quad + \langle \psi_{\ddd-j}, (1 -sP(1-sP\Ko_{\ddd-j+1}P)^{-1}) \psi_{\ddd-j+1} \rangle R_{j-1,k}.
\end{split}
\end{equation}
If we apply~\eqref{eq:generalmid} again with $A=\Ko_{\ddd-j+2}$, $B=\Ko_{d-j+3}$, $g=\psi_{\ddd-j}$ and $h:= \hat{\mathbf{v}}_k$, then we obtain 
\begin{equation}
\begin{split}
	R_{jk} &= \langle \psi_{\ddd-j}, (1 -sP(1-sP\Ko_{\ddd-j+3}P)^{-1}(1-\Ko_{\ddd-j+3}))\hat{\mathbf{v}}_k \rangle\\
	&\quad + \langle \psi_{\ddd-j}, (1 -sP(1-sP\Ko_{\ddd-j+2}P)^{-1} )\psi_{\ddd-j+2} \rangle R_{j-2,k}\\
	&\quad + \langle \psi_{\ddd-j}, (1 -sP(1-sP\Ko_{\ddd-j+1}P)^{-1} )\psi_{\ddd-j+1} \rangle R_{j-1,k}.
\end{split}
\end{equation}
Repeating this procedure $j$ times, we obtain that $R_{jk}$ equals $L_{jk}$ plus a linear combination of $R_{j-1, k}, \cdots, R_{1, k}$. This implies that the determinant of $R$ equals the determinant of $L$. 
\end{proof}

For the spiked model of dimension $\ddd-j+1$ with the single spiked eigenvalue $a_k$,~\eqref{eq:algFredi} implies that 
\begin{equation} \label{eq:last_eq_in_proof_8.1}
	\bar{\fE}_{\ddd-j+1}(a_k; E; s)  
	=\frac1{\ga_{\ddd-j}(a_{k}) } \langle \psi_{\ddd-j}, (1 -sP(1-sP\Ko_{\ddd-j+1} P)^{-1}(1-\Ko_{\ddd-j+1} ))\hat{\mathbf{v}}_k \rangle.
\end{equation}
Comparing with the right-hand side of~\eqref{eq:dettemid}, we obtain Proposition~\ref{thm:algE}. 

\end{proof}

\subsubsection*{Acknowledgments}
We would like to thank Marco Bertola, Robbie Buckingham, Seung-Yeop Lee and Virgil Pierce 
for keeping us informed of the progress of their work. 
We would also like to thank Mark Adler and John Harnad for helpful communications
and Percy Deift for his insight that was helpful in the proof of Theorem~\ref{thm:critical1}. 
In addition, we are grateful 
to an  anonymous referee  whose comments helped us improve the exposition of this paper greatly.

\subsubsection*{Funding}

This work of J.B.\ was supported in part by NSF grants DMS075709
and DMS1068646.

\def\cydot{\leavevmode\raise.4ex\hbox{.}}


\end{document}